\documentclass{article}

\usepackage[utf8]{inputenc}
\usepackage[T1]{fontenc}
\usepackage[margin=1in]{geometry}

\usepackage{nameref,zref-xr}
\zxrsetup{toltxlabel}

\usepackage{natbib}
\setcitestyle{authoryear,open={[},close={]}}

\usepackage[]{amsmath,amssymb,epsfig}
\usepackage{amsthm}
\usepackage{amssymb}
\usepackage{graphicx}
\usepackage{epstopdf}
\usepackage{comment}
\usepackage{array}
\usepackage{algorithm}
\usepackage{url}
\usepackage{bm}
\usepackage{dsfont}

\usepackage{lscape}
\usepackage{algpseudocode}
\usepackage{setspace}
\usepackage{multicol}
\usepackage{multirow}
\usepackage{color}
\usepackage{colortbl}
\usepackage{xcolor}

\usepackage{hyperref}
\hypersetup{
    colorlinks=true,
    citecolor=blue
    }

\usepackage[capitalise,noabbrev]{cleveref}
\crefname{equation}{}{}
\crefname{enumi}{}{}
\usepackage{mathtools}

\usepackage{placeins}
\usepackage[%
    font={small,sf},
    labelfont=bf,
    format=hang,
    format=plain,
    margin=0pt,
    width=0.8\textwidth,
]{caption}
\usepackage[list=true]{subcaption}

\usepackage{enumerate}
\usepackage{tikz}
\usepackage{tikz-3dplot}
\usepackage{amssymb}
\usepackage{mathtools}

\usepackage{xifthen}
\usepackage{graphicx}
\usepackage{pgfplots}
\pgfplotsset{compat=newest}
\usetikzlibrary{plotmarks}
\usetikzlibrary{arrows.meta}
\usepgfplotslibrary{patchplots}
\usepackage{grffile}
\usepackage{amsmath}
\usepackage{tikz}
\usetikzlibrary{positioning}
\usetikzlibrary{shapes.geometric}

\usepackage{nicefrac}
\usepackage{tabularx}
\usepackage{booktabs}       

\DeclareMathOperator{\Diag}{Diag}

\DeclareMathOperator{\Sp}{Sp}

\DeclareMathOperator{\Tr}{Tr}
\DeclareMathOperator{\sign}{sign}
\DeclareMathOperator{\cov}{cov}
\DeclareMathOperator{\rmi}{\mathrm{i}}

\DeclareMathOperator{\cond}{\mathrm{cond}}
\DeclareMathOperator{\supp}{\mathrm{supp}}

\newtheorem{definition}{Definition}
\newtheorem{theorem}{Theorem}
\newtheorem{lemma}{Lemma}
\newtheorem{remark}{Remark}
\newtheorem{proposition}{Proposition}
\newtheorem{example}{Example}

\newtheorem{fact}{Fact}
\newtheorem{assumption}{Assumption}

\def\cyclepopping{\textsc{CyclePopping}}
\def\HKPV{\textsc{HKPV}}

\newcommand{\calA}{\mathcal{A}}
\newcommand{\calB}{\mathcal{B}}

\newcommand{\rmd}{\mathrm{d}}

\newcommand{\E}{\mathbb{E}}
\newcommand{\I}{\mathbb{I}}
\newcommand{\R}{\mathbb{R}}

\newcommand{\bmA}{\mathbf{A}}

\newcommand{\op}{{\mathrm{op}}}

\newcommand{\intdim}{\mathrm{intdim}}
\newcommand{\lmax}{\lambda_{\mathrm{max}}}
\newcommand{\mumax}{\kappa}

\newcommand{\deff}{d_{\rm eff}}

\newcommand{\calC}{\mathcal{F}}

\newcommand{\calS}{\mathcal{S}}

\newcommand{\calV}{\mathcal{V}}
\newcommand{\calE}{\mathcal{E}}

\newcommand{\DPP}{\mathrm{DPP}}

\newcommand{\bmf}{\bm{f}}
\newcommand{\bmx}{\bm{x}}
\newcommand{\bmy}{\bm{y}}
\newcommand{\bmb}{\bm{b}}
\newcommand{\bmxi}{\bm{\xi}}
\newcommand{\bmdelta}{\bm{\mathrm{e}}}
\newcommand{\bmpsi}{\bm{\psi}}

\newcommand{\pIS}{p_{(\mathrm{IS})}}

\newcommand{\dinf}{d_{\rm inf}}

\newcommand{\lev}{\mathrm{l}}

\newcommand{\pp}{\mathcal{X}}

\newcommand{\Uone}{{\rm U}(1)}
\renewcommand{\Re}{\operatorname{Re}}
\renewcommand{\vec}{\operatorname{vec}}

\newcommand{\Arg}{\mathrm{arg}}

\renewcommand{\top}{\mathsf{\scriptscriptstyle T}}

\providecommand{\keywords}[1]
{
  \small	
  {\textit{Keywords---}} #1
}

\title{Sparsification of the regularized magnetic Laplacian\\
with multi-type spanning forests}
\author{M. Fanuel and R. Bardenet\\
Université de Lille, CNRS, Centrale Lille\\
UMR 9189 - CRIStAL, F-59000 Lille, France\\
\texttt{\{michael.fanuel, remi.bardenet\}@univ-lille.fr} \\}
\date{\today}
\begin{document}
\maketitle

\begin{abstract}
    In this paper, we consider a $\Uone$-connection graph, that is, a graph where each oriented edge is endowed with a unit modulus complex number that is conjugated under orientation flip.
    A natural replacement for the combinatorial Laplacian is then the \emph{magnetic} Laplacian, an Hermitian matrix that includes information about the graph's connection.
    Magnetic Laplacians appear, e.g., in the problem of angular synchronization.
    In the context of large and dense graphs, we study here sparsifiers of the magnetic Laplacian $\Delta$, i.e., spectral approximations based on subgraphs with few edges.
    Our approach relies on sampling multi-type spanning forests (MTSFs) using a custom determinantal point process, a probability distribution over edges that favours diversity.
    In a word, an MTSF is a spanning subgraph whose connected components are either trees or cycle-rooted trees.
    The latter partially capture the angular inconsistencies of the connection graph, and thus provide a way to compress the information contained in the connection.
    Interestingly, when the connection graph has weakly inconsistent cycles, samples from the determinantal point process under consideration can be obtained \emph{à la Wilson}, using a random walk with cycle popping.
    We provide statistical guarantees for a choice of natural estimators of the connection Laplacian, and investigate two practical applications of our sparsifiers{: ranking with angular synchronization and graph-based semi-supervised learning.}
    {From a statistical perspective, a side result of this paper of independent interest is a matrix Chernoff bound with intrinsic dimension, which allows considering the influence of a regularization -- of the form $\Delta + q \mathbb{I}$ with $q>0$ -- on sparsification guarantees.}

\end{abstract}
\tableofcontents

\section{Introduction} 
\label{sec:introduction}
Consider a connected undirected graph of $n$ nodes and $m$ edges.
A $\Uone$-\emph{connection graph} is defined by endowing each oriented edge $uv$  with a complex phase $\phi_{uv} = \exp(-\rmi \vartheta(uv))$ such that $\vartheta(vu) = -\vartheta(uv)\in [0,2\pi)$.
The corresponding $\Uone$-\emph{connection} is the map $e\mapsto \phi_e$, which associates to an oriented edge a complex phase that is conjugated under orientation flip.
The associated \emph{magnetic Laplacian matrix}\footnote{
    The magnetic Laplacian is a special case of the vector bundle Laplacian of~\citet{kenyon2011}, and a close relative to the \emph{connection Laplacian}~\citep{Bandeira}.
} is
\[
    \Delta
    = \sum_{\text{edge } uv}
    w_{uv}
    (\bmdelta_u -\phi_{uv} \bmdelta_v)
    (\bmdelta_u -\phi_{uv} \bmdelta_v)^*,
\]
where $\bmdelta_u$ denotes an $n\times 1$ vector of the canonical basis associated to node $u$ and $w_{uv} \geq 0$ is the weight of edge $uv$.
The case $\phi_\cdot=1$ corresponds to the usual combinatorial Laplacian.
More generally, the smallest \emph{magnetic eigenvalue} $\lambda_{\min}(\Delta)$ is zero iff there exists a real-valued function $h$ such that $\vartheta(uv) = h(u) - h(v)$ for all oriented edges $uv$.
In that case, $\Delta$ is unitarily equivalent to the combinatorial Laplacian;\footnote{
    {This result has probably been well-known for a long time. 
    A few recent references are } \citep[Theorem 2.2]{ChungEtAl} or \citep{Berkolaiko,AFST_2011_6_20_3_599_0,Fanuel2018MagneticEigenmapsVisualization}.
    In particular, such a function $h$ always exists if the underlying graph is a tree.
}
the connection graph is then said to be \emph{consistent}, and the connection to be \emph{trivial}.

{
In this paper, we build a sparse symmetric random matrix $\widetilde{\Delta}_t$ such that, for $q\geq 0$ a regularization parameter, the condition number of 
\begin{equation}
    (\widetilde{\Delta}_t + q\I)^{-1}(\Delta + q \I) \label{eq:typical_precond}
\end{equation}
is small.
This makes linear systems that involve \eqref{eq:typical_precond} more amenable to iterative solvers.
Indeed the condition number controls the stability of the solver, while sparsity can ensure that the inverse or a Cholesky decomposition of $\widetilde{\Delta}_t + q\I$ is cheap to compute.
Our randomized approximation of the magnetic Laplacian relies on an ensemble of independent spanning subgraphs $\calC_1, \dots, \calC_t$, and takes the form
\begin{equation}
    \widetilde{\Delta}_t
    = \frac{1}{t}\sum_{\ell = 1}^{t} \widetilde{\Delta}(\calC_\ell) \text{ with } \widetilde{\Delta}(\calC) = \sum_{\substack{\text{edge } \\ uv \in \calC}}
    \frac{w_{uv}}{\lev(uv)}
    (\bmdelta_u -\phi_{uv} \bmdelta_v)
    (\bmdelta_u -\phi_{uv} \bmdelta_v)^*,\label{eq:def_intro_sparsifier}
\end{equation}
where $\lev(uv)$ is the \emph{leverage score} of edge $uv$, measuring its marginal importance, to be defined in \cref{sec:dpp_sampling_of_edges__multitype_spanning_forests}.
The spanning subgraphs we use are \emph{multi-type spanning forests} \citep[MTSF;][]{kenyon2019}.
MTSFs are disjoint unions of trees --~a connected graph without cycle~-- and cycle-rooted trees --~a connected graph with a unique cycle; see \cref{fig:MTSF}.
The probability distribution over these subgraphs heavily depends on the graph structure, the connection and the parameter $q$.
Intuitively, while spanning trees are known to preserve information on the combinatorial Laplacian, the presence of cycles in MTSFs shall help preserve the additional information about the connection.
}

{
    We now introduce two practical motivations for the sparsification of the magnetic Laplacian: angular synchronization and semi-supervised learning. 
    Given a $\Uone$-connection $\vartheta$, angular synchronization \citep{Singer11} is, loosely speaking, the problem of finding a function $h$ valued in $[0, 2 \pi)$ such that
    \begin{equation}
        \label{e:angular_synchronization}
        \vartheta(uv) \approx h(u) - h(v) \text{ for all } uv\in \mathcal{E}.
    \end{equation}
    An intuition of the connection between angular synchronization and the magnetic Laplacian can be gained as follows.
    Define $f_h(u) = e^{\rmi h(u)}$ for any node $u$, and let $\bmf_h$ be the vector whose $u$-th entry is $f_h(u)$. 
    The value of the quadratic form of the magnetic Laplacian at $\bmf_h$ can be interpreted as a weighted sum of squared errors of the angular synchronization problem \cref{e:angular_synchronization}, in the sense that
    \begin{equation*}
        \bmf^*_h \Delta \bmf_h = 4\sum_{\text{edge } uv} w_{uv}
        \sin^2(\vartheta(uv) -  h(v) +  h(u)).
    \end{equation*} 
    Thus, intuitively, the phase of an eigenvector of $\Delta$ with small eigenvalue corresponds to an approximate solution to \eqref{e:angular_synchronization}.
    Given an approximation of a small eigenvalue, solving for such an eigenvector is a linear problem involving $\Delta$ that can take advantage of preconditionning by a sparsifier.
    A more formal statement clarifying the relation between the ``frustation'' of the connection and the smallest eigenvalue of the (normalized) magnetic Laplacian is given in \citep{Bandeira}, in the form of a Cheeger inequality.
    In the paper, we shall rather be more explicit about the link between the least eigenvector of $\Delta$ and a particular \emph{relaxation} of the synchronization problem.
}

{
Another practical motivation, this time for studying the \emph{regularized} Laplacian $\Delta + q \I$ with $q > 0$, comes from machine learning. 
Graph-based semi-supervised learning consists in estimating function values on a fixed set of nodes, given only a few evaluations at training nodes (``labeled" data).
Namely, let $(y_{\text{train},\ell},u_\ell)_{1\leq \ell \leq l}$ be labeled data, constisting of pairs of ``outputs'' $y_{\text{train}}\in \mathbb{C}$ and ``input'' nodes $u_\ell$ for $1\leq \ell \leq l$. 
This label information is encoded in a vector $\bmy$ as follows: $y(u_\ell) = y_{\text{train},\ell}$ for $1\leq \ell \leq l$ and $y(u) = 0$ otherwise.
A prediction on unlabeled nodes can be obtained by solving
\begin{equation}
    \label{e:ssl}
    \bmf_\star \in \arg\min_{\bmf \in \R^n} \bmf^* \Delta \bmf + q \|\bmf - \bmy\|_2^2,
\end{equation}
possibly followed by a rounding step. 
The spirit of the method is to fit $\bmf$ to the data $\bmy$ while promoting smoothness by penalizing rapidly varying modes, namely high-frequency oscillations.
The latter are associated with eigenvectors of $\Delta$ with large eigenvalues, where \emph{large} is to be understood compared to the scale $q\geq 0$.
In other words, frequencies larger than $q$ are to be suppressed because they are interpreted as noise.
The Laplacian-regularized least squares problem \eqref{e:ssl} is solved by computing the solution of a Laplacian linear system \emph{à la} Tikhonov
\begin{equation}
    (\Delta + q \I)\bmf = q \bmy.\label{eq:Laplacian_q_reg_system}
\end{equation}
From \eqref{eq:Laplacian_q_reg_system}, it is clear that $\bmf_\star$ is obtained by smoothing out Fourier modes of $\bmy$ with frequencies $\lambda$ typically larger than $q$; the filter function being given by $\lambda \mapsto (1+\lambda/q)^{-1}$.
In practice, iterative solvers can leverage an approximation of the matrix on the left-hand side to accelerate the computation, which justifies again the study of sparsifiers.
}

\begin{figure}[t]
    \centering
    \begin{subfigure}[b]{0.24\textwidth}
        \centering
        \includegraphics[scale=0.25]{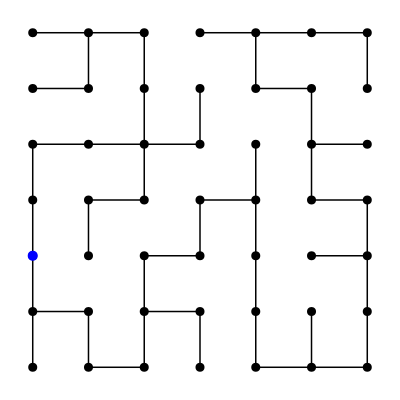}
        \caption{spanning tree\label{fig:ST}}
    \end{subfigure}
    \begin{subfigure}[b]{0.24\textwidth}
        \centering
        \includegraphics[scale=0.25]{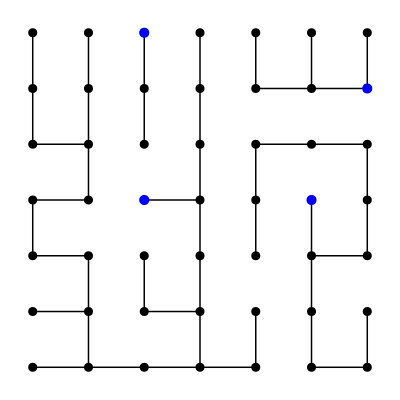}
        \caption{spanning forest (SF)\label{fig:SF}}
    \end{subfigure}
    \begin{subfigure}[b]{0.25\textwidth}
        \centering
        \includegraphics[scale=0.25]{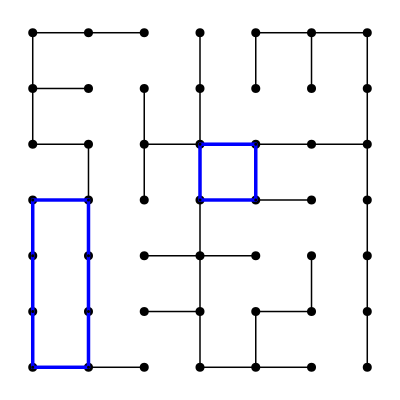}
        \caption{cycle-rooted SF\label{fig:CRSF}}
    \end{subfigure}
    \begin{subfigure}[b]{0.24\textwidth}
        \includegraphics[scale=0.25]{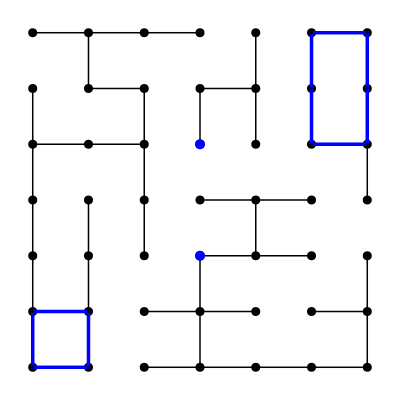}
        \centering
        \caption{multi-type SF\label{fig:MTSF}}
    \end{subfigure}
    \caption{
        Different spanning subgraphs of a grid graph.
        In blue, we show the root nodes used for sampling the subgraphs with the cycle-popping algorithm of \cref{sec:sampling_a_multitype_spanning_forest}, as well as the cycles in which some of the trees are rooted.
        }
    \label{fig:graphs}
\end{figure}
\subsection{Outline of the paper and contributions \label{sec:outline_contributions}}
{To set the scene, Section~\ref{sec:related_work} discusses related work about Laplacian sparsification and random spanning subgraphs.
Next, in Section~\ref{sec:magnetic_laplacian}, we give more details about the two motivating applications of this paper: angular synchronization and graph-based semi-supervised learning.}
Section~\ref{sec:dpp_sampling_of_edges__multitype_spanning_forests} reviews probability distributions over MTSFs, which build on recent generalizations of uniform spanning trees.
Like the set of edges of a uniform spanning tree, one of the most natural distributions over MTSFs associated with $\Delta + q \I$ is a determinantal point process (DPP).
DPPs are distributions, here over subsets of edges, that favour diversity and originate in quantum optics \citep{Mac75,HKPV06}.

\paragraph*{Statistical guarantees.} Our contributions start with Section \ref{sec:sparsification_of_the_regularized_magnetic_laplacian}, which contains our statistical guarantees for the sparsification of {$\Delta + q \I$ for $q\geq 0$} based on the DPP introduced in Section~\ref{sec:dpp_sampling_of_edges__multitype_spanning_forests}.
Our bounds depend on the expected number of edges $\deff = \Tr(\Delta(\Delta + q \I)^{-1})$ of the MTSF, and on $\mumax = \|\Delta(\Delta + q \I)^{-1}\|_{\op}$.
 {Let us fix beforehand a tolerance $\epsilon \in (0,1)$.}
Informally, we give a sparsifier $\widetilde{\Delta}_t {+ q \I}$ built with a batch of $t$ independent MTSFs, such that if $\deff/ \mumax$ is large enough and
\begin{equation}
t \gtrsim \frac{\mumax}{\epsilon^{2}}\log\left(\frac{\deff}{\mumax\delta}\right),\label{eq:lower_bound_on_batch}
\end{equation}
then, with probability at least $1-\delta$,
\begin{equation}
(1-\epsilon)(\Delta + q \I)
\preceq {\widetilde{\Delta}_t + q \I}
\preceq (1 + \epsilon)(\Delta + q \I);\label{eq:multiplicativeBoundNoq}
\end{equation}
see \cref{thm:BatchMTSFs}.
{
Note that we were inspired by \citet*{KKS22}, who derived the same result for the combinatorial Laplacian in the particular case $q=0$.
Not suprisingly, in their result as in ours, roughly $\log(n)/\epsilon^2$ independent subgraphs are sufficient to have a $(1\pm \epsilon)$-sparsifier of the Laplacian.
}

{ The \emph{multiplicative} guarantee \eqref{eq:multiplicativeBoundNoq} is well-suited for analysing the quality of a preconditionner of the form \cref{eq:typical_precond} for linear systems such as \cref{eq:Laplacian_q_reg_system}.
One interesting and unusual point is the effect of the regularization parameter $q$ on the approximation on the \emph{unregularized} Laplacian. Loosely speaking, there is a trade-off between the number of edges in the sparsifier and the accuracy of the approximation.}

{
\begin{itemize}
    \item \textbf{Lower bound on batch size decreases with $q$.} On the one hand, $\kappa$, $\deff$ as well as $\deff/ \mumax$ are decreasing functions of $q$; see \cref{prop:decreasing_deff} in Appendix for more details.
    Thus, the lower bound \eqref{eq:lower_bound_on_batch} on the batch size is also a decreasing function of $q$.
    In other words, the condition \eqref{eq:lower_bound_on_batch} on the number of independent MTSFs becomes less strict as $q$ increases.
    \item \textbf{The statistical guarantee is loose up to scale $q$.} On the other hand, the inequality \cref{eq:multiplicativeBoundNoq} is in fact comparing ellipsoids associated to quadratic forms. To understand the role of $q$ in this picture, consider individually the orthonormal eigenpairs of $\Delta$, namely, $(\mathbf{v}_\ell, \lambda_\ell)_\ell$ for which $\lambda_\ell \neq 0$. Thus, by definition of the Loewner partial order on positive semidefinite matrices, the inequality \cref{eq:multiplicativeBoundNoq} yields the following bound on the relative error on the ``approximation of the $\ell$-th eigenpair''
    \begin{equation}
        \left| \frac{\mathbf{v}_\ell^* \widetilde{\Delta}_t\mathbf{v}_\ell - \lambda_\ell}{\lambda_\ell} \right|\leq \epsilon (1 +q/\lambda_\ell).\label{eq:relative_abs_err_eigenpair}
    \end{equation}
    The guarantee on the approximation of the $\ell$-th eigenpair is looser when $\lambda_\ell \ll q $, whereas for $\lambda_\ell \gg q$ the effect of $q$ is almost negligible and the right-hand side of \eqref{eq:relative_abs_err_eigenpair} is close to $\epsilon$.
\end{itemize}
Therefore, we conclude that, for $q$ large w.r.t.\ the least eigenvalues of $\Delta$, the sparsifier does not necessarily approximate well the projection of $\Delta$ onto the eigenspaces corresponding to the bottom of its spectrum, i.e., what intuitively corresponds to the large-scale structures of the graph. 
This is consistent with the fact that the expected number of edges in the sparsifier is a decreasing function of $q$.
In particular, a $(1\pm \epsilon)$-approximation of $\Delta+q\I$ with $q>0$ can allow fewer edges than an $(1\pm \epsilon)$-approximation of $\Delta$.}

{A key side result of this paper is} a matrix Chernoff bound with an intrinsic matrix dimension for determinantal point processes, in the spirit of~\citet*{KKS22}; {see \cref{thm:Chernoff_StronglyRayleigh}. 
A corollary of this tail bound is the statistical guarantee \cref{eq:multiplicativeBoundNoq} for the sparsifier approximation which is stated formally in \cref{thm:BatchMTSFs}.} 

\paragraph*{Sampling is fast.}  {We advocate the use of a specific MTSF sampling algorithm -- called $\cyclepopping{}$ -- based on a cycle-popping random walk \citep{kassel2017,guo_jerrum_2021}.
This algorithms works under the hypothesis of the \emph{weak inconsistency} of the graph cycles.
While we do not claim any novelty for the algorithm, which was proposed by \citet{kassel2017} and proved correct in more details by \citet{Kassel15}, our contribution is to demonstrate this dedicated algorithm as a viable candidate to implement our sparsifiers. 
In particular, $\cyclepopping{}$ is fast compared to alternatives such as the QR-inspired algorithm of \cite*{HKPV06}, called HKPV algorithm, which is a generic DPP sampler. 
This algebraic algorithm has a complexity $\mathcal{O}(mn^2)$ in the absence of any particular structure and can be slow in practice in the case of large graphs.
Our numerical experiments in \cref{sec:HKPV_extra_results} on graphs of relatively small scale confirm these remarks.
To give an order of magnitude, an execution of $\cyclepopping{}$ for a graph with about $10^5$ nodes and $10^6$ edges takes about $3$ seconds on a laptop as reported in \cref{fig:cond_number_Stanford_q_0}.}

\paragraph*{Self-normalized importance sampling.}  
As previously mentioned, $\cyclepopping{}$ requires that cycles are weakly inconsistent. 
To circumvent this limitation, in \cref{sec:montecarlo_with_selfnormalized_importance_sampling}, we further investigate a self-normalized Monte-Carlo algorithm to deal with larger inconsistencies and provide a central limit theorem for the sparsifiers obtained by this technique.

\paragraph*{Two use cases.}  
Finally, in \cref{sec:empirical_results}, we illustrate our results on two case studies: ranking from pairwise comparisons using Sync-Rank~\citep{Cucuringu16} and preconditioning regularized Laplacian systems.
Deferred mathematical results and proofs are given in \cref{sec:deferred_proofs} and  additional empirical results can be found in \cref{sec:additional_numerical_simulations}.

\begin{table}[b]
    \begin{center}
        \begin{tabular}{lll}
            \toprule
            Abbreviation & Full name & Definition\\
            \midrule
            ST & spanning tree & a connected spanning {subgraph} without cycle.\\
            SF & spanning forest & a spanning {subgraph} given by a disjoint union of trees.  \\
            CRT & cycle-rooted tree & a connected subgraph with only one cycle. \\
            CRSF & cycle-rooted spanning forest & a disjoint union of cycle-rooted trees.\\
            MTSF & multi-type spanning forest & a disjoint union of trees and cycle-rooted trees.\\
            \bottomrule
        \end{tabular}
    \end{center}
    \caption{List of abbreviations for subgraphs with their definitions.\label{tab:spanning_graphs}}
\end{table}
\subsection{Limitations \label{sec:limitations}}
{
    For the benefit of the quick reader, we immediatedly point out to some limitations of the paper. 
}
{
    First, our statistical guarantees rely on the assumption that \emph{leverage scores} $\lev(uv)$ are known \emph{exactly}, in order to construct sparsifiers with formula \cref{eq:def_intro_sparsifier}.
    A control on the effect of approximating the leverage scores is missing in our work, as it is the case for similar sparsification methods in the literature to our knowledge. 
    To circumvent this question for the moment, we investigate both a heuristic that replaces leverage scores by $\lev_{\rm unif}(e) = |\calC| / m$, and the classical Johnson-Lindenstrauss (JL) lemma  for sketching leverage scores.
    Second, the $\cyclepopping{}$ algorithm provably works only under the condition that cycles in the connection graphs are weakly inconsistent; see below for a definition. 
    To avoid this limitation, we propose a self-normalized importance sampling method for which we derive a central limit theorem, so that the statistical guarantee is \emph{weaker} than \cref{eq:multiplicativeBoundNoq} in that case.
    We leave the extension of $\cyclepopping{}$ to strongly inconstistent cycles for future work.
    Third, we do not discuss how to choose the parameter $q$ in practice.
    This question is related to the estimation of hyperparameters in semi-supervised learning. 
    In line with this comment, we do not discuss how we can sample sparsifiers for several values of $q$ in a efficient way; a possible approach would be to use the coupling algorithm by \citet{AvGaud2018}.
    Although our companion Julia code provides comparisons of the performance of the Laplacian approximation with different baselines, our simulations are only a proof-of-concept and leave space for runtime improvements.
}

\subsection{Notations}
Sets are denoted with a calligraphic font, e.g., $\mathcal{E}$, $\mathcal{V}$, except for $[n] = \{1, \dots, n\}$.
Matrices are denoted by uppercase letters ($\Delta,B,Y$) whereas vectors are typeset as bold lowercase ($\bmx,\bmy, \bmdelta$).
The canonical basis of $\R^n$ is denoted by $(\bmdelta_u)_{u\in [n]}$, and $\bm{1}$ is the all-ones vector. For convenience, we denote the $u$-th entry of $\bmf$ by $f(u)$.
The Hermitian conjugate of complex matrix $A$ is written $A^*$, and its transpose is $A^\top$.
The Moore-Penrose pseudo-inverse of a matrix $A$ is denoted by $A^+$.
Also, $|A| = (A^* A)^{1/2}$ denotes the matrix absolute value of the real symmetric matrix $A$.
For $A,B$ real symmetric matrices,
we write $A \preceq B$ if $\bm{x}^* A \bm{x} \leq \bm{x}^* B \bm{x}$ for all $\bm{x}\in \mathbb{C}^n$.
We denote by $\cond(A)$ the condition number of a matrix $A$, namely, the ratio of its largest to smallest singular value.
The spectrum of a symmetric matrix $A$ is denoted by $\Sp(A)$. The operator norm of $A$ is $\|A\|_{\op} = \max_{\|\bmx\|_2 = 1}\|A\bmx\|_2$.
In what follows, $\mathcal{G}$ is a connected undirected graph with vertex set $\mathcal{V}$ and edge set $\mathcal{E}$; for simplicity, denote $n=|\mathcal{V}|$ and $m = |\mathcal{E}|$.
{Edges are assumed to be simple, i.e., there is no multiple edge.} 
Each $e\in \mathcal{E}$ comes with a weight $w_e > 0$.
{Unless stated otherwise, $\mathcal{G}$ has no self-loop.}
We denote an \emph{oriented} edge as $e=vv'$ where $v\in \mathcal{V}$ is $e$'s head and $v'\in \mathcal{V}$ is $e$'s tail.
Note that, by a slight abuse of notation, to denote an edge in $\mathcal{E}$ as a pair of nodes, an orientation has to be chosen, and thus $vv'\in \mathcal{E}$ and $v'v \in \mathcal{E}$ denote the same edge.
At this level, the orientation is purely arbitrary, that is, a matter of convention.
Let further $1_{\calS}(\cdot)$ be the indicator function of the set $\calS$.
For conciseness, we finally denote $a\vee b = \max\{a,b\}$ and $a \wedge b = \min\{a,b\}$.

\section{Related work} 
\label{sec:related_work}

A key question in graph theory concerns approximating the combinatorial Laplacian of a graph by the Laplacian of one of its subgraphs with few edges.
The resulting Laplacian is called a \emph{sparsifier} of the original Laplacian \citep*{SparsificationSurvey}.
\emph{Spectral} sparsifiers aim to preserve properties of the spectrum of the combinatorial Laplacian, whereas \emph{cut-sparsifiers} intend to preserve the connectivity of the graph, namely a cut, see, e.g., \citet*{Fung11}.
Spectral sparsifiers of the graph Laplacian were defined by~\citet*{SpielmanTeng}.
\cite*{BatsonSpielmanSrivastava} show the existence of an $(1\pm \epsilon)$-accurate \emph{spectral sparsifier} with $\mathcal{O}(n/\epsilon^2)$ edges and give a polynomial time deterministic algorithm; see also~\citet*{Zou12}.
These algorithms require {a running time of $\Omega(n^4)$}.

Direct improvements of linear-sized spectral sparsifiers using \emph{randomization} are given by~\cite*{AZZO15,LeSu18} with $\mathcal{O}(n^{2+\epsilon})$ and $\Omega(n^2)$ time complexity, respectively.
The seminal work of~\cite*{SpielmanSrivastava} showed that, with high probability, a $(1\pm \epsilon)$-accurate \emph{spectral sparsifier} of the combinatorial Laplacian is obtained by i.i.d.\  sampling of $\mathcal{O}(n \log(n)/\epsilon^2)$ edges {with respect to a distribution weighted by leverage scores.
Similar guarantees are derived by \citet{cohen2015uniform} for an edge-by-edge, leverage-score-based, sampling approach.
Leverage score estimation can be performed using the Johnson-Lindenstrauss lemma, which only requires solving $\mathcal{O}(\log n)$ linear systems in the Laplacian; see \citep{SpielmanSrivastava,cohen2015uniform}.
The solver of \citep{KyngSparseCholeskyConnection} can be used to solve these linear systems
in nearly-linear time.
}
To our knowledge, randomized spectral sparsifiers have shown so far superior performance in practice over their deterministic counterparts.

In the perspective of sampling as few edges as possible, \emph{joint} edge samples have also been used to yield sparsifiers. An $\epsilon$-accurate cut sparsifier is obtained with high probability by a union of $\Theta\left(\epsilon^{-2}\log^2(n)\right)$ uniform random spanning trees~\citep*{Fung11}.
Similarly,
\citet*{KyngSong} show that averaging the Laplacians of $\mathcal{O}\left(\epsilon^{-2}\log^2(n)\right)$ uniform random spanning trees yields a $(1\pm \epsilon)$ spectral sparsifier with high probability.
In a recent work, \cite*{KKS22} improve this result to $\mathcal{O}\left(\epsilon^{-2}\log(n)\right)$ random spanning trees.
Also, a short-cycle decomposition approach for the Laplacian sparsification is given by~\citet*{CPSSW18}.

As we just mentioned, spectral sparsification can be performed by sampling spanning trees.
This problem has been tackled using random walks \citep*{Aldous89,Broder89,Wilson96}.
More recently, \citet*{Schild18} introduced Laplacian solvers within random walks, to improve the shortcutting methods of~\citet*{MST15,KeMa09}.

The determinantal point {process} associated to the edges of a uniform Spanning Tree (ST) is well-known to be associated to the combinatorial Laplacian~\citep*{Pemantle91}.
STs are closely related to random Spanning Forests (SFs), which are in turn associated to the regularized combinatorial Laplacian; see~\cite{AvGaud2018} for a detailed study of random-walk-based sampling algorithms.
Recently, Monte Carlo approaches for approximating the regularized combinatorial Laplacian using random forests were studied by \citet*{pilavci2020,PABT2020,PABGAA2020}, where variance reduction methods and statistical analyses are also given in the context of graph signal processing.
In contrast with our paper, \citet{pilavci2020,PABT2020} consider estimators of the inverse regularized combinatorial Laplaciann using the random process associated with the roots of random forests.

The magnetic Laplacian appears in several works in physics, see, e.g., \citet*{zbMATH01219775,AFST_2011_6_20_3_599_0,Berkolaiko}.
It was also used for visualizing directed networks in~\cite{Fanuel2018MagneticEigenmapsVisualization}.
As mentioned above, it is a special case of the connection Laplacian~\citep*{Bandeira}.
Sparsifying a connection Laplacian with leverage score sampling was studied by~\cite{ChungEtAl} in the more general setting of an $\mathrm{O}(d)$-connection, but restricting to the case where the group elements are consistent over any graph cycle.
{From a different perspective}, a sparsified Cholesky solver for connection Laplacians was also proposed by~\citet*{KyngSparseCholeskyConnection}.

{Finally, we briefly mention generic algorithms for sampling from determinantal measures. 
As mentioned in \cref{sec:dpp_sampling_of_edges__multitype_spanning_forests}, any DPP with a symmetric correlation kernel can be sampled using a mixture argument and a generic linear algebraic chain rule; see~\citet[Alg.\ 18]{HKPV06} for the original argument, as well as the discussions around \citet[Alg.\ 1]{LMR2015} and \citet[Alg.\ 1]{KT12}).
However, this generic approach requires an eigendecomposition of the correlation kernel.
It is possible to circumvent this eigendecomposition with the methods of \citet{launay_galerne_desolneux_2020, poulson2020high} but the overall complexity remains similar.
Monte Carlo Markov Chain approaches for sampling DPPs were also proposed by \cite{AnariLGVV21}; see also \citep{hermon2023modified} for a discussion of the mixing time of the corresponding Markov chain.
The aforementioned methods do not use any particular structure, whereas the random walk-based approaches that we use in this paper heavily rely on the graph structure \citep{Wilson96,kassel2017}.
}

\section{The magnetic Laplacian and two motivations} 
\label{sec:magnetic_laplacian}

\paragraph*{Combinatorial Laplacian.}
Let $\mathcal{G}$ be a connected undirected graph, with vertex set $\mathcal{V}$ and edge set $\mathcal{E}$, of respective cardinality $n=|\mathcal{V}|$ and $m = |\mathcal{E}|$.
In what follows, we often abuse these notations and denote the vertex set by $[n]$.
It is customary to define the \emph{oriented} edge-vertex incidence matrix $B_0\in\mathbb{R}^{m\times n}$ as
\[
    B_0(e,v)
    = \begin{cases}
        1  & \text{ if $v$ is $e$'s head }, \\
        -1 & \text{ if $v$ is $e$'s tail,} \\
        0  & \text{ otherwise,}
    \end{cases}
\]
and the combinatorial Laplacian as
\begin{equation}
    \Lambda = B^\top_0 W B_0 \in\mathbb{R}^{n\times n},
    \label{e:combinatorial_laplacian}
\end{equation}
where $W = \Diag(w) \in\mathbb{R}^{m\times m}$ is a diagonal matrix containing non-negative edge weights.
Obviously, different choices of orientation of the edges yield different $B_0$ matrices, but they all produce the same Laplacian matrix.
The matrix $B_0$ is a representation of a discrete derivative for a given ordering of the nodes and an ordering of the edges with an arbitrary orientation.
More precisely, for an oriented edge $e = uv$, the discrete derivative of a function $f$ on $\mathcal{V}$ in the direction $e$ reads
$\rmd_0 f(e)
\triangleq f(u) - f(v)
= (\bmdelta_u^\top - \bmdelta_v^\top) \bmf
= \sum_{w\in[n]}B_0(e,w)f(w)$.

\paragraph*{Magnetic Laplacian.}
We now twist this construction by following the definition of~\citet{kenyon2011}, that we particularize to the group of unit-modulus complex numbers $\Uone$.
Let $\phi_{ve} = \phi_{ev}^{-1}\in \Uone$ and define $\phi_{vv'} = \phi_{ev'}\phi_{ve}$.
Consider the twisted \emph{oriented} and \emph{complex} incidence matrix $B\in\mathbb{C}^{m\times n}$ given by
\[
    B(e,v)
    = \begin{cases}
        \phi_{ve}  & \text{ if $v$ is $e$'s head,} \\
        -\phi_{ve} & \text{ if $v$ is $e$'s tail,} \\
        0          & \text{ otherwise.}
    \end{cases}
\]
The matrix $B$ corresponds to a representation of the twisted discrete derivative.\footnote{In the words of \cite{kenyon2011}, we consider a $\Uone$-connection 
such that $\phi_{vv'} = \exp(-\rmi \vartheta(vv'))$, and $\phi_{ev'} = \exp(-\rmi \vartheta(vv')/2)$ whereas $\phi_{ev} = \exp(\rmi \vartheta(vv')/2)$ for $e = vv'$.
Our definition of the (twisted) discrete derivative differs from \citet{kenyon2011} by a minus sign.}
For any oriented edge $e = uv$, the twisted derivative of $f$ in the direction $e$ is  given by $\rmd f(e) \triangleq \phi_{ue}f(u) - \phi_{ve}f(v)= (\phi_{eu}\bmdelta_u - \phi_{ev}\bmdelta_v)^* \bmf = \sum_{w\in[n]}B(e,w)f(w)$.
Similarly to the combinatorial Laplacian \eqref{e:combinatorial_laplacian}, the magnetic Laplacian is defined as
\begin{equation}
    \Delta = B^* W B \in\mathbb{C}^{n\times n},
\label{e:magnetic_laplacian}
\end{equation}
with, again, a diagonal edge weight matrix $W\in\mathbb{R}^{m\times m}$ with non-negative entries.
Thus, by construction, the spectrum of $\Delta$ is non-negative.
As anticipated in \cref{sec:introduction} and in the case of a connected graph, the least eigenvalue of $\Delta$ vanishes iff the corresponding connection graph is trivial. In contrast, $\Lambda$ has always an eigenvalue equal to zero.


\subsection{The spectral approach to angular synchronization} 
\label{sub:Eigenvalue problem and angular synchronization}
As explained below, this problem is related to the magnetic Laplacian that we sparsify in this paper.
Applications range from synchronizing distributed networks~\citep{CLS12} to image reconstruction from pairwise differences~\citep{Yu2009} solving jigsaw puzzles~\citep{huroyan2020solving}, aligning respiratory signals~\citep{mcerlean2024unsupervised}, or robust ranking from pairwise comparisons~\citep{Cucuringu16}; see \cref{sub:magnetic_laplacian_sparsification_and_ranking} for details on the latter.
Although we focus on $\mathrm{U}(1)$ in this paper, different groups {have} also been considered, e.g. synchronization using $\mathrm{SO}(3)$ for identifying the 3D structure of molecules \citep{CSC12}.

Following the seminal paper by \cite{Singer11}, the rationale of the spectral approach to angular synchronization goes as follows.
Synchronization \eqref{e:angular_synchronization} can be achieved by solving
\begin{equation}
    \max_{\mathbf{h} \in [0,2\pi)^n}
    \sum_{uv\in \calE}\cos\Big(\vartheta(uv) + h(v) -h(u)\Big). \label{eq:synchronization}
\end{equation}
In particular, the cosine loss function is approximately quadratic for small argument values, and promotes robustness as it saturates for large argument values.
Several relaxations of the above problem are obtained by noting that $\cos(t) = \Re\left(e^{\rmi t}\right)$.
Consider the following spectral relaxation
\begin{equation}
    \max_{\bm{f}\in \mathbb{C}^n: \|\bm{f}\|^2 = n}
    \sum_{uv\in \calE} w_{uv}
    \Re\left(f(u)^*e^{\rmi \vartheta(uv)}f(v)\right),\label{eq:spectral_syncrank_normalized}
\end{equation}
where we have further introduced weights such that $w_{uv}\geq 0$.
Once a solution $\bmf$ to \eqref{eq:spectral_syncrank_normalized} has been found, the angle estimates $\Arg(f(u))$ are set to the argument of each entry of $\bmf$, with the convention that $\Arg(0) = 0$.

We now relate this problem to an eigenvalue problem for the magnetic Laplacian.
Indeed, consider the spectral problem
\begin{equation}
    \min_{\|\bm{f}\|^2 = n}
    \sum_{uv\in \calE} w_{uv}
    \left|f(u) - e^{\rmi \vartheta(uv)} f(v)\right|^2,\label{eq:spectral_syncrank}
\end{equation}
where the objective is proportional to the quadratic form defined by the magnetic Laplacian \eqref{e:magnetic_laplacian}.
The objectives of \cref{eq:spectral_syncrank_normalized} and \cref{eq:spectral_syncrank} relate as
\[
    \sum_{uv\in \calE} w_{uv}
    \left|f(u) - e^{\rmi \vartheta(uv)} f(v)\right|^2 = -2\sum_{uv\in \calE} w_{uv}
    \Re\left(f(u)^*e^{\rmi \vartheta(uv)}f(v)\right) + \sum_{uv\in \calE} w_{uv}
    \left(|f(u)|^2 +  |f(v)|^2 \right).
\]
Since the last term above is constant when all the components of $\bmf$ have the same modulus, namely if $f(u) = \mathrm{cst} \cdot e^{\rmi h(u)}$ for all $u\in [n]$, \cref{eq:spectral_syncrank} is indeed another relaxation of \cref{eq:synchronization}.

\subsubsection{Least eigenvector approximation: sparsify-and-eigensolve}
Solving \cref{eq:spectral_syncrank} amounts to finding the least eigenspace of the magnetic Laplacian, which we assume to have dimension one in this paper.
We now explain how a spectral sparsifier \cref{eq:multiplicativeBoundNoq} can be used to approximate this eigenvector with the smallest eigenvalue. This approach may be called sparsify-and-eigensolve, in the spirit of sketch-and-solve methods; see \citet[Section 10.3]{MT20}.

Assume w.l.o.g.\ that the respective least eigenvectors of $\Delta$ and $\tilde\Delta$ satisfy $\|\bmf_1\|_2 = \|\tilde{\bmf}_1\|_2 = \sqrt{n}$ in light of the spectral relaxation~\cref{eq:spectral_syncrank_normalized}.
Let $\delta_\star = \lambda_2(\Delta) - \lambda_1(\Delta)$ be the spectral gap in the Laplacian spectrum, where $\lambda_1(\Delta) \leq \dots \leq \lambda_n(\Delta)$ are the eigenvalues of $\Delta$.
If $\|\Delta - \widetilde{\Delta}\|_{\op}< \delta_\star$, then the Davis-Kahan theorem guarantees that
\[
    \text{dist}(\bmf_1 ,\tilde{\bmf}_1)
    = \min_{\theta\in [0,2\pi)}
    \| \bmf_1 -\tilde{\bmf}_1e^{\rmi \theta}\|_2
    \leq
    \frac
    {\sqrt{2}\|(\Delta - \widetilde{\Delta}) \bmf_1\|_{2}}
    {\delta_\star - \|\Delta - \widetilde{\Delta}\|_{\op}}
    \leq
    \frac
    {\sqrt{2}\|\Delta - \widetilde{\Delta}\|_{\op}}
    {\delta_\star - \|\Delta - \widetilde{\Delta}\|_{\op}}
    \sqrt{n},
\]
see~\citet[Lemma 11]{ZhBou18}; see also \citet{YWS15,DaKa70} for more general statements.
Now by definition \eqref{eq:multiplicativeBoundNoq}, $\|\Delta - \widetilde{\Delta}\|_{\op}\leq \epsilon \|\Delta\|_{\op}$.
Thus, if $\epsilon \|\Delta\|_{\op}< \delta_\star$, we obtain
\begin{equation}
    \text{dist}(\bmf_1 ,\tilde{\bmf}_1)
    \leq
    \frac
    {\epsilon \lambda_{n}(\Delta)}
    {\delta_\star - \epsilon \lambda_{n}(\Delta)}
    \sqrt{2n},\label{eq:guarantee_eigenvector}
\end{equation}
so that computing the least eigenvector of $\tilde{\Delta}$ is a controlled replacement for that of $\Delta$.
Note that by analogy with the combinatorial Laplacian, we conjecture that a large spectral gap $\delta_\star$ in the magnetic Laplacian spectrum presumably corresponds to a graph without clear bottleneck, i.e., with a good connectivity.

Although appealing, we have however observed in practice that sparsify-and-eigensolve is often not very accurate; see the numerical simulations in \cref{sec:exp_SyncRank}.
In contrast, we now provide another use of sparsifiers for improving the numerical convergence of eigensolvers.

\subsubsection{Preconditioning eigenvalue problem: sparsify-and-precondition} \label{subsubsec:Preconditioning_eigenvalue_problem}
Still in the context of angular synchronization, the spectral sparsifier can also be used to build a preconditioner for computing the least eigenpair $(\lambda_1,\bmf_1)$ of $\Delta$.
We dub this approach \emph{sparsify-and-precondition}, by analogy with the \emph{sketch-and-precondition} paradigm in randomized linear algebra \citep[Section 10.5]{MT20}.
Sparsify-and-precondition assumes the knowledge of an approximation $\hat{\lambda}_1$ of $\lambda_1$.
The least eigenvector is then approximated by the solution of the linear system $(\Delta-\hat{\lambda}_1\I_n) \bmf = 0$; see \citet{AKNOZ17} for more details.
Since the convergence rate of iterative linear solvers depends on the condition number of the matrix defining the system, a common technique is to rather solve $T^{-1}(\Delta-\hat{\lambda}_1\I_n) \bmf = 0$, when $T$ is a well-chosen positive definite matrix, called a \emph{preconditioner}, such that $\cond\left(T^{-1}\Delta\right)$ is smaller than $\cond(\Delta)$ \citep{KN03}.

In this paper, we naturally choose $T = \widetilde{\Delta}$ for which derive the multiplicative guarantee \cref{eq:multiplicativeBoundNoq} with high probability.
Incidentally, a preconditioner for another eigenvalue solver, called the Lanczos method, was also discussed earlier by~\citet{MS93, Morgan2000}.

\subsection{Preconditioned magnetic Laplacian systems for semi-supervised learning} 
\label{sub:regularized_laplacian_systems_preconditioning_and_sparse_cholesky_factorization}
Outside angular synchronization, the approximation $\widetilde{\Delta}$ can also serve as a preconditioner for more general Laplacian systems; see \citet[Section 17]{Vishnoi2013} and references therein.
In the context of semi-supervised learning~\citep{ZBLWS13}, optimization problems of the following form
\begin{equation}
    \label{e:semi-supervised_learning}
    \bmf_\star \in \arg\min_{\bmf \in \R^n} \bmf^* \Delta \bmf + r(\bmf, \bmy),
\end{equation}
are often considered when only a few outputs are given for training, namely, $y_{\text{train},\ell}$ are associated to a subset of nodes $u_\ell$ for $1\leq \ell \leq l$. Here $r(\bmf, \bmy)$ is a risk between $\bmf$ and the vector $\bmy$ which stores label information as follows: $y(u_\ell) = y_{\text{train},\ell}$ for $1\leq \ell \leq l$ and $y(u) = 0$ otherwise.
Solving one of these problems yields a prediction on the unlabeled nodes, possibly after a rounding step such as $\Arg(\bmf_\star)$.

If we choose $r(\bmf,\bmy) = q \| \bmf - \bmy \|^2$ with a parameter $q > 0$, the first order optimality condition is the \emph{regularized} Laplacian system $(\Delta + q \I_n)\bmf = q \bmy$.
If $\epsilon\in (0,1)$ and the following guarantee holds
\begin{align}
    (1-\epsilon)(\Delta + q \I_n)
    \preceq \widetilde{\Delta} + q \I_n
    \preceq (1+\epsilon) (\Delta + q \I_n),
    \label{eq:multiplicativeBound}
\end{align}
then
$
    \Sp((
        \widetilde{\Delta} + q \I_n
    )^{-1}
    \left(
        \Delta + q \I_{n}
        \right))
$
is in the interval $[(1+\epsilon)^{-1}, (1-\epsilon)^{-1}]$.
Finding a sparsifier $\widetilde{\Delta}$ satisfying \eqref{eq:multiplicativeBound} thus guarantees a well-conditioned linear system.

In the context of angular synchronization, \citet[Section 7]{Cucuringu16} also discusses a semi-supervised approach for ranking from pairwise comparisons with anchors, i.e., with a few nodes for which the ranking is given a priori. Thus, the problem \cref{e:semi-supervised_learning} including the magnetic Laplacian can also be considered for ranking with anchors.


Moreover, from the computational perspective, solving the preconditioned linear system might also be fast if $\widetilde{\Delta}$ is sparse.
In \cref{sec:cholesky_decomposition_of_the_magnetic_laplacian_of_a_mtsf}, we discuss an ordering of the nodes so that the Cholesky factorization $R^\top R = \widetilde{\Delta} + q \I_n$ is very sparse when the sparsifier is associated with one multi-type spanning forest.
Since the triangular matrix $R$ is sparse, the system $R \bmf = \bmb$ can be solved quickly.
In particular, in the case $q=0$ and for a non-trivial connection, a sparsifier $\widetilde{\Delta}$ can be obtained by sampling one CRSF.
Denote by $\calV_c$ the set of nodes in the cycles of this CRSF.
The Cholesky factorization of $\widetilde{\Delta}$ yields a triangular matrix $R$ with $\mathcal{O}(n + |\calV_c|)$  non-zero off-diagonal entries, and is obtained in $\mathcal{O}(n + |\calV_c|)$ operations; see \cref{prop:CholeskyMTSF}.
Thus, the linear system $R \bmf = \bmb$ can be solved in $\mathcal{O}(n +|\calV_c|)$ operations.

Finally, we note that another possible choice for the risk in \eqref{e:semi-supervised_learning} is $r(\bmf,\bmy) = -2\Re(\bmf^*\bmy)$, whose first-order optimality conditions yield an unregularized Laplacian system as first order optimality conditions.
Equality constraints can also be imposed by taking the risk as the characteristic function of the constraint set.

\section{Multi-Type Spanning Forests and determinantal point processes} 
\label{sec:dpp_sampling_of_edges__multitype_spanning_forests}

We now describe the distribution we sample from to build our sparsifiers satisfying \cref{eq:multiplicativeBoundNoq} with high probability.
To simplify the expressions, we assume hereinafter that all edge weights are equal to $1$, i.e., $W = \I_m$.

\subsection{A DPP that favors inconsistent cycles}
For each cycle-rooted tree in an MTSF, with cycle $\eta$, define the \emph{holonomy} of $\eta$ by
\begin{equation}
    \label{e:holonomy}
    \mathrm{hol}(\eta) = \prod_{vv'\in \eta}\phi_{vv'}{\triangleq e^{\rmi \theta(\eta)}},
\end{equation}
{where $\theta(\eta) = \sum_{vv'\in \eta} \vartheta(vv')$ is the corresponding angle.} 
We consider the distribution over MTSFs $\calC$ given by
\begin{equation}
    \label{eq:proba_MTSF}
    p(\calC)
    = \frac{q^{|\rho(\calC)|}}{\det (\Delta + q\I_n)}
    \prod_{\eta \in \text{cycles}(\calC)}
    \Big(2 - 2\cos \theta(\eta)\Big),
\end{equation} where
$|\rho(\calC)|$ is the number of trees of the MTSF $\calC$.
Definition \eqref{eq:proba_MTSF} calls for comments.
First, although the calculation of the holonomy $\mathrm{hol}(\eta)$ depends on the orientation of the cycle $\eta$, we note that {$\cos \theta(\eta)$} is invariant under orientation flip.
In particular, the measure \cref{eq:proba_MTSF} favours MTSFs with inconsistent cycles, namely, cycles $\eta$ such that {$1- \cos \theta(\eta)$} is large.
Second, a small value of $q>0$ promotes MTSFs with more cycle-rooted trees than trees.
Third, for a trivial connection $\phi_{vv'} = 1$, \cref{eq:proba_MTSF} is a measure on spanning forests~\citep[SF;][]{AvGaud2018} (see \cref{fig:SF}) and, in the limit $q\to 0$, we recover the uniform measure on spanning trees~\citep[ST;][]{Pemantle91,Lyons2003}; see \cref{fig:ST}.
Fourth, for a non-trivial connection, the limit $q\to 0$ gives a measure on cycle-rooted spanning forests \citep[CRSFs;][Section~3]{kenyon2011}.
A CRSF is nothing else than an MTSF where all the connected components are cycle-rooted trees; see \cref{fig:CRSF}.
In the most general case of a non-trivial connection with $q>0$, it is easy to check that the number of edges in an MTSF $\calC$ is $|\calC| = n - |\rho(\calC)|$.  In the  case $q=0$, $\calC$ is almost surely a CRSF and  $|\calC| = n$.

The measure \cref{eq:proba_MTSF}  on MTSFs is actually associated with a discrete determinantal point process (DPP) on the edges of the graph.
A discrete DPP is a distribution over subsets of a finite ground set $[m]=\{1,\dots,m\}$ that favours diverse sets, as measured by a kernel matrix.
We restrict here to DPPs with symmetric kernels.

\begin{definition}[Discrete DPP]\label{def:discreteDPP}
    Let $K$ be an $m\times m$ Hermitian matrix with eigenvalues within $[0,1]$, called a correlation kernel.
    We say that $\calC\sim\DPP(K)$ if
    \[
        \Pr(\calA\subseteq \calC) = \det(K_{\calA\calA}),
    \]
    for all $\calA\subseteq [m]$. If $K$ has all its eigenvalues being equal to either $0$ or $1$, $\DPP(K)$ is called projective.
\end{definition}
The distribution in Definition~\ref{def:discreteDPP} is characterized by its so-called \emph{inclusion} probabilities.
Its existence is guaranteed by the Macchi-Soshnikov theorem \citep{Mac75,Soshnikov00}.
When the spectrum of $K$ is further restricted to lie within $[0,1)$, the DPP is called an $\mathrm{L}$-ensemble, and the probability mass function actually takes a simple form \citep{Mac75}.
For an exhaustive account on discrete DPPs, we refer to \cite{KT12}.
\begin{proposition}[$\mathrm{L}$-ensemble]
    Let $\calC\sim \DPP(K)$, with $\Sp(K)\subset[0,1)$.
    Let $\mathrm L = (I-K)^{-1}K$ be the so-called likelihood kernel of the DPP.
    Then
    \begin{equation}
    \label{eq:l-ensemble-likelihood}
        \Pr(\calC) = \det(\mathrm{L}_{\calC\calC})/\det(\mathrm{L}+\I).
    \end{equation}
\end{proposition}

We now show that \eqref{eq:proba_MTSF} is a DPP.
Let $q>0$.
In what follows, we simply set $w_e = 1$ for all $e\in \mathcal{E}$ to simplify mathematical expressions.
We define an $\mathrm{L}$-ensemble with correlation kernel
\begin{equation}
    \label{eq:K}
    K
    = B\left(\Delta+ q \I_n\right)^{-1}B^*
    = q^{-1} BB^* \left(q^{-1} BB^*+  \I_m \right)^{-1},
\end{equation}
where we read off the likelihood kernel $\mathrm{L} = q^{-1}BB^*$.
A direct calculation shows that this DPP is precisely the distribution we introduced in \eqref{eq:proba_MTSF}, i.e., $\Pr(\calC) = p(\calC)$.
More precisely, using the definition of an $\mathrm{L}$-ensemble and Sylvester's identity $\det(q^{-1} B B^* + \I_m) = \det(q^{-1} B^* B +  \I_n)$, we can rewrite \eqref{eq:l-ensemble-likelihood} as
\[
    \Pr(\calC) = q^{|\rho(\calC)|}\frac{\det(B_{\calC:} B^*_{\calC:})}{\det(\Delta + q \I_n)}.
\]
The product of factors involving the holonomies in \eqref{eq:proba_MTSF} arise from the use of the Cauchy-Binet identity applied to $\det(B_{\calC:} B^*_{\calC:})$; see \citet{kenyon2011} for a complete discussion.
\begin{remark}[Regularized Laplacian.]\label{rem:reg_Laplacian_loops}
{In the context of semi-supervised learning, we often consider the regularized Laplacian $\Delta + q \mathbb{I}$ with the regularization parameter $q >0$; as we discussed in \cref{sub:regularized_laplacian_systems_preconditioning_and_sparse_cholesky_factorization}.
Thus, the $q$-dependent MTSF measure \cref{eq:proba_MTSF} has a natural motivation in this case.
However, it is easy to see that the measure \cref{eq:proba_MTSF}
\begin{equation*}
    p(\calC)
    \propto q^{|\rho(\calC)|}
    \prod_{\eta \in \text{cycles}(\calC)}
    \Big(2 - 2\cos \theta(\eta)\Big),
\end{equation*} can be obtained as a particular case of a measure over CRSFs associated to an auxilliary graph, obtained by augmenting the graph with self-loops.\footnote{
    We thank one of the anonymous referees for this suggestion.}
This \emph{ad hoc} construction is obtained by adding to each node of the original graph a loop with holonomy $e^{\rmi \vartheta}$ and a product of
edge weights $w$ such that $w \cdot (2-2\cos \vartheta) = q$.
In this paper, we chose not to use this viewpoint, since the definition of an incidence matrix is less obvious in this case,and since a synchronization problem with a self-loop has no direct meaning.
Furthermore, in our motivating example of semi-supervised learning, tuning the parameter $q$ allows for the classical bias-variance trade-off.
This is why we prefer to keep the role of this parameter explicit.}
\end{remark}
\begin{table}
    \begin{center}
        \begin{tabular}{ll l l}
            \toprule
            Parameter & Connection & Spanning graph sampled by $\DPP(K)$ & Illustration\\
            \midrule
            $q = 0$ & trivial & Spanning tree (ST) & \cref{fig:ST}\\
            $q >0$ & trivial & Spanning forest (SF) & \cref{fig:SF}\\
            $q = 0$ & non-trivial & Cycle-rooted spanning forest (CRSF) & \cref{fig:CRSF}\\
            $q >0$ & non-trivial & Multi-type spanning forest (MTSF) & \cref{fig:MTSF}\\
            \bottomrule
        \end{tabular}
    \end{center}
    \caption{Joint edge samples associated with $\DPP(K)$ and  $K = B\left(\Delta+ q \I_n\right)^{+}B^*$ as a function of $q$. Note that, if the connection is trivial, the magnetic Laplacian $\Delta$ is unitarily equivalent to the combinatorial Laplacian $\Lambda$; see \cref{sec:introduction}.}
    \label{tab:connection-sample-type}
\end{table}

\subsection{A few consequences of determinantality}
From \cref{def:discreteDPP} the marginal inclusion probability of any edge is given by
\begin{equation}
        \Pr(e\in \calC) = K_{ee} \triangleq \lev(e), \label{eq:LSdef}
\end{equation}
which is the so-called {\emph{magnetic}} leverage score of $e\in \calE$.
{In the absence of a connection, the (standard) leverage scores obtained from the diagonal of $K = B_0 \Lambda^+ B_0^\top$} play an important role in the seminal work of \citet{SpielmanSrivastava}, who build spectral sparsifiers using i.i.d.\ {\emph{leverage score-weighted}} edge sampling{; see also \citep{cohen2015uniform} for a version with edge-by-edge sampling}.
In this sense, the DPP \eqref{eq:proba_MTSF} generalizes the latter construction, where off-diagonal entries of $K$ encode negative dependence between edges.

The expansion of the normalizing constant of \eqref{eq:proba_MTSF}, see \citet[Theorem 2.4]{kenyon2019} and originally derived for CRSFs by~\citet{FORMAN199335}, is given by
\begin{equation}
    \label{eq:Normalization}
    \det(\Delta + q \I_n)
    = \sum_{\text{MTSFs } \calC}
    q^{|\rho(\calC)|}
    \prod_{\text{cycles } \eta \in \calC}
    \Big(2 - {2\cos \theta(\eta)}\Big),
\end{equation}
where the sum is over all non-oriented MTSFs, and $\mathrm{hol}(\eta)$ are $1/\mathrm{hol}(\eta)$ are the holonomies of the two possible orientations of the cycle $\eta$.
We refer to~\citet{ALCO_2020} for a generic proof technique.

Another consequence of the determinantal structure of the distribution \cref{eq:proba_MTSF} is that the expected number of edges in an MTSF and its variance are given in closed form as
\[
    \E_{\calC \sim \DPP(K)}[|\calC|]
    = \Tr\Big(\Delta (\Delta + q \I_n)^{-1}\Big),
    \quad\text{and}\quad
    \mathbb{V}_{\calC \sim \DPP(K)}[|\calC|]
    = \Tr\Big(\Delta (\Delta + q \I_n)^{-2}\Big),
\]
see e.g. \citep{KT12}.
In particular, this confirms that the expected number of edges culminates at $n$ for $q \to 0$ (CRSF) and goes to zero as $q \to +\infty$. In other words, a large value of $q$ tends to promote MTSFs with small trees.
\begin{remark}[$q = 0$]
    In the case of a non-trivial connection, the correlation kernel \cref{eq:K} for $q = 0$ is the orthogonal projector $B \Delta^{-1} B^*$. The corresponding DPP is then a projective DPP whose samples are CRSFs; see \cref{fig:CRSF} for an illustration and \citet{kenyon2011} for a reference. This DPP samples sets of constant cardinality.
\end{remark}

In the absence of any specific structure, sampling a discrete DPP can be done thanks to a generic algorithm \citep{HKPV06}, which relies on the eigendecomposition of the correlation kernel \cref{eq:K}, followed by a linear algebraic procedure analogous to a chain rule.
The time complexity of this algorithm is $\mathcal{O}(m^3)$ where $m$ is the number of edges.
For the specific case of \cref{eq:K}, in the vein of \cite{Wilson96} and \cite{kassel2017}, we show in \cref{sec:sampling_a_multitype_spanning_forest} how to leverage the graph structure to obtain a cheaper but still exact sampling random walk-based algorithm.

\section{Statistical guarantees of sparsification} 
\label{sec:sparsification_of_the_regularized_magnetic_laplacian}
We now give our main statistical results, guaranteeing that the DPP of Section~\ref{sec:dpp_sampling_of_edges__multitype_spanning_forests} yields a controlled sparse approximation of the magnetic Laplacian.
Throughout this section, we work under the following assumption.
\begin{assumption}[Non-singularity of $\Delta + q \I_n$]\label{assump:non-singularity}
    Let a connected graph be endowed with a $\Uone$-connection and denote by $\Delta + q \I_n$  its (regularized) magnetic Laplacian with $q \geq 0$.  We assume either that
    \begin{itemize}
        \item[(i)] the $\Uone$-connection is non-trivial (cf. \cref{tab:connection-sample-type}), or
        \item[(ii)] $q >0$.
    \end{itemize}
\end{assumption}
Under \cref{assump:non-singularity}, $\Delta + q \I_n$ has only strictly positive eigenvalues since either $\Delta$ is strictly positive definite or $q>0$.
\subsection{Sparsification with one MTSF \label{sec:one_MTSF_section}}
For an MTSF $\calC$, we define the $\vert \calE\vert\times \vert \calE\vert$ sampling matrix $S(\calC)$ by
\[
    S_{ee'}(\calC) = \delta_{ee'} 1_\calC(e') /\sqrt{\lev(e)},
\]
where $\lev(e)$ is the leverage score of the edge $e\in \calE$; see \cref{eq:LSdef}. This matrix is obtained from the identity matrix by dividing each column indexed by some $e\in \calC$ by $\sqrt{\lev(e)}$ and filling with zeros all the remaining columns.
By construction, under the DPP with correlation kernel \cref{eq:K}, this sampling matrix satisfies 
$
    \E_{\calC}[SS^\top] = \I_{m},
$
since $\lev(e) = \Pr(e\in \calC)$.
Consider now
\begin{equation}
    \widetilde{\Delta}(\calC) = B^* SS^\top  B.
    \label{eq:Delta_Tilde}
\end{equation}
We emphazise that $\widetilde{\Delta}$ is the Laplacian of a weighted graph with fewer edges than the original graph, thus motivating the word ``sparsification''.
The weights of the edges in the sparsified graph can be read from the non-zero (diagonal) elements of $SS^\top$.

\begin{proposition}[Sparsification from one MTSF]
    \label{prop:OneMTSF} Let \cref{assump:non-singularity} hold and $K = B\left(\Delta+ q \I_n\right)^{-1}B^*$.
    Let $\calC\sim \DPP(K)$, and $\widetilde{\Delta}(\calC)$ be given by \eqref{eq:Delta_Tilde}.
    Let further $\mumax = \|K\|_{\op}$ and $\deff = \Tr(K)$ be the average number of edges sampled by $\DPP(K)$.
    Finally, let $\delta \in (0,1)$, and $\epsilon >0$ such that
    $$
    \epsilon^2 = 37 \mumax \left( 2 \log\left(\frac{4 \deff}{\delta\mumax }\right)\vee \sqrt{3}\right).
    $$
    With probability at least $1-\delta$,
    \[
        (1-\epsilon)(\Delta + q \I_n)
        \preceq \widetilde{\Delta}(\calC) + q \I_n
        \preceq (1+\epsilon)(\Delta + q \I_n).
    \]
\end{proposition}
The proof is based on a new matrix Chernoff bound given in \cref{thm:Chernoff_StronglyRayleigh}.
\begin{proof}
    Under \cref{assump:non-singularity}, $\Delta + q \I_n$ has only strictly positive eigenvalues. Hence, we can introduce the $n\times m$ matrix $\Psi = (\Delta + q \I_n)^{-1/2} B^*$, so that $K = \Psi^* \Psi.$
    Now, define
    \[
        Y_e =\frac{1}{\lev(e)} \bmpsi_e \bmpsi_e^* \text{ for all } e\in [m],
    \]
    where $\bmpsi_e$ denotes the $e$-th column of $\Psi$ and $\lev(e) = \bmpsi_e^*\bmpsi_e$.
    By definition, $\|Y_e\|_{\op}\leq 1$ and $\E_{\calC}\left[\sum_{e\in \calC}Y_e\right] = \sum_{e=1}^m \bmpsi_e\bmpsi_e^* = \Psi\Psi^*$.
    By construction, $\mumax = \lambda_{\max}(\Psi\Psi^*) \leq 1$.
    With the notations of \cref{thm:Chernoff_StronglyRayleigh}, for all $\varepsilon\in(0,1]$, it holds
    \[
        \Pr\left[
        \left\|
        \sum_{e\in \calC} Y_e - \E_{\calC}
        \left[\sum_{e\in \calC}Y_e\right]
        \right\|_{\op}
        \geq \varepsilon\mumax
        \right]
        \leq
        2\intdim(M)
        \left( 1+ c_1\frac{r^2}{\mumax^2 \varepsilon^4} \right)
        \exp\left(-c_2\frac{\varepsilon^2 \mumax}{r}\right),
    \]
    where we take $r = 1$ and with $c_1 = 3\cdot 37^2$ and $c_2 = \frac{1}{2 \cdot 37}$.
    Now, fix $\varepsilon^4 \geq c_1/\mumax^2$, so that
    \[
        \Pr\left[\left\|\sum_{e\in \calC} Y_e - \E_{\calC}\left[\sum_{e\in \calC}  Y_e\right]\right\|_{\op} \geq \varepsilon\mumax\right]  \leq 4\intdim(M) \exp\left(-c_2 \varepsilon^2 \mumax\right).
    \]
    Let $\delta\in (0,1)$ such that
    $
        4\intdim(M) \exp\left(-c_2 \varepsilon^2 \mumax\right)\leq \delta,
    $
    provided that $\varepsilon$ is large enough.
    We identify $M = \Psi \Psi^*$ and we find $\intdim(M) = \intdim(\Psi^* \Psi) = \deff /\mumax$.
    Now, we choose  $\varepsilon^2 \geq \frac{1}{c_2\mumax}\log(\frac{4 \deff}{\mumax\delta})$ and define $\epsilon = \mumax \varepsilon$, whereas we emphasize that $\varepsilon$ and $\epsilon$ are distinct quantities.
    Therefore,
    we have
    \[
        -\epsilon  \I \preceq \sum_{e\in \calC} Y_e - \E_{\calC}\left[\sum_{e\in \calC}Y_e\right]\preceq \epsilon \I,
    \]
    with probability at least $1-\delta$.
    Therefore, we have
    $
        -\epsilon  \I \preceq \Psi SS^\top \Psi^* - \Psi\Psi^* \preceq \epsilon  \I,
    $
    by definition of $S = S(\calC)$ at the begining of \cref{sec:one_MTSF_section}.
    Note that we take
    \[
        \epsilon^2 = \left(c_1^{1/2}\mumax \right)\vee \left(\frac{\mumax}{c_2 }\log\left(\frac{4 \deff} {\mumax \delta}\right) \right).
    \]
    Finally, plugging back $\Psi = (\Delta + q \I_n)^{-1/2} B^*$, we obtain the following result. With probability at least $1-\delta$, it holds that
    \[
        -\epsilon  (\Delta + q \I_n) \preceq \widetilde{\Delta}(\calC) - \Delta \preceq \epsilon  (\Delta + q \I_n).
    \]
    This completes the proof.
\end{proof}

It is well-known \citep[Theorem 13.3]{Vishnoi2013} that we can construct a Cholesky decomposition of a suitably permuted Laplacian matrix of a spanning tree with at most $\mathcal{O}(n)$ non-zero off-diagonal entries.
Similarly, we can construct a very sparse Cholesky decomposition of the regularized magnetic Laplacian matrix of an MTSF thanks to a suitable node ordering.
The sparsity of this decomposition depends on the set of nodes in the cycles of the MTSF, that we denote by $\calV_c$.
In a word, the resulting triangular matrix has $\mathcal{O}(n + |\calV_c|)$ non-zero off-diagonal entries, and is obtained in $\mathcal{O}(n + |\calV_c|)$ operations.
This is stated more formally in \cref{prop:CholeskyMTSF}.
\begin{proposition}[Sparse Cholesky decomposition for the Laplacian of one MTSF]
    \label{prop:CholeskyMTSF}
    Consider an MTSF with $r$ rooted trees and $c$ cycle-rooted trees, with cycle-lengths $n_1, \dots, n_c$.
    Let $\widetilde{\Delta}$ be the magnetic Laplacian of this MTSF endowed with a unitary connection and let $q\geq 0$.
    There exits an ordering of the nodes of this MTSF such that the Cholesky factorization of $\widetilde{\Delta} + q \mathbb{I}$ has at most $n -r + (n_1 -3) +\dots + (n_c -3)$ non-zero off-diagonal entries.
    This Cholesky factorization then requires $\mathcal{O}(n + n_1 + \dots + n_c)$ operations.
\end{proposition}
In particular, \Cref{prop:CholeskyMTSF} implies { that $\widetilde{\Delta}$ can be factorized with a triangular matrix $R$  with at most $\mathcal{O}(n + |\calV_c|)$ non-zero off-diagonal entries. 
Hence, the linear system $R \bmf = \bmy$ can be solved in $\mathcal{O}(n +|\calV_c|)$ in time.}
The proof of \cref{prop:CholeskyMTSF}  given in \cref{sec:cholesky_decomposition_of_the_magnetic_laplacian_of_a_mtsf} also gives a specific recipe to construct the permutation of the nodes such that $R$ is sparse.

\subsection{Sparsification from a batch of independent MTSFs}

The value of $\epsilon$ in the multiplicative bound from \cref{prop:OneMTSF} can be large.
To take it down, we now consider a batch (or ensemble), $\calC_1, \dots, \calC_t$ of $t$ MTSFs, drawn independently from the same DPP as in \cref{prop:OneMTSF}.
The corresponding average of sparse Laplacians,
\[
    \frac{1}{t} \sum_{\ell=1}^t\widetilde{\Delta}(\calC_\ell) = B^* S^{(b)} S^{(b)\top}B,
\]
is itself the Laplacian of a sparse graph, with edge set $\cup_{\ell=1}^t \calC_\ell$, with suitably chosen edge weights.
Note that each edge in the batch is a multi-edge.
The $\vert \calE\vert\times \vert \calE\vert$ sampling matrix of the batch is diagonal with entries given by
\begin{align}
    \label{eq:sampling_matrix}
    S_{ee'}^{(b)}
    =
    \delta_{ee'}
    \sqrt{\frac{\mathrm{n}(e)}{t\times \lev(e)}}
    1_{\cup_{\ell=1}^t \calC_\ell}(e^\prime),
\end{align}
where $\mathrm{n}(e)$ is the number of times edge $e$ appears in the batch of MTSFs.
\begin{theorem}[Sparsification with a batch of MTSFs]
    \label{thm:BatchMTSFs}
    Let \cref{assump:non-singularity} hold.
    Let $\epsilon\in (0,1)$ and denote by $\calC_1, \dots, \calC_t$ MTSFs drawn independently from $\DPP(K)$ with $K = B\left(\Delta+ q \I_n\right)^{-1}B^*$.
    Let $\mumax = \|K\|_{\op}$ and $\deff = \Tr(K)$ the average number of edges sampled by this DPP.
    Denote by $\widetilde{\Delta}(\calC_\ell)$ for $1\leq \ell \leq t$ the corresponding sparsified magnetic Laplacian.

    Let $\delta \in (0,1)$.
    Then, if
    \[
        t
        \geq
         \frac{37\mumax}{ \epsilon^{2}}
        \left(
        2 \log \left(\frac{4 \deff}{\delta\mumax} \right) \vee \sqrt{3}
        \right),
    \]
    {with probability at least $1-\delta$},
    \[
        (1-\epsilon)(\Delta + q \I_{n})
        \preceq
        \frac{1}{t}
        \sum_{\ell=1}^t\widetilde{\Delta}(\calC_\ell)
        + q \I_n
        \preceq (1+\epsilon)(\Delta + q \I_{n}).
    \]
\end{theorem}
This result is similar to the conclusion of~\cite{KKS22} who show that $\mathcal{O}\left(\epsilon^{-2}\log n\right)$ random spanning trees give a $(1\pm \epsilon)$ sparsifier of the combinatorial Laplacian.
\begin{proof}
    Let $\pp$ be the point process on $[tm]$ made of the concatenation of $t\geq 1$ independent copies of $\DPP(K)$.
    We first note that $\pp$ is a again a DPP, and its correlation kernel is a block diagonal matrix with $t$ copies of $K$ along the diagonal.
    We now apply \cref{thm:Chernoff_StronglyRayleigh} to that DPP.

    Define the $n\times m$ matrix $\Psi = (\Delta + q \I_n)^{-1/2} B^*$, so that
    $
        K = \Psi^* \Psi.
    $
    We also need $t$ identical copies of the same set of matrices.
    Hence, we define
    \begin{equation}
        \label{e:extended_Yes}
        Y_{e + \ell m} = \frac{1}{t\times \lev(e)}\bmpsi_e \bmpsi_e^*, \text{ for all } \ell \in\{0, t-1\} \text{ and all } e \in [m],
    \end{equation}
    where $\bmpsi_e$ denotes the $e$-th column of $\Psi$.
    By definition, $\|Y_e\|_{\op}\leq 1/t$ and $\E_{\pp}\left[\sum_{e\in \pp} Y_e\right] = \sum_{e=1}^m \bmpsi_e\bmpsi_e^* = \Psi\Psi^*$.
    Let now
    $
        M \triangleq \E_{\pp}\left[\sum_{e\in \pp} Y_e\right] = \Psi\Psi^*
    $
    and $\mumax = \mumax(M) = \mumax(K)\leq 1$.
    By \cref{thm:Chernoff_StronglyRayleigh}, for all $\varepsilon\in(0,1]$, it holds
    \[
        \Pr\left[\left\|\sum_{e\in \pp}  Y_e - \E_{\pp}\left[\sum_{e\in \pp} Y_e\right]\right\|_{\op} \geq \varepsilon\mumax\right]  \leq 2\intdim(M)\left( 1+c_1 \frac{r^2}{\mumax^2 \varepsilon^4} \right) \exp\left(-c_2\frac{\varepsilon^2 \mumax}{r}\right),
    \]
    for strictly positive constants $c_1$ and $c_2$ and where $r=1/t$. Here, $\intdim(M)= \deff/\mumax$.
    If we take $t \geq \frac{\sqrt{c_1}}{\mumax\varepsilon^2}$, then
    $
        1 + c_1 \frac{1/t^2}{\mumax^2 \varepsilon^4}\leq 2
    $.
    In particular, taking $t$ large enough so that
    $$
        1\geq \delta \geq 4\frac{\deff}{\mumax}\exp\left(-c_2 \varepsilon^2 t \mumax\right),
    $$
    we obtain, with probability at least $1-\delta$,
    \begin{equation}
        \label{e:result_of_thm2}
        -\varepsilon \mumax \I \preceq \sum_{e\in \pp} Y_e - \E_{\pp}\left[\sum_{e\in \pp} Y_e\right]\preceq \varepsilon \mumax \I.
    \end{equation}
    The condition on the batchsize {reads}
    \[
        t \geq \frac{1}{\mumax \varepsilon^2} \left\{  \left(c_2^{-1}\log\left(\frac{4\deff}{\mumax \delta}\right)\right)\vee\sqrt{c_1} \right\}.
    \]
    As we did in the proof of \cref{prop:OneMTSF}, we define $\epsilon = \varepsilon \mumax \in (0,1]$.
    Upon plugging back the definition \eqref{e:extended_Yes} in \cref{e:result_of_thm2}, we obtain the desired $(1\pm \epsilon)$-multiplicative bound with probability larger than $1-\delta$, provided that
    \[
        t \geq \frac{\mumax}{ \epsilon^2} \left\{ \left( c_2^{-1}\log\left(\frac{4\deff}{\mumax \delta}\right)\right)\vee\sqrt{c_1} \right\}.
    \]
    This completes the proof.
\end{proof}

\section{A Matrix Chernoff bound with intrinsic dimension}
\cref{prop:OneMTSF,thm:BatchMTSFs} rely on a new matrix Chernoff bound, with intrinsic dimension, for a sum of positive semidefinite matrices.
We state this result here because it might be of independent interest.
{Note that the intrinsic dimension $\Tr(M)/\|M\|_{\op}$ of a positive semi-definite $d\times d$ Hermitian matrix $M$ gives an estimation of its  rank, by accounting for the number of its eigenvalues which are significantly larger than zero.
The matrix Chernoff bound of \citet{KKS22}, which has inspired this paper, only depends on the matrix dimension $d$,  which may be larger than the intrinsic dimension in some applications.}

\begin{theorem}[Chernoff bound with intrinsic dimension]
    \label{thm:Chernoff_StronglyRayleigh}
    Let $\pp\sim\DPP$ on $[m]$.
    Let $Y_1, \dots, Y_m$ be $d\times d$ Hermitian matrices such that $0\preceq Y_i\preceq r \I$  for all $i\in [m]$ and some $r>0$.
    Assume $\E_{\pp}\left[\sum_{i\in \pp} Y_i\right] \preceq M$ for some Hermitian matrix $M$.
    Let $\intdim(M) = \Tr(M)/\|M\|_{\op}$ and $\mumax = \|M\|_{\op}$.
    Then, for all $\varepsilon\in(0,1]$, it holds {that}
    \[
        \Pr\left[
        \left\|
        \sum_{i\in \pp} Y_i
        - \E_{\pp}\left[\sum_{i\in \pp} Y_i\right]
        \right\|_{\op}
        \geq \varepsilon\mumax
        \right]
        \leq 2\intdim(M)
        \left( 1+ c_1\frac{r^2}{ \mumax^2\varepsilon^4} \right)
        \exp\left(-c_2\frac{\mumax\varepsilon^2}{r}\right),
    \]
    where $c_1 = 3\cdot 37^2$ and $c_2 = \frac{1}{2\cdot 37}$.
\end{theorem}
The proof of \cref{thm:Chernoff_StronglyRayleigh} is given in \cref{sec:main_proof}, after introducing a few useful tools from \cref{app:Chernoff_StronglyRayleigh} to \cref{sec:Chernoff_bounds}.
The proof techniques are {mainly} borrowed from~\citet{KKS22} and adapted to account for a matrix intrinsic dimension \emph{\`a la}~\cite{Tropp15}. {By specializing the approach of \citet{KKS22} to DPPs, the proof can be significantly simplified, hopefully contributing to a wider adoption of the techniques.
}
Let us now discuss the interpretation of \cref{thm:Chernoff_StronglyRayleigh}.

An obvious limitation of the bound in \cref{thm:Chernoff_StronglyRayleigh} is that it describes the concentration only for $\varepsilon$ sufficiently small, whereas the bounds obtained by \citet{Tropp15} for i.i.d.\ sampling  do not suffer from this limitation. Despite this fact, it is still possible to use this result to prove, e.g., \cref{thm:BatchMTSFs}.
The upper bound in \cref{thm:Chernoff_StronglyRayleigh} can be made small provided that the parameter $r$ is small enough and that $\mumax = \| M \|_{\op}$ remains constant.
The latter condition is enforced by making the DPP depend on $r$ so that the expected cardinality of $\pp$ grows if $r$ goes to zero.
Note that the upper bound on the failure probability decreases if the ratio $\mumax\varepsilon^2/r$ increases.
Thus, for a fixed $\varepsilon\in (0, 1]$, the upper bound on the failure probability will take a small value as $r$ goes to zero and as the expected number of sampled points increases;  see the proof of \cref{thm:BatchMTSFs} for more details.
{Finally, the  explicit numerical constants appearing in this tail bound are given for completeness but we do not believe that they are optimal in any sense.}

\subsection{Proof outline}
\label{app:Chernoff_StronglyRayleigh}

We need to prove a matrix Chernoff bound.
We strongly rely on the proof techniques in \citep{KyngSong,KKS22,Tropp15}, restated here in the vocabulary of point processes.
We specialize arguments to DPPs whenever we can, yielding quicker proofs in some places.
In particular, we mix matrix martingales arguments à la \cite{Tropp15}, the notion of $\ell_\infty$-independence of \cite{KKS22}, and a dilation argument for DPPs    from \citet{Lyons2003}.

Since this section is relatively long, we provide here a short outline.
\cref{app:Preliminaries} introduces some of the vocabulary of point processes, and the recent notion of $\ell_\infty$-independence of \cite{KKS22} for discrete measures.
In particular, we give tools to bound the conditionals in the chain rule for DPPs.
In \cref{sec:preliminary_results}, we list and derive tools to investigate the concentration properties of random matrices, in particular the matrix Laplace transform, which involves the expectation of the trace of a matrix exponential.
\cref{lem:GeneralizedLemma5.2} gives conditions under which we can push the expectation inside the trace-exponential.
In \cref{sec:Chernoff_bounds}, we use this \cref{lem:GeneralizedLemma5.2} along with a deviation bound from  \cite{Tropp15} and \cite{Oli10} to obtain a matrix Chernoff bound.
This bound features Tropp's intrinsic dimension, applied to the DPP kernel. However, at this stage, the bound only applies to $k$-homogeneous DPPs, that is, DPPs with almost surely fixed cardinality.
We remove the latter restriction using a dilation argument due to \cite{Lyons2003}.
This provides \cref{thm:Intrinsic_Matrix_Chernoff_bound_linfty}, which we introduce for future reference.
We conclude in \cref{sec:main_proof} with the proof its corollary, \cref{thm:Chernoff_StronglyRayleigh}.

\subsection{Discrete measures, point processes and $\ell_\infty$-independence}
\label{app:Preliminaries}

Throughout this section, we consider a probability measure $\mu$ on $\{0,1\}^m$.
Recall that $\mu$ is said to be $k$-homogeneous if the Boolean vector $\bmxi\sim\mu$ almost surely has exactly $k$ entries equal to $1$.
Finally, for $\xi\sim\mu$, we consider
\begin{equation}
    \label{def:pp}
        \pp = \pp(\bmxi) \triangleq \{i \in [m] : \xi_i = 1\} \subset [m].
\end{equation}
$\pp$ is a point process on $[m]$, i.e., a random subset of $[m]$.
Similarly, we can talk of the Boolean measure $\mu$ associated to a point process $\pp$.

\subsubsection{Intensity and density}
\label{sec:Intensity_Density}
For a point process $\pp$, define its \emph{intensity} as $\rho_\pp:[m]\rightarrow \mathbb{R}_{\geq 0}$ given by
$$
\rho_\pp(i) = \E_{\pp}[1_{\pp}(i)], \quad i\in[m].
$$
\begin{example}
    \label{ex:DPPs}
    \cref{def:discreteDPP} implies that for $\pp\sim \text{DPP}(K)$, $\rho_\pp(i) = K_{ii}$ for any $i\in [m]$.
\end{example}

Furthermore, if $\pp$ is $k$-homogeneous, its \emph{density} at $i\in [m]$ is
$
\nu_\pp (i) = \rho_\pp(i)/k.
$
Intuitively, the intensity $\rho_\pp(i)$ is the probability that $i$ belongs to a sample of $\pp$, whereas the density $\nu_\pp$ is the $1$-homogeneous point process whose samples are obtained by first getting a sample of $k$ points from $\pp$ and then sampling uniformly one of these random $k$ points.

\subsubsection{Conditioning}
To define $\ell_\infty$-independence, we need to express conditioning.
First, for any $\calA_1 \subset [m]$ with $\mu(\calA_1) \neq 0$, define the \emph{reduced} probability measure $\mu_{\calA_1}$ on $[m]\setminus \calA_1$ by
\[
    \mu_{\calA_1} (\calB) \propto \mu(\calA_1 \cup \calB), \quad \calB \subseteq [m]\setminus \calA_1.
\]
This reduced measure translates fixing some of the components $\xi\sim\mu$ to $1$.
The ``dual'' operation consists in the exclusion of $\calA_0 \subset [m]$, where $\calA_0$ is such that $\mu(\calC) < 1$ for all subsets $\calC \subseteq \calA_0$.
We denote the corresponding reduced measure on $[m]\setminus \calA_0$ as
\[
    \mu^{\calA_0} (\calB) \propto \mu(\calB), \quad\calB \subseteq [m]\setminus \calA_0.
\]
This corresponds to fixing $\xi_i = 0$ for all $i\in \calA_0$.
For disjoint $\calA_0, \calA_1\subset[m]$, we can further define the reduced probability measure $\mu_{\calA_1}^{\calA_0}$ corresponding to both excluding $\calA_0$ and including $\calA_1$ by
\[
    \mu^{\calA_0}_{\calA_1} (\calB) \propto \mu(\calA_1\cup\calB), \quad\calB \subseteq [m]\setminus (\calA_0\cup \calA_1).
\]
The notation is consistent, in that, for instance, $\mu_{\emptyset}^{\calA_1} = \mu^{\calA_1}$.
Moreover, the correspondence \eqref{def:pp} between point processes and measures allows to define the point process $\pp_{\calA_1}^{\calA_0}$ associated to the reduced measure $\mu_{\calA_1}^{\calA_0}$.
Finally, for simplicity, for $e\in[m]$, we respectively write $\pp_e$ and $\pp^e$ instead of $\pp_{\{e\}}$ and $\pp^{\{e\}}$.

\begin{example}[\cref{ex:DPPs}, continued]
    \label{ex:conditioning_DPPs}
    Let $\pp\sim\text{DPP}(K)$ with symmetric kernel $K$, and $e\in [m]$ such that $\rho_X(e) = K_{ee} > 0$, then $\pp_e$ is still a DPP, with kernel $f,g\mapsto K_{fg} - K_{fe}K_{eg}/K_{ee}$; see \citep{ShTa03}.
    More generally, for any disjoint $\calA_0, \calA_1$ such that $\calA_1$ is in the support of $\text{DPP}(K)$, then $\pp^{\calA_0}_{\calA_1}$ exists and is a DPP, see e.g. \citep{KT12} for details.
\end{example}

\subsubsection{$\ell_\infty$-independence}
\citet{KKS22} introduced $\ell_\infty$-independence for Boolean measures, and we restate it here in terms of point processes.
Given a point process $\pp$ on $[m]$ with intensity $\rho_\pp$,  define
\begin{equation}
    \delta_{\pp,e}(i) = \begin{cases}
        1 - \rho_\pp(e) & \text{ if } i = e,\\
        \rho_{\pp_e}(i) - \rho_{\pp}(i) & \text{ otherwise},
    \end{cases} \label{eq:delta_e_i}
\end{equation}
for all $e,i \in \supp(\rho_\pp)$.
\begin{definition}[$\ell_\infty$-independence] \label{def:pp_parameter}
    We say that $\mu$ is $\ell_\infty$-independent with parameter (w.p.) $\dinf$ if for all disjoint subsets $\calA_0$ and $\calA_1$ of $[m]$ such that the point process $\mathcal{Y} = \pp_{\calA_1}^{\calA_0}$ exists, for all $e \in \supp(\rho_{\mathcal{Y}})$,
    \[
     \sum_{i\in \supp(\rho_{\mathcal{Y}})} |\delta_{\mathcal{Y},e}(i)| \leq \dinf.
    \]
\end{definition}
Furthermore, to better interpret $\ell_\infty$-independence, it is useful to see a point process as a random counting measure.
Indeed, the quantity that is upper bounded in \cref{def:pp_parameter} is the total variation distance
\[
    d_{\infty}\left( \delta_e + \E_{X_e}\left[\sum_{i\in X_{e}} \delta_i\right], \E_{X}\left[\sum_{j\in X} \delta_j\right]\right)= \sum_{i\in \supp(\rho_{\pp})} |\delta_{\pp,e}(i)|.
\]

DPPs\footnote{And, more generally, measures that satisfy the stochastic covering property \citep{pemantle_peres_2014}.} satisfy $\ell_\infty$-independence with parameter $\dinf = 2$. This is the content of \cref{ex:dp_dinf}.
\begin{example}[Examples~\ref{ex:DPPs} and \ref{ex:conditioning_DPPs}, continued]
    \label{ex:dp_dinf}
    Let $\pp \sim \DPP(K)$ for some symmetric $m\times m$ matrix $K$.
    Then $\rho_\pp(e) = K_{ee} \leq 1$ for all $e\in [m]$.
    Furthermore, for all $e\in [m]$ such that $K_{ee} >0$, we have $\rho_{\pp_e}(i) = K_{ii} - K_{ie}^2/K_{ee}$ for all $i\in [m]$ such that $i\neq e$.
    We deduce that, for $e\in [m]$ such that $K_{ee} >0$,
    \[
        1- \rho_\pp(e) = 1 - K_{ee} \geq 0,
    \]
    and that, for $i\neq e$,
    $$
    \rho_{\pp}(i) - \rho_{\pp_e}(i) =   K_{ie} K_{ei}/K_{ee} = (K_{ie})^2/K_{ee} \geq 0.$$
    It follows that
    \[
        \sum_{i\neq e}| \rho_{\pp_e}(i) - \rho_{\pp}(i)| = \sum_{i\neq e} \frac{K_{ei} K_{ie}}{K_{ee}} = \frac{(K^2)_{ee} - (K_{ee})^2}{K_{ee} }.
    \]
    Now, recall that by the Macchi-Soshnikov theorem, existence of $\pp$ implies $K^2\preceq K$; see Section~\ref{sec:dpp_sampling_of_edges__multitype_spanning_forests}.
    We thus obtain, using notation \eqref{eq:delta_e_i},
    \begin{equation}
        \sum_{i\in \supp(\rho_{\pp})} |\delta_{\pp,e}(i)| = 1  - K_{ee} + \frac{(K^2)_{ee} - (K_{ee})^2}{K_{ee}}  \leq 2 -  2 K_{ee},
        \label{e:conclusion}
    \end{equation}
    for all $e\in [m]$ such that $K_{ee} >0$.

    Now, as noted in Example~\ref{ex:conditioning_DPPs}, if $\pp$ is a DPP, then for any disjoint subsets $\calA_0$ and $\calA_1$ of $[m]$ such that $\pp_{\calA_1}^{\calA_0}$ exists, $\pp_{\calA_1}^{\calA_0}$ is a DPP.
    We can thus prove \eqref{e:conclusion} for this new DPP.
    Consequently, a DPP satisfies $\ell_\infty$-independent with parameter $\dinf = 2$.

    Finally, note that $k$-homogeneous DPPs with symmetric kernels are actually those for which $K^2=K$, so called \emph{projection}, \emph{projective}, or \emph{elementary} DPPs, see, e.g., \citep{HKPV06}.
    In particular, we note that the bound in \eqref{e:conclusion} is tight if $\pp$ is $k$-homogeneous.
    In that case, for a fixed $i\in \supp(\rho_\pp)$, we also have
    $
        \E_{\pp}[\delta_{\pp,e}(i)] = 0.
    $
\end{example}
We end this section on DPPs and independence by a lemma relating $\ell_\infty$-independence and marginals.
In \citep{KKS22}, this lemma is a consequence of what the authors call \emph{average multiplicative independence}.
For our limited use, we can avoid defining this new notion of independence.
\begin{lemma}\label{lem:consequence_of_avg_mult_indpt}
    Let $\pp$ be a point process on $[m]$, which is $k$-homogeneous and $\ell_\infty$-independent with parameter $\dinf$.
    Denote by $\rho_\pp$ and $\nu_\pp$ its intensity and density; see \cref{sec:Intensity_Density}.
    For all $e$ and $i \in \supp(\rho_\pp)$, using notation
    \cref{eq:delta_e_i}, it holds
    \[
        \E_{e\sim \nu_{\pp}} | \delta_{\pp,e}(i) | \leq \dinf  \nu_{\pp}(i), \text{ for all } i \in \supp(\rho_\pp).
    \]
    Furthermore, if $\pp\sim \DPP(K)$ for some orthogonal projector $K$, then $\E_{e\sim \nu_{\pp}}  \delta_{\pp,e}(i) = 0$, and
    \[
        \E_{e\sim \nu_{\pp}} | \delta_{\pp,e}(i) |  = d_\pp(i)\nu_\pp(i) , \text{ for all } i \in \supp(\rho_\pp),
    \]
    where $d_\pp(i) = 2 - 2 K_{ii}$.
\end{lemma}
For completeness, \Cref{lem:consequence_of_avg_mult_indpt} is proved in \cref{sec:proof_of_inf_am}.
\subsection{Manipulating the matrix Laplace transform}
\label{sec:preliminary_results}
Consider a set of Hermitian $d\times d$ matrices satisfying $0\preceq Y_i\preceq r \I$ for $1\leq i\leq m$ and some $r>0$.
Further let $\pp$ be a $k$-homogeneous point process on $[m]$.
In order to study the concentration of $\sum_{i\in \pp}  Y_i$ about its expectation for a $k$-homogeneous $\pp$, we provide Lemma~\ref{lem:GeneralizedLemma5.2} below.
The case $r = 1$ has already been considered by~\citet{KKS22}.
To make our manuscript explicit and self-contained, we extend this analysis to any $r>0$.
\begin{lemma}[Extension of Lemma 5.2 in~\cite{KKS22}]
    \label{lem:GeneralizedLemma5.2}
    Let $\pp$ be a $k$-homogeneous point process on $[m]$, which is $\ell_\infty$-independent with parameter $\dinf$.
    Consider a set of Hermitian $d\times d$ matrices satisfying $0\preceq Y_i\preceq r \I$ for $1\leq i\leq m$ and some $r >0$.
    Let $c \geq r \dinf \vee  \frac{r}{2} \left( 9 \dinf^2 + 1\right) $.
    Then, for every
    $\theta \in[-1/(2c), 1/(2c)]$ and for every Hermitian matrix $H$, the following holds
    \begin{equation}
        \label{e:push_expectation_inside_trace_exponential}
        \E_{\pp}
        \Tr \exp\left( H + \theta\sum_{i\in \pp} Y_i \right)
        \leq
        \Tr \exp\left(
        H +
        (\theta + c\theta^2)
        \E_{\pp}\left[\sum_{i\in \pp} Y_i\right]
        \right).
    \end{equation}
\end{lemma}
The proof of \cref{lem:GeneralizedLemma5.2} is by induction.
Before we proceed, we derive two intermediate results, \cref{lem:Fact5.4} and \cref{lem:bound_on_c}.

\Cref{lem:Fact5.4} is an adaptation of \citet[Claim~5.4, arxiv v2]{KKS22}.
It allows controlling the difference between an expected sum of matrices sampled w.r.t.\ a $k$-homogeneous point process and the same sum conditioned on a certain term being already present.
For any $e$ such that $\rho_\pp(e) > 0$, define the increment
\begin{equation}
    Z_e(\pp) = \E_{\pp_e} \left[Y_e + \sum_{i \in \pp_e} Y_i\right] - \E_{\pp} \left[\sum_{i\in \pp} Y_i\right] =  \sum_{i\in \supp(\rho_{\pp})} \delta_{\pp,e}(i) Y_i,\label{eq:Z}
\end{equation}
where $\pp_e$ is the reduced process and $\delta_{\pp,e}(i)$ is given in \cref{eq:delta_e_i}.
This quantity can be interpreted as the difference sequence of a matrix martingale associated to the sequential revelation of $\pp = \{e_1, \dots , e_k\}$ by the chain rule\footnote{This martingale appears in \citep{pemantle_peres_2014} in the scalar case.}.
Indeed, note that $\E_{e\sim \nu_\pp}[Z_e(\pp)] = 0$, and define the martingale
\[
    M_j = Y_{e_1} + \dots + Y_{e_j} + \E_{\pp_{e_1, \dots, e_j}}\left[\sum_{i\in \pp_{e_1, \dots, e_j} }Y_i\right] - \E_{\pp} \left[\sum_{i\in \pp} Y_i\right].
\]
Then $Z_{e_j}(\pp_{e_1, \dots, e_{j-1}}) = M_j -M_{j-1}$ for $j\in[k]$, with the initialization $M_0 = 0$.
Note that $\pp_{e_1, \dots, e_{k}}$ is simply the empty set.
In particular, a telescopic sum argument yields
\[
     M_k = Z_{e_k}(\pp_{e_1, \dots, e_{k-1}}) + \dots+ Z_{e_2}(\pp_{e_1}) + Z_{e_1}(\pp) =  \sum_{i \in \pp} Y_i - \E_{\pp} \left[\sum_{i\in \pp} Y_i\right],
\]
which will be implicitly used below in the proof of \cref{lem:GeneralizedLemma5.2}.
\begin{fact}
    \label{lem:Fact5.4}
    Let $k\geq 1$ and $\pp$ be a $k$-homogeneous point process on $[m]$, which is $\ell_\infty$-independent with parameter $\dinf$.
    Let $0\preceq Y_i\preceq r \I$ for $1\leq i\leq m$.
    Let $e$ be such that $\rho_\pp(e) > 0$ and $Z_e = Z_e(\pp)$ be defined in \cref{eq:Z}.
    Then
    \begin{itemize}
        \item[(i)] $|Z_e| \preceq r \dinf  \I$,
        \item[(ii)] $\E_{e\sim \nu_\pp}[Z_e^2] \preceq r \dinf^2  \E_{e\sim \nu_\pp}[Y_e]$,
    \end{itemize}
    where $\nu_\pp$ is the intensity of $\pp$.
\end{fact}
\cref{lem:Fact5.4} straightforwardly relies on \cref{lem:consequence_of_avg_mult_indpt}, and we give a proof in \cref{sec:proof_of_Fact5.4}.
Now, we use \cref{lem:Fact5.4} to derive the following lemma with explicit constants.
\begin{lemma}[Extension of Lemma~5.3 in~\cite{KKS22}]
    \label{lem:bound_on_c}
    Under the assumptions of \cref{lem:Fact5.4}, if $c\geq r \dinf \vee  \frac{r}{2} \left( 9 \dinf^2 + 1\right)$, then for all $\theta \in [-1/(2c), 1/(2c)]$,
    \begin{align*}
        \E_{e\sim \nu_\pp}\left[e^{(\theta + c \theta^2)Z_e - c\theta^2 Y_e}\right]\preceq \I.
    \end{align*}
\end{lemma}
\begin{proof}
    The idea of this proof is to use $\exp(A-B) \preceq \I + A - B + 2A^2 + 2 B^2$ for $A$ and $B$ {Hermitian} such that $A\preceq \I$ and $B\succeq 0$ \citep[Fact 4.2]{KKS22}. Here we choose $A = (\theta + c \theta^2)Z_e$ and $B = c\theta^2 Y_e$.

    First, we have
    \begin{align*}
        (\theta + c \theta^2)Z_e
         & \preceq \frac{3}{4c} Z_e\quad \text{(if $\theta\leq 1/(2c)$)}             \\
         & \preceq \frac{1}{c} r \dinf \I \quad \text{(by \cref{lem:Fact5.4} (i))} \\
         & \preceq \I. \quad \text{  (for $c\geq r \dinf$)}
    \end{align*}
    Then, we find
    \begin{align*}
        \E_{e\sim \nu_\pp}
        \left[
            e^{(\theta + c \theta^2)Z_e - c\theta^2 Y_e}
            \right]
         & \preceq
        \E_{e\sim \nu_\pp}
        \left[
            \I
            + (\theta + c \theta^2)Z_e
             -c \theta^2 Y_e
            + 2(\theta + c \theta^2)^2 Z_e^2
            + 2 (c\theta^2)^2 Y_e^2
            \right]
        \\
         & \preceq
        \I
        - c \theta^2 \E_{e\sim \nu_\pp}[Y_e]
        + \frac{9}{2} \theta^2 \E_{e\sim \nu_\pp}[Z_e^2]
        + \frac{1}{2}\theta^2 r \E_{e\sim \nu_\pp}[Y_e]
        \quad \text{(since $|\theta|\leq 1/(2c)$ and $Y_e\preceq r \I$)}
        \\
         & \preceq
        \I
        -  \theta^2
        \E_{e\sim \nu_\pp}[Y_e]
        \left(
        c-\frac{9}{2}r \dinf^2 -\frac{1}{2}r
        \right) \
        \text{(by \cref{lem:Fact5.4} (ii))}
        \\
         & \preceq \I,
    \end{align*}
    where we used that $\E_{e\sim \nu_\pp}[Z_e] = 0$ and where the last inequality is due to  $c\geq \frac{9}{2}r \dinf^2 + \frac{1}{2}r$.
\end{proof}
We are now ready to prove \cref{lem:GeneralizedLemma5.2}. This proof simply follows the argument of~\citet[Lemma 5.2]{KKS22} and is given below in order to make this paper self-contained.
The only changes are as follows: the upper bound $Y_i\preceq r\I$ and the condition on $c$ which is derived in \cref{lem:bound_on_c}\footnote{It is worth noticing that the proof technique of the Freedman inequalily for matrix martingales obtained by \citet[Theorem 1.2]{Oli09} uses the Golden-Thomson inequality in a similar way as \citet{KKS22} and in the proof of \cref{lem:GeneralizedLemma5.2} below.}.

\begin{proof}[Proof of \cref{lem:GeneralizedLemma5.2}]
    We proceed by induction over the degree of homogeneity of the point process, using the chain rule.

    First, we show that the statement is true if $k=1$, i.e., a point process with almost surely one point.
    In that case, the left-hand side of \cref{e:push_expectation_inside_trace_exponential} reduces to
    \begin{align}
        \E_{\pp}\Tr \exp\left( H + \theta\sum_{i\in \pp} Y_i \right) = \E_{e\sim \nu_\pp} \Tr \exp\left( H + \theta Y_e \right). \label{eq:first_step_initialization}
    \end{align}
    We now add and subtract the quantity that appears at the right-hand side of \eqref{e:push_expectation_inside_trace_exponential} and write
    \[\theta Y_e = (\theta + c \theta^2)\E_{e'\sim \nu_\pp} Y_{e'} + (\theta + c \theta^2) Z_e - c\theta^2 Y_e,
    \]
    where $Z_e = Y_e - \E_{e'\sim \nu_\pp} Y_{e'}$.
    Identity \cref{eq:first_step_initialization} and the Golden-Thomson inequality\footnote{For $A$ and $B$ Hermitian matrices, $\Tr(\exp(A+B))\leq \Tr(\exp(A) \exp(B))$. To our knowledge, this technique is due to \citet{AW02}.} together yield
    \begin{align}
        \E_{e\sim \nu_\pp}\Tr\exp(H + \theta Y_e) &\leq \Tr\left[\exp(H + (\theta + c \theta^2)\E_{e'\sim \nu_\pp} Y_{e'}) \E_{e\sim \nu_\pp}\exp((\theta + c \theta^2) Z_e - c\theta^2 Y_e)\right]\nonumber\\
        & \leq \Tr\left[\exp(H + (\theta + c \theta^2)\E_{e'\sim \nu_\pp} Y_{e'}) \|\E_{e\sim \nu_\pp}\exp((\theta + c \theta^2) Z_e - c\theta^2 Y_e)\|_{\op}\right],\label{eq:second_step_initialization}
    \end{align}
    since $\exp(H + (\theta + c \theta^2)\E_{e'\sim \nu_\pp} Y_{e'})$ is positive semidefinite. By \cref{lem:bound_on_c}, if $c\geq r\dinf \vee  \frac{r}{2} \left( 9 \dinf^2 + 1\right)$, for all $\theta \in [-1/(2c), 1/(2c)]$, we have the following bound
    \[
        \|\E_{e\sim \nu_\pp}\exp((\theta + c \theta^2) Z_e - c\theta^2 Y_e)\|_{\op} \leq 1,
    \]
    which can be substituted in \cref{eq:second_step_initialization} to give the desired statement,
    \begin{align*}
        \E_{\pp}\Tr \exp\left( H + \theta\sum_{i\in \pp}  Y_i \right) & = \E_{e\sim \nu_\pp} \Tr \exp\left( H + \theta Y_e \right)\\
        &\leq \Tr\left(\exp(H + (\theta + c \theta^2)\E_{e'\sim \nu_\pp} Y_{e'})\right) \\
        &= \Tr\exp\left(H + (\theta + c \theta^2)\E_{\pp}\left[\sum_{i\in \pp} Y_i\right]\right).
    \end{align*}


    Now, let $l\geq 2$, and assume \cref{lem:GeneralizedLemma5.2} is true for $k=l-1$.
    Then, the chain rule and the induction hypothesis applied to the $(l-1)$-homogeneous point process $\pp_e$ yield
    \begin{align*}
    \E_{\pp}
        \Tr \exp\left( H + \theta\sum_{i\in \pp} Y_i \right) &= \E_{e\sim \nu_\pp}\E_{\pp_e} \Tr \exp\left( H + \theta Y_e + \theta\sum_{i\in \pp_e} Y_i \right) \\
        & \leq  \E_{e\sim \nu_\pp} \Tr \exp\left(
            H + \theta  Y_e +
            (\theta + c\theta^2)
            \E_{\pp_e}\left[\sum_{i\in \pp_e} Y_i\right] \right).
    \end{align*}
    Now we write the argument of the exponential as follows
    \begin{align*}
        \theta Y_e + (\theta + c\theta^2) \E_{\pp_e}\left[\sum_{i\in \pp_e} Y_i\right] &=  (\theta + c\theta^2) \E_{\pp_e}\left[Y_e + \sum_{i\in \pp_e} Y_i \right] -c \theta^2  Y_e \\
        &=  (\theta + c\theta^2) \E_{\pp}\left[\sum_{i\in \pp} Y_i \right] + (\theta + c\theta^2) Z_e -c \theta^2  Y_e,
    \end{align*}
    where $Z_e$ is defined in \cref{lem:Fact5.4}.
    Using the Golden-Thomson inequality again, we deduce
    \begin{align*}
        \E_{\pp}
            \Tr \exp\left( H + \theta\sum_{i\in \pp} Y_i \right)
            & \leq  \E_{e\sim \nu_\pp} \Tr \exp\left(H + (\theta + c\theta^2) \E_{\pp}\left[\sum_{i\in \pp} Y_i \right] + (\theta + c\theta^2) Z_e -c \theta^2  Y_e\right) \\
            &\leq  \Tr \left(\exp\left(H +(\theta + c\theta^2) \E_{\pp}\left[\sum_{i\in \pp} Y_i \right]\right)\E_{e\sim \nu_\pp}\exp\Big((\theta + c\theta^2) Z_e -c \theta^2  Y_e\Big)\right)\\
            &\leq  \Tr \exp\left(H +(\theta + c\theta^2) \E_{\pp}\left[\sum_{i\in \pp} Y_i \right]\right) \left\| \E_{e\sim \nu_\pp}\exp\Big((\theta + c\theta^2) Z_e -c \theta^2  Y_e\Big) \right\|_{\op}\\
            &\leq  \Tr \exp\left(H +(\theta + c\theta^2) \E_{\pp}\left[\sum_{i\in \pp} Y_i \right]\right),
        \end{align*}
        where the last bound is due to \cref{lem:bound_on_c}.
        This concludes the proof.

\end{proof}
\subsection{Intrinsic Matrix Chernoff bound for a DPP with fixed cardinality}\label{sec:Chernoff_bounds}

Along the lines of \citet{KKS22}, we now use Lemma~\ref{lem:GeneralizedLemma5.2} to obtain a deviation bound on the sum of matrices indexed by a $k$-homogeneous point process.

\begin{theorem}[Intrinsic Matrix Chernoff bound] \label{thm:Intrinsic_Matrix_Chernoff_bound_linfty}
    Let $\pp$ be a point process on $[m]$.
    Assume that $\pp$ is $k$-homogeneous and $\ell_\infty$-independent with parameter $\dinf$.
    Let $Y_1, \dots, Y_m$ be a collection of $d\times d$ Hermitian matrices such that $0\preceq Y_i\preceq r\I$  for all $i\in [m]$ and for some $r>0$.
    Assume $\E_{\pp}\left[\sum_{i\in \pp} Y_i\right] \preceq M$ for some Hermitian matrix $M$.
    Denote $\intdim(M) = \Tr(M)/\|M\|_\op$ and $\mumax = \|M\|_\op$.
    Let $c = r \dinf \vee  \frac{r}{2} \left( 9 \dinf^2 + 1\right)$.
    Then, for all $\varepsilon\in(0,1]$, it holds
    \begin{align*}
        \Pr
        \left[
        \left\|
        \sum_{i\in \pp} Y_i
        - \E_{\pp}\left[\sum_{i\in \pp} Y_i\right]
        \right\|_{\op}
        \geq \varepsilon \mumax
        \right]
        \leq
        2 \intdim(M)
        \left( 1+ \frac{12c^2}{\mumax^2 \varepsilon^4} \right)
        \exp\left(-\frac{\varepsilon^2 \mumax}{4c}\right).
    \end{align*}
\end{theorem}
Here are the main differences w.r.t.~\citet{KKS22}.
First, in their assumptions, the authors of the latter paper suppose there exist two non-negative constants $\mu_1$ and $\mu_2$ such that $\mu_1 \I\preceq \E_{\pp}\left[\sum_{i\in \pp} Y_i\right] \preceq \mu_2 \I$, as well as assume $\|Y_i\|_\op\leq 1$ for all $i\in [m]$.
We relax both assumptions.
Second, we replace the matrix dimension $d$ by the intrinsic dimension of \cite{Tropp15}.
Third, the matrices here are Hermitian rather than symmetric.

\begin{proof}
The proof consists of four steps.
First, we apply the matrix Laplace transform bound of~\citet[Proposition 7.4.1]{Tropp15}; see also \citet{AW02}.
Let $U$ be a random Hermitian
matrix, let $\psi: \R \to \R_{\geq 0}$ be a nonnegative function that is nondecreasing on $[0, \infty)$, then for each $t \geq 0$, we have
\begin{equation}
    \label{eq:Markov}
    \Pr\left[\lmax(U) \geq t \right]
    \leq \frac{1}{\psi(t)} \E_\pp [\Tr \psi(U)].
\end{equation}
Let $\theta>0$, and apply the bound to the centered random variable
\begin{equation}
    U = \sum_{i\in \pp} Y_i  - \E_{\pp}\left[\sum_{i\in \pp} Y_i\right],
    \label{eq:U}
\end{equation}
with\footnote{Here we differ from \cite{KKS22} who take $\psi(t) = e^{\theta t}$.} $\psi(t) = e^{\theta t} -t\theta -1$.
It comes
\begin{align}
    \label{eq:right-hand side_1}
    \Pr\left[\lmax\left(\sum_{i\in \pp}  Y_i - \E_{\pp}\left[\sum_{i\in \pp} Y_i\right]\right) \geq t \right]
     & \leq
    \frac{1}{e^{\theta t} -t\theta -1}
    \E_{\pp}
    \Tr \left(
    \exp\left(\theta\left(\sum_{i\in \pp} Y_i - \E_{\pp}\left[\sum_{i\in \pp} Y_i\right] \right)\right)
    -\I
    \right).
\end{align}

The second step of the proof consists in using $\ell_\infty$-independence to bound the right-hand side of \eqref{eq:right-hand side_1}, following the same strategy as~\citet[Proof of Thm 5.1]{KKS22}.
Recall that \cref{lem:GeneralizedLemma5.2} yields
\begin{equation}
    \label{eq:KeyKKS22}
    \E_{\pp}
    \Tr \exp\left( H + \theta'\sum_{i\in \pp}  Y_i \right)
    \leq
    \Tr \exp
    \left(
    H
    + (\theta' + c\theta'^2)
    \E_{\pp}\left[\sum_{i\in \pp} Y_i\right]
    \right),
\end{equation}
for any Hermitian $H$, and for all $\theta' \in [\frac{-1}{2c}, \frac{1}{2c}]$, where $c =r \dinf \vee  \frac{r}{2} \left( 9 \dinf^2  + 1\right)$.
Taking $\theta' = \theta$ and $H = -\theta\E_{\pp}\left[\sum_{i\in \pp} Y_i\right]$, \eqref{eq:right-hand side_1} yields
\begin{align}
    \Pr\left[
        \lmax\left(
        \sum_{i\in \pp} Y_i
        - \E_{\pp}\left[\sum_{i\in \pp}  Y_i\right]
        \right)
        \geq t
        \right]
    &\leq \frac{1}{e^{\theta t} -t\theta -1}
    \Tr \left( \exp \left( -c\theta^2 H \right)-\I\right)\\
    &\leq \frac{1}{e^{\theta t} -t\theta -1}
    \Tr \left( \exp \left( c\theta^2 M \right)-\I\right),
    \label{eq:use_of_M_1}
\end{align}
for all $\theta \in (0, \frac{1}{2c}]$.

The third step is in the introduction of the intrinsic dimension, to further bound the right-hand side of \eqref{eq:use_of_M_1}.
Let $\varphi$ be a convex function on the interval $[0, \infty)$, and assume that $\varphi(0) =0$.
For any positive-semidefinite matrix $\bmA$, \citep[Lemma 7.5.1]{Tropp15} asserts that $\Tr \varphi( \bmA)\leq \intdim (\bmA) \varphi(\|\bmA\|)$.
Take here $\varphi(a) = e^a -1$, and define $g(\theta) =  c\theta^2$.
We obtain
\[
    \Tr \Big( \exp \left(g(\theta) M \right)-\I\Big)
    \leq \intdim(M) \varphi\Big(g(\theta) \|M\|\Big)
    = \intdim(M) \varphi\Big(g(\theta) \mumax\Big),
\]
where we used that $\intdim(\alpha \bmA) = \intdim(\bmA)$ for all $\alpha\neq 0$ and $\bmA\succeq 0$.
By a direct substitution, we obtain
\begin{align}
    \label{eq:generic_Chernoff}
    \Pr\left[
        \lmax\left(
        \sum_{i\in \pp} Y_i - \E_\pp\left[\sum_{i\in \pp} Y_i\right]
        \right)
        \geq t
        \right]
    \leq \intdim(M)
    \frac{e^{g(\theta) \mumax}}{e^{\theta t} -t\theta-1},
\end{align}
for all $\theta\in (0,\frac{1}{2c}]$. 

The fourth step consists in simplifying the obtained bound.
Following~\citet[Eq (7.7.4)]{Tropp15}, we have
    \begin{equation}
        \label{e:tool1}
        \frac{e^{g(\theta) \mumax}}{e^{\theta t} -t\theta-1}
        = \frac{e^{\theta t}}{e^{\theta t} -t\theta-1}
        e^{- \theta t + g(\theta) \mumax }
        \leq \left(1+ \frac{3}{\theta^2 t^2}\right)
        e^{- \theta t +g(\theta) \mumax },
    \end{equation}
    for any positive $\theta$.
    Now, we let $t =\mumax \varepsilon$ with $\varepsilon \in(0,1]$.
    We further note that the argument of the exponential on the right-hand side of \eqref{e:tool1} is minimized by $\theta_\star = \frac{\varepsilon}{2c}$, so that $0<\theta_\star\leq \frac{1}{2c}$.
    Thus, it holds that
    \begin{align*}
        \label{eq:generic_Chernoff}
        \Pr\left[
            \lmax\left(
            \sum_{i\in \pp} Y_i - \E_{\pp}\left[\sum_{i\in \pp} Y_i\right]
            \right)
            \geq \varepsilon \mumax
            \right]
        \leq \intdim(M)
        \left( 1+ \frac{12 c^2}{\mumax^2 \varepsilon^4} \right)
        \exp\left(-\frac{\varepsilon^2 \mumax}{4c}\right).
    \end{align*}
    The same deviation bound for $\lmax\left(
        \E_{\pp}\left[\sum_{i\in \pp} Y_i\right] - \sum_{i\in \pp} Y_i\right)$ can be obtained along the same four steps, thus proving the statement by a union bound.

\end{proof}

\subsection{Going beyond fixed cardinality by a dilation argument}
\label{sec:main_proof}

We now remove the $k$-homogeneity assumption in \cref{thm:Intrinsic_Matrix_Chernoff_bound_linfty} using a homogenization technique that is proper to DPPs.
\begin{proof}[Proof of \cref{thm:Chernoff_StronglyRayleigh}]
    While $\pp$ is any DPP on $[m]$, we follow \citet{Lyons2003} in considering the DPP $\pp'$ on $[2m]$ with correlation kernel
    \[
        K^\prime =
        \begin{pmatrix}
            K & \sqrt{1 - K} \sqrt{K}\\
            \sqrt{1 - K} \sqrt{K} & 1 - K
        \end{pmatrix}.
    \]
    By construction, ${K^\prime}^2= K'$, and $K'$ has rank $m$, therefore $\pp^\prime$ is an $m$-homogeneous DPP, see, e.g., \citep{KT12}.

    Now define the matrices $Y_i = 0$ for all $ i \in \{m+1, \dots, 2m\}$.
    Elementary properties of DPPs imply that $\pp'\cap[m]$ has distribution $\text{DPP}(K)$; see again, for instance, \citep{KT12}.
    In particular, the law of
    \[
        \sum_{i\in \pp^\prime} Y_i = \sum_{i\in \pp^\prime} 1(i\in [m]) Y_i,
    \]
    is the same as the law of $\sum_{i\in \pp} Y_i$, and it is enough to prove a deviation bound for $\sum_{i\in \pp^\prime} Y_i$.
    We simply check that \cref{thm:Intrinsic_Matrix_Chernoff_bound_linfty} applies.
    First,
    \[
        \E_{\pp^\prime}\left[\sum_{i\in \pp^\prime} Y_i\right] = \E_{\pp}\left[\sum_{i\in \pp} Y_i\right] \preceq M.
    \]
    Additionally, since $\pp^\prime$ is a DPP, it is $\ell_\infty$-independent with parameter $\dinf = 2$; see \cref{ex:dp_dinf}.
    \cref{thm:Intrinsic_Matrix_Chernoff_bound_linfty} thus applies, with the constant  $c = r \dinf \vee  \frac{r}{2} \left( 9 \dinf^2 + 1\right)$ where $\dinf = 2$.
    For reference, the numerical values of the absolute constants appearing in \cref{thm:Chernoff_StronglyRayleigh} are $c_1 = 12 c^2/r^2 = 3 \cdot 37^2$ and $c_2 = \frac{r}{4c} =  \frac{1}{2\cdot 37}$.
\end{proof}

\section{Sampling multi-type spanning forests with random walks}
We promote sampling MTSFs with the help of an algorithm based on a random walk called $\cyclepopping{}$ rather than exact algebraic DPP samplers mentioned in \cref{sec:related_work}.
\subsection{$\cyclepopping{}$  for weakly inconsistent graphs} 
\label{sec:sampling_a_multitype_spanning_forest}

To make our statistical guarantees of Section~\ref{sec:sparsification_of_the_regularized_magnetic_laplacian} practical, we need a sampling algorithm for the DPP with kernel \eqref{eq:K}. 
In the case of \eqref{eq:K}, we show here how to leverage the structure of the probability distribution \cref{eq:proba_MTSF} in order to avoid any eigendecomposition, provided that the connection graph is not too inconsistent.
The algorithm described in this section is an {interesting particular case} of the work of \citet{Wilson96} and \citet{kassel2017}.
The main idea is to run a random walk on the nodes of the graph, with an additional absorbing node, popping cycles with some probability as they appear.
After post-processing, one obtains an MTSF with the desired distribution.
We first discuss a seemingly necessary assumption on the connection, and then the different components of the algorithm, before concluding on its running time.
\paragraph{Limitation to weakly inconsistent cycles.}
\citet{kassel2017} proposed a variant of the celebrated \emph{cycle-popping} algorithm of \cite{Wilson96} for efficiently sampling CRSFs, i.e., joint edge samples w.r.t.\ the DPP with kernel \cref{eq:K} with $q=0$.
In this section, we further adapt this idea to sample MTSFs, i.e., $q\geq 0$ in \cref{eq:K}.
To our knowledge, this cycle-popping algorithm can only be used to sample from \cref{eq:proba_MTSF} if
\begin{equation}
    \label{eq:cond_weakly_inconsistent}
    \cos \theta(\eta) \in [0,1],
    \text{ for all cycles } \eta\subset \mathcal{E},
\end{equation}
in which case we refer to the cycles as {\emph{weakly inconsistent}}.
It is yet unclear how to modify this algorithm if there is a cycle $\eta$ such that $1 - \cos \theta(\eta)>1${, in which case we say that it is a \emph{strongly inconsistent} cycle}.
To overcome this limitation, we shall use self-normalized importance sampling with a modification of \cref{eq:proba_MTSF} in Section~\ref{sec:montecarlo_with_selfnormalized_importance_sampling}.
For the rest of this section, we assume that all cycles in the graph are weakly inconsistent.

\paragraph{Oriented MTSFs.}
The algorithm to be described is based on a random walk which naturally generates oriented subgraphs of the original graph.
These oriented subgraphs will be denoted with a superscript $o$.
Examples are provided in \cref{fig:graphs}.

An \emph{oriented rooted tree} is a tree whose edges are oriented towards its root.
An \emph{oriented cycle-rooted tree} is a cycle-rooted tree constituted of an oriented cycle whose edges point in the same direction, whereas all other edges point towards the oriented cycle.
Finally, an \emph{oriented MTSF} contains only \emph{oriented rooted trees} and \emph{oriented cycle-rooted trees}.
We consider the problem of sampling an oriented MTSF with probability
\begin{equation}
    \label{eq:proba_oriented_MTSF}
    \Pr(\calC^{o})
    \propto
    q^{|\rho(\calC^{o})|}
    \prod_{\text{cycles } \eta^{o} \in \calC^{o}}
    \alpha(\eta^{o}),
\end{equation}
where $q>0$, and where $\alpha$ takes values in $[0,1]$ and is invariant under orientation flip.
In particular, taking
\begin{equation}
    \label{eq:def_alpha}
    \alpha(\eta^{o})
    = \frac{1}{2}\Big(2 - 2\cos \theta(\eta^{o})\Big),
\end{equation}
sampling an oriented MTSF with probability \cref{eq:proba_oriented_MTSF}, and forgetting the edge orientation, we obtain an MTSF distributed according to \cref{eq:proba_MTSF}.
In particular, the factor $1/2$ in \eqref{eq:def_alpha} compensates for the fact that there are exactly two oriented cycles for each (unoriented) cycle in an MTSF.

\paragraph{An auxiliary oriented MTSF with a root.}
Following the original approach of \cite{Wilson96} for spanning trees, we now associate to each oriented MTSF an \emph{oriented MTSF with a root}.
Specifically, we first augment $\mathcal{G}$ with an extra node $r$ ($r$ stands for ``root'', as shall become clear), to which every node is connected with an edge with weight $q$, resulting in the graph $\mathcal{G}_r$.
This can be seen as adding a Dirichlet boundary condition at node $r$~\citep[Section 2.2]{Poncelet_2018}.
Note that, if $q  = 0$, node $r$ will play no role.

An \emph{oriented MTSF with root} $r$ is defined as the disjoint union of a single oriented tree rooted at $r$ and oriented cycle rooted trees.
The algorithm to be described below samples an oriented MTSF with root $r$.
\begin{figure}[t]
    \begin{subfigure}[b]{0.25\textwidth}
        \centering
        \begin{tikzpicture}[%
                scale = 0.5,
                every node/.style={draw,fill=gray!10,circle,minimum size=1pt}]
            \node (one)  at (0,1*1.2){ };
            \node (two)  at (1*1.2,1*1.2){ };
            \node (three)  at (2*1.2,1*1.2){ };
            \node (four)  at (3*1.2,1*1.2){ };

            \node (five)  at (0*1.2,0*1.2){ };
            \node (six)  at (1*1.2,0*1.2){ };
            \node (seven)  at (2*1.2,0){ };
            \node (eight)  at (3*1.2,0){ };

            \node (nine)  at (0,-1*1.2){ };
            \node (ten)  at (1*1.2,-1*1.2){ };
            \node (eleven)  at (2*1.2,-1*1.2){ };
            \node (twelve)  at (3*1.2,-1*1.2){ };

            \draw[<-,blue, thick] (one) -- (two);
            \draw[<-,blue, thick] (five) -- (one);
            \draw[<-,blue, thick] (six) -- (five);
            \draw[<-,blue, thick] (two) -- (six);

            \draw[->,blue, thick] (three) -- (two);
            \draw[->,blue, thick] (four) -- (eight);

            \draw[->,blue, thick] (nine) -- (five);
            \draw[->,blue, thick] (ten) -- (nine);
            \draw[->,blue, thick] (seven) -- (eight);
            \draw[<-,blue, thick] (eleven) -- (twelve);
            \draw[<-,blue, thick] (twelve) -- (eight);
            \draw[<-,blue, thick] (seven) -- (eleven);
        \end{tikzpicture}
        \caption{Oriented CRSF.}
    \end{subfigure}
    \hfill
    \begin{subfigure}[b]{0.25\textwidth}
        \centering
        \begin{tikzpicture}[%
                scale = 0.5,
                every node/.style={draw,fill=gray!10,circle,minimum size=1pt}]
            \node (one)  at (0,1*1.2){ };
            \node (two)  at (1*1.2,1*1.2){ };
            \node (three)  at (2*1.2,1*1.2){ };
            \node (four)  at (3*1.2,1*1.2){ };

            \node (five)  at (0*1.2,0*1.2){ };
            \node (six)  at (1*1.2,0*1.2){ };
            \node (seven)  at (2*1.2,0){ };
            \node (eight)  at (3*1.2,0){ };

            \node (nine)  at (0,-1*1.2){ };
            \node (ten)  at (1*1.2,-1*1.2){ };
            \node (eleven)  at (2*1.2,-1*1.2){ };
            \node[blue] (twelve)  at (3*1.2,-1*1.2){};

            \draw[<-,blue, thick] (five) -- (one);
            \draw[<-,blue, thick] (six) -- (five);
            \draw[->,blue, thick] (two) -- (six);

            \draw[->,blue, thick] (three) -- (two);
            \draw[->,blue, thick] (four) -- (eight);
            \draw[->,blue, thick] (seven) -- (eight);

            \draw[->,blue, thick] (nine) -- (five);
            \draw[->,blue, thick] (ten) -- (nine);
            \draw[->,blue, thick] (eleven) -- (twelve);
            \draw[<-,blue, thick] (twelve) -- (eight);
            \draw[<-,blue, thick] (six) -- (ten);
        \end{tikzpicture}
        \caption{Oriented MTSF.}
    \end{subfigure}
    \hfill
    \begin{subfigure}[b]{0.25\textwidth}
        \centering
        \begin{tikzpicture}[%
                scale = 0.5,
                every node/.style={draw,fill=gray!10,circle,minimum size=1pt}]
            \node[blue] (one)  at (0,1*1.2){ };
            \node (two)  at (1*1.2,1*1.2){ };
            \node (three)  at (2*1.2,1*1.2){ };
            \node (four)  at (3*1.2,1*1.2){ };

            \node (five)  at (0*1.2,0*1.2){ };
            \node (six)  at (1*1.2,0*1.2){ };
            \node[blue] (seven)  at (2*1.2,0){ };
            \node (eight)  at (3*1.2,0){ };

            \node (nine)  at (0,-1*1.2){ };
            \node (ten)  at (1*1.2,-1*1.2){ };
            \node (eleven)  at (2*1.2,-1*1.2){ };
            \node (twelve)  at (3*1.2,-1*1.2){ };

            \draw[->,blue, thick] (five) -- (one);
            \draw[->,blue, thick] (six) -- (five);
            \draw[->,blue, thick] (two) -- (six);
            \draw[->,blue, thick] (three) -- (seven);

            \draw[->,blue, thick] (four) -- (eight);
            \draw[<-,blue, thick] (seven) -- (eight);

            \draw[<-,blue, thick] (ten) -- (nine);
            \draw[->,blue, thick] (eleven) -- (twelve);
            \draw[->,blue, thick] (twelve) -- (eight);
            \draw[<-,blue, thick] (six) -- (ten);
        \end{tikzpicture}
        \caption{Oriented SF.}
    \end{subfigure}
    \caption{
    Different oriented spanning subgraphs of a grid graph: an oriented Cycle-Rooted Spanning Forest (CRSF), an oriented Multi-Type Spanning Forest (MTSF) and an oriented Spanning Forest (SF).
        The root of an oriented tree (in blue) is the node with no out-going edge.}
    \label{fig:graphs}
\end{figure}
\paragraph{Random successor.}
Denote by $d(v)$ the number of neighbors of $v$ in $\mathcal{G}$.
The sampling algorithm relies on a random walk on the nodes of $\mathcal{G}_r$. Drawing from the Markov kernel, a procedure denoted by \textsc{RandomSuccessor}$(v)$, goes as follows: at node $v$, move to the root $r$ with probability $q/(q+d(v))$ or move to another neighbour of $v$ with probability $1/(q+d(v))$.

\paragraph{$\cyclepopping{}$.}
Let $\calC^{o}$ be an oriented subgraph and $w\notin\calC^{o}$ be a node of $\mathcal{G}_r$.
We first define the procedure $P(w,\calC^{o})$, closely following \citet[Section 2]{kassel2017}.
Starting from $w$, perform a random walk thanks to \textsc{RandomSuccessor}, until it reaches its first self-intersection $v$, hits $\calC^{o}$, or reaches the absorbing node $r$.
This trajectory is an oriented branch.
Then,
\begin{itemize}
    \item if $v\in \calC^{o}$ or $v = r$ (the root), add this oriented branch to $\calC^{o}$.
    \item if $v$ is a self-intersection, then the branch contains an oriented cycle $\eta^{o}$. Draw a Bernoulli random variable, independent from the rest, with success probability $\alpha(\eta^{o})\in [0,1]$.
          \begin{itemize}
              \item If the coin toss is successful, add the branch with the oriented cycle to $\calC^{o}$ and stop.
              \item else \emph{pop} the cycle (i.e., remove its edges from the branch) and continue \textsc{RandomSuccessor} until self-intersection, intersection with $\calC^{o}$, or until it reaches the absorbing node $r$ again.
          \end{itemize}
\end{itemize}
The entire algorithm, denoted by $\cyclepopping{}$,
goes as follows:
initialize $\calC^{o}$ as empty and take any $w\in\mathcal{G}$ as starting node to execute $P(w,\calC^{o})$.
If $\calC^{o}$ does not span all the nodes of $\mathcal{G}_r$, repeat $P(w,\calC^{o})$ starting from any $w\notin \calC^{o}$.
The procedure stops when $\calC^{o}$ spans all the nodes.

\paragraph{Post-processing.}
The cycle-popping procedure yields an oriented MTSF with root $r$. As a simple postprocessing step, we obtain an \emph{oriented MTSF} of $\mathcal{G}$ by overriding the extra node $r$ and the links that point towards it.
Each of the nodes formerly connected to $r$ are then the roots of the remaining \emph{oriented MTSF} of $\mathcal{G}$.
In particular, the number of links pointing to $r$ in the underlying oriented MTSF with a root becomes the number of roots of the associated oriented MTSF.

\begin{proposition}[Correctness]
    \label{prop:MTSF_correctness}
    Under \cref{assump:non-singularity} and if \cref{eq:cond_weakly_inconsistent} holds, the algorithm terminates and its output is distributed according to \cref{eq:proba_oriented_MTSF}.
\end{proposition}
{A proof of \cref{prop:MTSF_correctness} is given in \cite{FaBa24} in the spirit of the work of \cite{Marchal99} on uniform spanning trees, which differs from the stacks-of-cards approach of \citep{Wilson96,Kassel15}.}
Note that, if the cycle weights are given by \cref{eq:def_alpha}, the output of the algorithm is an oriented MTSF without $2$-cycles since the latter are popped with probability $1$, and therefore, by forgetting the edge orientations we obtain an MTSF with root $r$. 
If $q=0$, the output is simply a CRSF.

\begin{remark}
    Again, under \cref{assump:non-singularity} and if \cref{eq:cond_weakly_inconsistent} holds, the expected number of steps $T$ to finish MTSF sampling with cycle-popping { is the trace of the inverse of the \emph{normalized} Laplacian, i.e.,
    \begin{equation}
        \mathbb{E}[T] = \Tr\Delta_{\mathrm{N}}^{-1}. \label{eq:ET}
    \end{equation}
    where we defined $\Delta_{\mathrm{N}} = (D+q\mathbb{I})^{-1}(\Delta + q \mathbb{I})$, and $D$ is the diagonal matrix of the degrees with $\deg(v) = \sum_{u} w_{uv}$ and for $q\geq 0$.
For the proof of \eqref{eq:ET}, we refer to the companion paper \citep{FaBa24}, where we actually derive the probability law of $T$.}
A few comments are in order.
\begin{itemize}
    \item {When $q=0$, interestingly, the \emph{normalized} magnetic Laplacian is the relevant object appearing the corresponding Cheeger inequality \citep{Bandeira}, which related the frustration of the connection to the least eigenvalue of $\Delta_{\mathrm{N}}$. 
    Intuitively, we expect the average sampling time \cref{eq:ET} to be large when the least eigenvalue of the normalized magnetic Laplacian $D^{-1}\Delta$ is small, namely when there are only very weak inconsistencies in the connection graph.
    As we mentioned already, in practice, we observe that $\cyclepopping{}$ is faster than HKPV algorithm \citep{HKPV06}; see the simulations of \cref{sec:comp_time_HKPV_cyclepopping}.}
    \item {Another case of interest is when the connection is consistent and $q>0$. 
    In this case, the magnetic Laplacian is unitarily equivalent to the combinatorial Laplacian $\Lambda$ and the expected number of steps reads
    \[
        \mathbb{E}[T] = \Tr\left(D (\Lambda + q \I)^{-1}\right) + \Tr\left(q (\Lambda + q \I)^{-1}\right),
    \]
    where the second term is the expected number of steps to the roots of the spanning forest -- or the expected number of trees -- which is always bounded from above by the number of nodes $n$ \citep{AvGaud2018,PABGAA2020}.
    The first term in the above formula is clearly a decreasing function of $q$.
    An upper bound on $\mathbb{E}[T]$ is $2m/q + n$; see \cite[Section A.1]{thesis_pilavci}.
    As a comparison, the expected number of steps to sample determinantal STs rooted in $r$ thanks to Wilson's algorithms is $\Tr(( D^{-1}\Lambda)_{\hat{r}})^{-1}$, where the subscript $\cdot_{\hat{r}}$ indicates that the $r$-th row and the $r$-th column are removed.
    Thus, we expect the practical sampling time with $\cyclepopping{}$ to be a decreasing function of $q$, as illustrated on a toy graph in \cref{sec:emp_sampling_time} where $\cyclepopping{}$ becomes faster than Wilson's algorithm for sampling STs if $q$ is large enough.}
\end{itemize}

\end{remark}




\subsection{{Strongly inconsistent} graphs with importance sampling} 
\label{sec:montecarlo_with_selfnormalized_importance_sampling}

If there are strongly inconsistent cycles in the connection graph, i.e., if \eqref{eq:cond_weakly_inconsistent} does not hold, then $\cyclepopping{}$  cannot be used to sample from
$
    p(\calC) 
    \propto
    q^{|\rho(\calC)|}
    \prod_{\text{cycles } \eta \in \calC}
    2(1 - \cos \theta(\eta)),
$
since some cycle weights are larger then $1$.
However, we can cap the cycle weights at $1$ and use the cycle-popping algorithm of Section~\ref{sec:sampling_a_multitype_spanning_forest} to sample from the auxiliary distribution
\begin{equation}
    \label{eq:proba_capped}
    \pIS(\calC)
    \propto
    q^{|\rho(\calC)|}
    \prod_{\text{cycles } \eta \in \calC}
    2\big\{1 \wedge \big(1 - \cos \theta(\eta)\big)\big\}.
\end{equation}
Note that \cref{eq:proba_capped} is not necessarily determinantal. We can sample from \cref{eq:proba_capped} with the help of \cyclepopping{}, indeed the proof of correctness remains valid if some cycles have weight $1$.
The effect of capping the cycle weights to unity will be accounted for by considering the reweighted sparsifier
\begin{equation}
    \label{eq:MC_sum_self-normalized}
    \widetilde{\Delta}_{t}^{(\mathrm{IS})}
    = \frac{1}{\sum_{s=1}^t w(\calC^\prime_s)} \sum_{\ell=1}^t w(\calC^\prime_\ell)
    \widetilde{\Delta}(\calC^\prime_\ell), \text{ with  } \calC^\prime_\ell \stackrel{\text{i.i.d.}}{\sim} \pIS \text{ for } 1\leq \ell \leq t.
\end{equation}
We defined the importance weight
\begin{equation}
    w(\calC) 
    = a \frac{p(\calC)}{\pIS(\calC)}
    = a \prod_{\text{cycles } \eta\in \calC}
    \Big\{
    1 \vee \Big(1 - \cos \theta(\eta)\Big)
    \Big\},
    \label{eq:importance_weight}
\end{equation}
where the normalization constant $a$ is there for $w$ to sum to $1$.
Note that the above {calculation} relies on the simple identity $\{1 \wedge x\} \{1 \vee x\} = x$ for all $x\geq 0$, and that we shall not need to know the normalization constant for the weights.

\begin{proposition}\label{prop:asymp_confidence}
    Let $p\in (0,1)$ and let \cref{assump:non-singularity} hold.
    Let $\calC^\prime_1, \calC^\prime_2,\dots, $ be i.i.d.\ random MTSFs with the capped distribution \eqref{eq:proba_capped}, and consider the sequence of matrices $(\widetilde{\Delta}_{t}^{(\mathrm{IS})}{ + q \I_n})_{t\geq 1}$ defined by \eqref{eq:MC_sum_self-normalized}.
    Finally, let $z >0$ be such that $\Pr(\|\bm{u}\|_2 \leq z) = p$ for $\bm{u}\sim \mathcal{N}(0,\I_{n^2})$.
    Then, as $t\rightarrow \infty$,
    \[
       \Pr \left[ -z (\Delta + q \I_n)\preceq (\widetilde{\Delta}_{t}^{(\mathrm{IS})} {+ q \I_n}) - (\Delta {+ q \I_n}) \preceq z (\Delta + q \I_n) \right] \rightarrow 1-p\,.
    \]
\end{proposition}
The proof of \cref{prop:asymp_confidence}, given below, uses standard techniques such as Slutsky's lemma~\citep[Theorem 11.4]{Gut}. \cref{prop:asymp_confidence} provides us with an asymptotic {statistical} guarantee for the reweighted importance sampling procedure \cref{eq:MC_sum_self-normalized}.
\begin{proof}
    We consider the event
    \begin{equation}
        \label{e:event_of_interest}
        \left\|\left(\Delta + q \I_n\right)^{-1/2}\left(
            \sum_{\ell=1}^t
            \frac{w(\calC^\prime_\ell)}{\sum_{s=1}^t w(\calC^\prime_s)}
            \widetilde{\Delta}(\calC_\ell^\prime) - \Delta \right) \left(\Delta + q \I_n\right)^{-1/2}\right\|_{2} \leq z.
    \end{equation}

    To analyse the probability of this event as $t\to \infty$, we first rephrase the problem in a slightly more general way.
    We introduce the $n\times m$ matrix $\Psi = (\Delta + q \I_n)^{-1/2} B^*$. 
    Let $\bmpsi_e$ be the $e$-th column of $\Psi$ and define the following $n\times n$ matrices
    \[ Y(\calC) =\sum_{e\in \calC}\frac{1}{\lev(e)} \bmpsi_e \bmpsi_e^* \text{ and } \bar{Y} = \sum_{e=1}^m \bmpsi_e \bmpsi_e^*,
    \]
    with $\lev(e) = \bmpsi_e^*\bmpsi_e$ for all $e\in [m]$. Here $\bar{Y} = \sum_{\calC\subseteq [m]} p(\calC) Y(\calC)$.
    Using the definition of $Y$ and $\Psi$, inequality \eqref{e:event_of_interest} can be rewritten as
    \begin{equation*}
        \left\| \bar{Y}_{t}^{(\mathrm{IS})} - \bar{Y}\right\|_2 \leq z \quad \text{ with } \bar{Y}_{t}^{(\mathrm{IS})} = \frac{1}{\sum_{s=1}^t w(\calC_s^\prime)} \sum_{\ell=1}^t w(\calC_\ell^\prime)
        Y(\calC_\ell^\prime).
    \end{equation*}
    This bound in operator norm is implied by the following bound for the Frobenius norm
    \begin{equation}
        \left\| \bar{Y}_{t}^{(\mathrm{IS})} - \bar{Y}\right\|_F \leq z, \label{eq:diff_Fr}
    \end{equation}
    the probability of which we shall now investigate.

    \paragraph*{Vectorization and writing a {Central Limit Theorem}.} For a matrix $A$, denote by $\vec A$ the vector obtained by stacking the columns of $A$ on top of each other.
    Let $\bm{y}(\calC) = \vec(Y(\calC))$ and $\bar{\bm{y}}   = \vec(\bar{Y})$.
    We aim to approximate the expectation
    $
        \bar{\bm{y}} = \sum_{\calC\subseteq [m]} p(\calC) \bm{y}(\calC),
    $
    since we cannot sample from $p(\calC)$.
    The (self-normalized) importance sampling estimator of this expectation reads
    \begin{equation*}
        \bar{\bm{y}}_{t}^{(\mathrm{IS})}
        =  \sum_{\ell=1}^t \frac{w(\calC_\ell^\prime)}{\sum_{s=1}^t w(\calC_s^\prime)}
        \bm{y}(\calC_\ell^\prime) = \vec\left(\bar{Y}_{t}^{(\mathrm{IS})}\right).
    \end{equation*}
    Since
    $\| \bar{\bm{y}}_{t}^{(\mathrm{IS})} - \bar{\bm{y}}\|_F = \| \bar{Y}_{t}^{(\mathrm{IS})} - \bar{Y} \|_F$,
    it is enough to control the Frobenius norm of
    \[
        \bar{\bm{y}}_{t}^{(\mathrm{IS})} - \bar{\bm{y}}
        = \frac{
            \frac{1}{t} \sum_{\ell=1}^t w(\calC_\ell^\prime)
            \Big(\bm{y}(\calC_\ell^\prime) - \bar{\bm{y}}\Big)/a
                }
                {
                \frac{1}{t} \sum_{\ell=1}^t w(\calC_\ell^\prime)/a
                },
    \]
    to be able to establish the bound \cref{eq:diff_Fr}.
    By the law of large numbers,  the limit of the denominator is
    \begin{equation}
            \lim_{t\to +\infty}\frac{1}{t} \sum_{\ell=1}^t w(\calC_\ell^\prime)/a = \sum_{\calC^\prime} \pIS(\calC^\prime) w(\calC^\prime)/a = 1, \label{eq:asymptotic_a}
    \end{equation}
    almost surely.
    Define the properly normalized random vector
    \[
        \bm{x}(\calC^\prime) = w(\calC^\prime)
        \Big(\bm{y}(\calC^\prime) - \bar{\bm{y}}\Big)/a,
    \]
    which has expectation zero w.r.t.\ \cref{eq:proba_capped}, i.e., $\sum_{\calC^\prime} \pIS(\calC^\prime)\bmx(\calC^\prime) = 0$.
    Moreover,
    \[
        \frac{1}{t} \sum_{\ell=1}^t w(\calC_\ell^\prime)
            \Big(\bm{y}(\calC_\ell^\prime) - \bar{\bm{y}}\Big)/a = \frac{1}{t} \sum_{\ell=1}^t \bm{x}(\calC_\ell^\prime) \triangleq \bar{\bmx}_t.
    \]
    Equivalently, the following equality holds
    \[
        \bar{\bm{y}}_{t}^{(\mathrm{IS})} - \bar{\bm{y}}
        = \frac{\bar{\bmx}_t}
                {
                \frac{1}{t} \sum_{\ell=1}^t w(\calC_\ell^\prime)/a
                },
    \]
    and implies an equality between  squared norms
    \begin{equation}
        \| \bar{\bmx}_t \|_F^2 = \left(\frac{1}{t} \sum_{\ell=1}^t w(\calC_\ell^\prime)/a\right)^2 \|\bar{\bm{y}}_{t}^{(\mathrm{IS})} - \bar{\bm{y}}\|_F^2 = \left(\frac{1}{t} \sum_{\ell=1}^t w(\calC_\ell^\prime)/a\right)^2 \| \bar{Y}_{t}^{(\mathrm{IS})} - \bar{Y} \|_F^2.\label{eq:x_vs_Y}
    \end{equation}
    By using Slutsky's lemma and {the classical Central Limit Theorem}, we have the convergence in distribution
    \[
        \sqrt{t}\left(\bar{\bm{y}}_{t}^{(\mathrm{IS})} - \bar{\bm{y}} \right) \Rightarrow \mathcal{N}(0,\Sigma),
    \]
    where $\Sigma = \cov(\bm{x})$; see~\cite{Geweke1989} for an univariate treatment.
    The covariance of $\bmx$ is
    \begin{equation}
        \cov(\bm{x}) = \frac{1}{a^2} \sum_{\calC^\prime}  \pIS(\calC^\prime)  w(\calC^\prime)^2
    \Big(\bm{y}(\calC^\prime) - \bar{\bm{y}}\Big)\Big(\bm{y}(\calC^\prime) - \bar{\bm{y}}\Big)^* =  \sum_{\calC^\prime}  p(\calC^\prime)  \frac{w(\calC^\prime)}{a}
    \Big(\bm{y}(\calC^\prime) - \bar{\bm{y}}\Big)\Big(\bm{y}(\calC^\prime) - \bar{\bm{y}}\Big)^*.\label{eq:cov}
    \end{equation}

    To obtain an empirical estimate of this covariance, we first remark that $a  = \sum_{\calC^\prime}\pIS (\calC^\prime) w(\calC^\prime)$ and thus an estimator of $a$ is
    \[
        \widehat{a} = \frac{1}{t}\sum_{\ell = 1}^t w(\calC_\ell^\prime)
    \]
    for $\calC_\ell^\prime $ i.i.d.\ w.r.t.\ \cref{eq:proba_capped}.
    Thus, we build an estimator of the covariance by using the first identity in \cref{eq:cov}, as follows
    \begin{equation}
        \widehat{\cov}_t(\bm{x}) = \frac{
            \frac{1}{t}\sum_{\ell = 1}^{t}  \  w(\calC_\ell^\prime)^2
    \Big(\bm{y}(\calC_\ell^\prime) - \bar{\bm{y}}\Big)\Big(\bm{y}(\calC_\ell^\prime) - \bar{\bm{y}}\Big)^*}
    {\left(\frac{1}{t}\sum_{\ell = 1}^t w(\calC_\ell^\prime)\right)^2} .\label{eq:cov_estimate}
    \end{equation}

    Then, by Slutsky's theorem again, we obtain the asymptotic {Central Limit Theorem}
    \begin{equation*}
        \sqrt{t} \left(\widehat{\cov}_t(\bm{x})\right)^{-1/2}\bar{\bmx}_t  \Rightarrow \mathcal{N}(0,\I_{n^2}).
    \end{equation*}
    The rest of the proof consists in bounding the empirical covariance matrix and writing the corresponding confidence ball.

    \paragraph*{An upper bound on the empirical covariance matrix.} Note that, since $\|\bmpsi_e \bmpsi_e^*\|_F = \lev(e) \leq 1$, we have thanks to a triangle inequality
    \begin{equation}
        \|\bm{y}(\calC) - \bar{\bm{y}}\|_F = \| Y(\calC) - \bar{Y} \|_F\leq \sum_{e\in \calC}\frac{1}{\lev(e)} \|\bmpsi_e \bmpsi_e^*\|_F + \sum_{e=1}^m \|\bmpsi_e \bmpsi_e^*\|_F \leq m  + \deff, \label{eq:bound_diff_y}
    \end{equation}
    where $\deff = \Tr(\Psi^*\Psi)$.
    Since, in light of \cref{eq:bound_diff_y}, we have
    \[
    \Big\| \Big(\bm{y}(\calC_\ell^\prime) - \bar{\bm{y}}\Big)\Big(\bm{y}(\calC_\ell^\prime) - \bar{\bm{y}}\Big)^*\Big\|_2 = \|\bm{y}(\calC_\ell^\prime) - \bar{\bm{y}}\|_F^2 \leq ( m + \deff )^2,
    \]
    an upper bound on this covariance matrix reads
    \begin{equation}
        \widehat{\cov}_t(\bm{x}) \preceq
        \frac{
            \frac{1}{t}\sum_{\ell = 1}^{t}  \  w(\calC_\ell^\prime)^2}
    {\left(\frac{1}{t}\sum_{\ell = 1}^t w(\calC_\ell^\prime)\right)^2} ( m + \deff )^2 \I \triangleq \omega_t \I, \label{eq:upper_bound_emp_cov}
    \end{equation}
    where we used the triangle inequality.
    \paragraph*{Confidence ball.}
    We would like to establish an asymptotic confidence interval. Suppose
    \[
        \|\sqrt{t} \left(\widehat{\cov}_t(\bm{x})\right)^{-1/2}\bar{\bmx}_t\|_F^2 \leq z^2,
    \]
    for some $z >0$.
    By using the upper bound on the empirical covariance matrix given in \cref{eq:upper_bound_emp_cov}, we find
    $
    \|\bar{\bmx}_t\|_F^2 \leq \frac{z^2 \omega_t}{t},
    $
    By using \cref{eq:x_vs_Y}, this bound gives
    \[
        \| \bar{Y}_{t}^{(\mathrm{IS})} - \bar{Y} \|_F \leq z \sqrt{\frac{ \omega_t}{t}}\left(\frac{1}{t} \sum_{\ell=1}^t w(\calC_\ell^\prime)/a\right),
    \]
    where the last factor on the right-hand side tends to $1$ asymptotically, see \cref{eq:asymptotic_a}.
    Asymptotically in $t$ and with confidence level $p$, we have the inequality
    \[
        \| \bar{Y}_{t}^{(\mathrm{IS})} - \bar{Y} \|_F \leq z \sqrt{\frac{ \omega_t}{t}},
    \]
    where $z$ is chosen such that $\Pr(\|\bm{u}\|_2 \leq z) = p$ for $\bm{u}\sim \mathcal{N}(0,\I_{n^2})$.
    This proves \cref{prop:asymp_confidence}.
\end{proof}
\section{Empirical results} 
\label{sec:empirical_results}
 After describing in \cref{sub:settings_and_baselines}  settings which are shared in most of our simulations, we provide a few illustrations of the applicability of the approximations given above.
 In \cref{sub:magnetic_laplacian_sparsification_and_ranking}, we consider the unsupervised angular synchronization problem and its application to ranking from pairwise comparisons.
Next, Laplacian preconditioning is discussed in \cref{sub:laplacian_preconditioning}.
The code to reproduce the experiments is freely available on GitHub.\footnote{
    \url{https://github.com/For-a-few-DPPs-more/MagneticLaplacianSparsifier.jl}
}
{Here are the key takeaways.}
\begin{itemize}
    \item {\emph{Angular synchronization performance}. The least eigenvector of the sparsifiers obtained with STs and CRSFs approximates well the least eigenvector of the full magnetic Laplacian.
    However, the difference in terms of ranking recovery is at best marginal compared with i.i.d.\ sparsifiers in the Erd\H{o}s-Rényi graphs that we considered.}
    \item {\emph{Preconditioning performance}. In all our simulations, i.i.d.\ edge sampling is outperformed by ST or MTSF sampling for constructing a preconditioner of the regularized Laplacian.
    In particular, in our simulations, CRSF-based sparsifiers allow to approximate well the unregularized magnetic Laplacian with a batch of as few as $2$ or $3$ CRSFs.
    These sparsifiers are also invertible by construction.
    It seems that the main shortcoming of a sparsifier based on i.i.d.\ edge sampling is that the resulting subgraph is not connected by design.}

    {Sparsifiers based on STs also perform well except when connection graphs contain only a few inconsistencies, namely, for example in the case of the Erd\H{o}s-Rényi Outlier (ERO) model considered here. 
    Contrary to the CRSF case, Laplacians obtained with batches of STs are not guaranteed to be nonsingular.}
    \item {\emph{Role of leverage score estimation}.
    Our simulations indicate that,  for graphs with a concentrated degree distribution, leverage score estimation is not crucial. The proposed uniform leverage scores heuristics works well in our simulations and yields sparsifiers which can efficiently reduce the Laplacian condition number at a low computational cost. Although the bias of the latter approach is manifest in our simulations, it seems to be a satisfying computational proxy.}
    \item {\emph{Timing}. For many small-scale graphs considered here, we observed that sampling STs with Wilson's algorithm is faster than sampling CRSFs or MTSFs with $\cyclepopping{}$, except in special cases such as graphs with bottlenecks, e.g., the so-called Barbell graph. 
    For real-world larger graphs, $\cyclepopping{}$ and Wilson's algorithm have similar sampling times, at least in the case where the root of Wilson's algorithm is sampled uniformly.
    An important computational cost -- and often the dominant cost -- is the estimation of leverage scores. As we already mentioned, leverage scores can be taken to be uniform without too big a compromise on preconditioning performance.}
\end{itemize}
\subsection{Settings and baselines}
\label{sub:settings_and_baselines}
In our numerical simulations, we generate random connection graphs with a controlled amount of noise, so that we can vary the level of consistency of the cycles.
\subsubsection{Random graphs and random connection graphs} \label{sec:random_graphs}


An Erd\H{o}s-Rényi graph $\mathrm{ER}(n,p)$ is a random graph with $n$ nodes, defined by independently adding an edge between each unordered pair of nodes with probability $p$.
{This simple type of random graph will be endowed with a random $\Uone$-connection to understand the role of the connection inconsistencies in graphs without clear bottlenecks or community structure.}

\paragraph{Noisy comparison models and random connection graphs.}
Our $\Uone$-connection graphs are generalizations of Erd\H{o}s-Rényi graphs with a planted trivial connection -- namely with consistent pairwise comparisons of the type $h_u-h_v$ for each edge $uv$ -- which is corrupted by noise.
Two statistical models for noisy pairwise comparisons are used: the $\mathrm{MUN}$ and $\mathrm{ERO}$ models.
These models are defined given a fixed ranking score vector $\bm{h}$ such that $h_u$ is the ranking score of $u\in \calV$.
Here we choose $\bm{h}$ as a (uniform) random permutation of $[1, \dots, n]$.
\begin{itemize}
    \item \textbf{Multiplicative noise affects all edges.} Following~\cite{Cucuringu16}, we define the Multiplicative Uniform Noise model, denoted by $\mathrm{MUN}(n,p,\eta)${, where each pairwise comparison is independently corrupted by a multiplicative noise.}
    More precisely, with probability $p$, and independently from other edges, there is an edge $e=uv$ with $1\leq u<v\leq n$ coming with an angle $\vartheta(uv)=(h_u-h_v) (1+\eta \epsilon_{uv})/(\pi(n-1))$ where $\epsilon_{uv}\sim \mathcal{U}([0,1])$ are independent noise variables. Then, $\vartheta(vu) \triangleq -\vartheta(uv)$.
    \item \textbf{Pure noise affects only a few edges.} Another noise model used here and defined by~\citet{Cucuringu16} is the Erd\H{o}s-R\'enyi Outliers (ERO) model, denoted by $\mathrm{ERO}(n,p,\eta)${, where -- roughly speaking -- each planted pairwise comparison has a probability to be replaced by a noise drawn from the uniform distribution.}
    For all pairs $uv$ with $1\leq u<v\leq n$, the corresponding edge $e = uv$ is added with probability $p$, independently again from other edges.
    With probability $1-\eta$, this new edge comes with an angle $\vartheta(uv) = (h_u-h_v)/(\pi(n-1))$, otherwise the angle is set to $\vartheta(uv)=\epsilon_{uv}/(\pi(n-1))$, where $\epsilon_{uv}$ is drawn from the discrete uniform distribution on $\{-n+1, \dots, n-1\}$.
    Again, $\vartheta(vu) \triangleq -\vartheta(uv)$.
\end{itemize}

\subsubsection{Random connection graphs from real networks} \label{sec:random_graph_real_nets}
{In the same way as random connections are defined on Erd\H{o}s-Rényi graphs in \cref{sec:random_graphs}, we consider here two real networks which are endowed with random connections. 
The motivation behind this is to design connection graphs with a more complex degree distribution and topology than ER graphs.}
\begin{itemize}
    \item {Epinions\footnote{Available at \url{https://snap.stanford.edu/data/soc-Epinions1.html}.} \citep{richardson2003trust}, a who-trusts-whom online social network (Epinions.com) with $n = 75,869$ nodes, $m= 405,057$ edges and considered as undirected.}
    \item {Stanford\footnote{Available at \url{https://snap.stanford.edu/data/web-Stanford.html}.} \citep{leskovec2009community}, the network of web pages of Stanford university. We consider its largest connected component as  an undirected graph with $n = 255,265$ nodes and  $m=1,941,926$ edges.}
\end{itemize}

{In exact analogy with the MUN and ERO models, we define two types of random connection graphs: (i) Epinions-MUN$(\eta)$ and Stanford-MUN$(\eta)$ for which the noise is multiplicative and uniform, (ii) Epinions-O$(\eta)$ and Stanford-O$(\eta)$ for which an edge is independently endowed with a random (uniform) complex phase with probability $\eta$.}
\subsubsection{Methods and Baselines} \label{sec:methods_and_baselines}

In our simulations, we compute sparsifiers of the form $\frac{1}{t} \sum_{\ell=1}^t\widetilde{\Delta}(\calC_\ell) = B^* S^{(b)} S^{(b)\top}B$ for $t\geq 1$ batches of samples; with the exception of CRSFs and MTSFs, for which we use the self-normalized formula in \cref{eq:MC_sum_self-normalized}.

{\paragraph*{Connectivity.}
To assess if a sparsifier associated with a batch of edges corresponds to a connected graph we simply define the connectivity as the function which gives $1$ if the graph is connected and gives $0$ otherwise.
Then, in our simulations, we often report the average connectivity over several independent runs, valued in $[0,1]$.}

\paragraph{Leverage score approximation.}
The sparsification guarantees presented in Section~\ref{sec:sparsification_of_the_regularized_magnetic_laplacian} implicitly assume that the sampling matrix \eqref{eq:sampling_matrix} can be calculated exactly.
However, evaluating or even approximating leverage scores is far from trivial.
This is a known issue even for i.i.d.\ sampling and in the case of the combinatorial Laplacian; see~\citet{SpielmanSrivastava} or~\cite{DKPRS2017} for a recent approach.
Hence, we consider two options for the leverage scores in \eqref{eq:sampling_matrix}.

{First}, we use the standard ``uniform leverage scores" heuristic: for each batch $\{\calC_\ell\}_{1\leq \ell \leq t}$, we weigh each edge $e\in \calC_\ell$ in \cref{eq:sampling_matrix} by
    \begin{equation}
        \widetilde{\lev}(e) = \lev_{\rm unif}(e)  \text{ with }\lev_{\rm unif}(e) = |\calC_\ell| / m.\label{eq:uniform_heuristics}
    \end{equation}
    This approach for sparsification with \cref{eq:uniform_heuristics} is called ``DPP(K) unif'' in plots. 
    The choice of normalizing constant can be explained as follows: the exact leverage scores are normalized so that their sum is the expected sample size $\sum_{e=1}^m \lev(e) = \E_{\calC \sim \DPP(K)}[|\calC|]$; see \cref{sec:dpp_sampling_of_edges__multitype_spanning_forests}. Similarly, since $\calC_\ell \sim \DPP(K)$, in this paper, we choose the normalization so that $\sum_{e=1}^m \lev_{\rm unif}(e) = |\calC_\ell|$.
    This heuristic is expected to be efficient whenever the connection graph is well connected. It is closely related to the heuristic used for preconditioning linear sytems in the context of kernel ridge regression in \cite{Rudi2017}.
    
    {Second, we leverage the Johnson-Lindenstrauss lemma to sketch LSs with the help of Rademacher random matrices, as advised in \citep[Section 4]{SpielmanSrivastava}.}
    {A simple adaptation is necessary to deal with the case $q\neq 0$, which amounts to consider an augmented graph where one extra \emph{auxiliary edge} is added to each node; see \citep{pilavci2020}.
    The augmented leverage scores associated with this augmented graph are the diagonal entries of the orthogonal projector onto the column space of the augmented incidence matrix $$ \begin{bmatrix}\sqrt{q} \cdot \I_{n\times n} \\ B\end{bmatrix}.$$
    By design, the desired leverage scores are the last $m$ entries of the augmented leverage scores.
    Thus, for sketching LSs in this case, we take 
    \begin{equation}
        k = \lceil 40 \log (m +n) + 1 \rceil \label{eq:k_for_JL}
    \end{equation}
    and define a random  matrix $Q$ -- filled with $\pm 1/\sqrt{k}$ independent Rademacher random variables -- of size a $(m + n)\times k$.
    Denote by $T$ the solution of 
    \begin{equation}
        (\Delta + q \I) T =  
        \begin{bmatrix}
            \sqrt{q} \cdot \I_{n\times n} & B^*
        \end{bmatrix}Q,\label{eq:linear_system_sketch}
    \end{equation}
    so that only $\mathcal{O}(\log (m+n))$ linear systems are solved.
    In this case, we define the sketched LSs by  
    \begin{equation}
        \widetilde{\lev}(e) = \lev_{\mathrm{JL}}(e) \triangleq \|\bmdelta_e^\top B T\|^2_2.\label{eq:lev_sketch}
    \end{equation}
    When $q=0$, we rather take $k = \lceil 40 \log (m) + 1 \rceil$ and the right-hand side of \cref{eq:linear_system_sketch} simply reads $B^*Q$. }

\paragraph{Baselines with i.i.d.\ sampling and with spanning trees.}
As a first type of baseline, we use sparsifiers obtained with edges sampled i.i.d. proportionally to the approximate leverage scores~\citep{SpielmanSrivastava}, computed using the Johnson-Lindenstrauss lemma -- this approach is referred to as ``i.i.d.\ JL-LS''.
For a fair comparison of these i.i.d.\ baselines with a batch of MTSFs, we sample $t$ independent copies $\calC_\ell$, $1\leq \ell \leq t$, where each $\calC_\ell$ is an i.i.d.\ sample of edges.
Next, for a given $q\geq 0$, we compute the average sparsifier $\frac{1}{t} \sum_{\ell=1}^t\widetilde{\Delta}_{\calC_\ell} + q \I_n$.

As a second baseline, we compute sparsifiers obtained by sampling batches of uniform spanning trees (STs) with Wilson's algorithm.
The corresponding edges are then taken with their corresponding complex phases, although the latter are not used for sampling the spanning trees. 
The sparsifiers are then built using the formula \cref{eq:sampling_matrix} with our spanning trees.
This yields the baseline ``ST JL-LS'', where leverage scores are again approximated using the Johnson-Lindenstrauss lemma.


\paragraph{Sampling strategy and self-normalized importance sampling.} To sample CRSFs and SFs, we use the algorithm described in~\cref{sec:sampling_a_multitype_spanning_forest}, whereas we use Wilson's algorithm~\citep{Wilson96} to sample STs.
As a common feature, the starting node of the cycle popping random walk is chosen uniformly at random.
For sampling STs, we first sample an oriented rooted spanning tree with a root node sampled uniformly at random, and then, we forget the root and the orientation of this spanning tree.

For sampling MTSFs or CRSFs in the graphs considered here, since we do not know if the cycles are weakly inconsistent, we use the importance sampling distribution \cref{eq:proba_capped} with capped cycle weights, and use the self-normalized formula in \cref{eq:MC_sum_self-normalized}.
By an abuse of notation, these methods are denoted as ``DPP(K)''.

\paragraph{Sparsification with batches of edges.} In the simulations where we compare the accuracy of sparsifiers for a different batchsize, each sparsifier is drawned independently from all the others.
In what follows, we often plot the performance of a sparsification method for a given batchsize and the $x$-axis then reports {the ratio of the number of edges to the number of nodes}.
Note that since the number of edges in an MTSF is random, horizontal error bars appear when we average over independent realizations.
In our experiments, these error bars are however very small.

\subsubsection{{Hardware and timing}} \label{sec:hardware}
{All the simulations were performed on a laptop with a $1.1$ GHz Dual-Core processor and $8$ GB RAM.
Timings are reported in order to provide a rough estimation of the compute time.
The Julia code was executed in a Jupyter notebook where timing was measured using the $@$timed function.
}
\subsection{Magnetic Laplacian sparsification and ranking} 
\label{sub:magnetic_laplacian_sparsification_and_ranking}

One of our motivations was angular synchronization; see Section~\ref{sec:magnetic_laplacian}.
One application of angular synchronization is ranking from pairwise comparisons~\citep[Sync-Rank]{Cucuringu16}, which we use here as a case study to illustrate the interest of sparsifying the magnetic Laplacian.
More specifically, we consider Sync-Rank, an algorithm for ranking items from pairwise comparisons, which has favorable robustness properties w.r.t.\ corruptions in the comparisons~\citep{Cucuringu16,Yu12}.
The optimization objective of Sync-Rank is directly inspired by the angular synchronization problem as we explain below.

\subsubsection{Angular synchronization and spectral relaxation.}

The Sync-Rank algorithm of \cite{Cucuringu16} starts from comparisons between neighbouring nodes, where $uv\in\calE$ can carries either a cardinal comparison $\kappa_{uv} \in [-n-1, n+1]$ or an ordinal comparison $\kappa_{uv}\in\{-1,+1\}$.
A positive $\kappa_{uv}$ is interpreted as $u$ being superior to $v$.

\begin{enumerate}
    \item We define the following angular embedding of the comparisons
    $
    \vartheta(uv) = \pi \kappa_{uv} /(n-1),
    $
    for all oriented edges $uv$ in the graph.
    \item Spectral Sync-Rank simply solves the spectral problem \cref{eq:spectral_syncrank}, with the connection graph obtained at step $1$. To account for a non-uniform degree distribution, \citet[Eq (12) and (13)]{Cucuringu16} recommends weights $w_{uv} = 1/\sqrt{d(u)d(v)}$, where $d(u)$ denotes the degree of $u$ in the (unweighted) connection graph. This type of normalization also typically improves graph clustering performance, and we use this weighting here in \eqref{eq:spectral_syncrank}.
    The output of this stage is an angular score $\hat{h}_u \in [0, 2\pi)$ for all $u\in \calV$, as defined above, by taking $\hat{h}_u  = \Arg(f(u))$.
    \item We now select a ranking from the angular scores.
    We find a permutation $\varsigma$ such that $\hat{h}_{\varsigma(1)} \geq \dots \geq \hat{h}_{\varsigma(n)}$.
    Let $\bm{\varsigma} = [\varsigma(1), \dots, \varsigma(n)]^\top$. Let $\bm{r} = [r_1, \dots, r_n]^\top$ be the inverse permutation of $\bm{\varsigma}$, so that $r_u \in\{1, \dots, n\}$ is an integer giving the induced ranking of the node $u\in \calV$.
    The ranking of the $n$ items is then obtained by looking for the circular {shift} of $r_1, \dots, r_n$ minimizing the number of upsets,
    \[
        \sigma_\star = \Arg \min_{\sigma \text{circular}} \sum_{\text{oriented edge }uv} | \sign(\kappa_{uv}) - \sign(r_{\sigma(u)} - r_{\sigma(v)})|;
    \]
    see~\citet[Alg.\ 1]{Cucuringu16}. Explicitly, a circular {shift} on integers in $[n]$ depends on an integer $s$ and is defined as $\sigma(\ell) =1 + (\ell  + s) \mod n$, with $s\in \{0, \dots, n-1\}$. The output is $\sigma_\star$.
\end{enumerate}
Note that we have no statistical guarantee on recovering the ranking.

\subsubsection{Empirical results for the sparsification of Sync-Rank.}\label{sec:exp_SyncRank}

\cref{fig:approx_mg_laplacian} displays the performance of the sparsify-and-eigensolve algorithms described in Section~\ref{sec:methods_and_baselines} on the two random graphs described in Section~\ref{sec:random_graphs}.
We plot two performance metrics, as a function of the total number of edges in each batch divided by the total number of nodes.
Each point on the $x$-axis corresponds to a sparsifier obtained with a  batchsize $t$ ranging from $1$ up to $6$. 
Thus, the $x$-axis reports the number of edges in $\cup_{\ell = 1}^t \calC_\ell$ over $n$.
The first row of \cref{fig:approx_mg_laplacian} reports the least eigenvector approximation accuracy, i.e., the distance $1-|\tilde{\bmf}_1^* \bmf_1|$ between the line of $\tilde{\bmf}_1$ and the line of $\bmf_1$.
Note that the eigenvectors are normalized and computed exactly.
The second row reports Kendall's tau coefficient between the recovered and exact synthetic ranking, a classical measure of accuracy in ranking recovery.
To assess the number of cycles captured by CRSF sampling, the third row displays the average number of sampled cycle rooted trees.
In the right column, we take an $\mathrm{ERO}(n,p,\eta)$ random graph and the left column corresponds to a $\mathrm{MUN}(n,p,\eta)$ graph.
The parameters are $n=2000$, $p=0.01$, $\eta = 0.1$.
The sampling is repeated $3$ times for a fixed connection graph.
We display the mean, whereas the error bars are plus/minus one standard deviation.
\begin{figure}[h!]
    \centering
    \begin{subfigure}[b]{0.49\textwidth}
        \centering
        \includegraphics[scale=0.4]{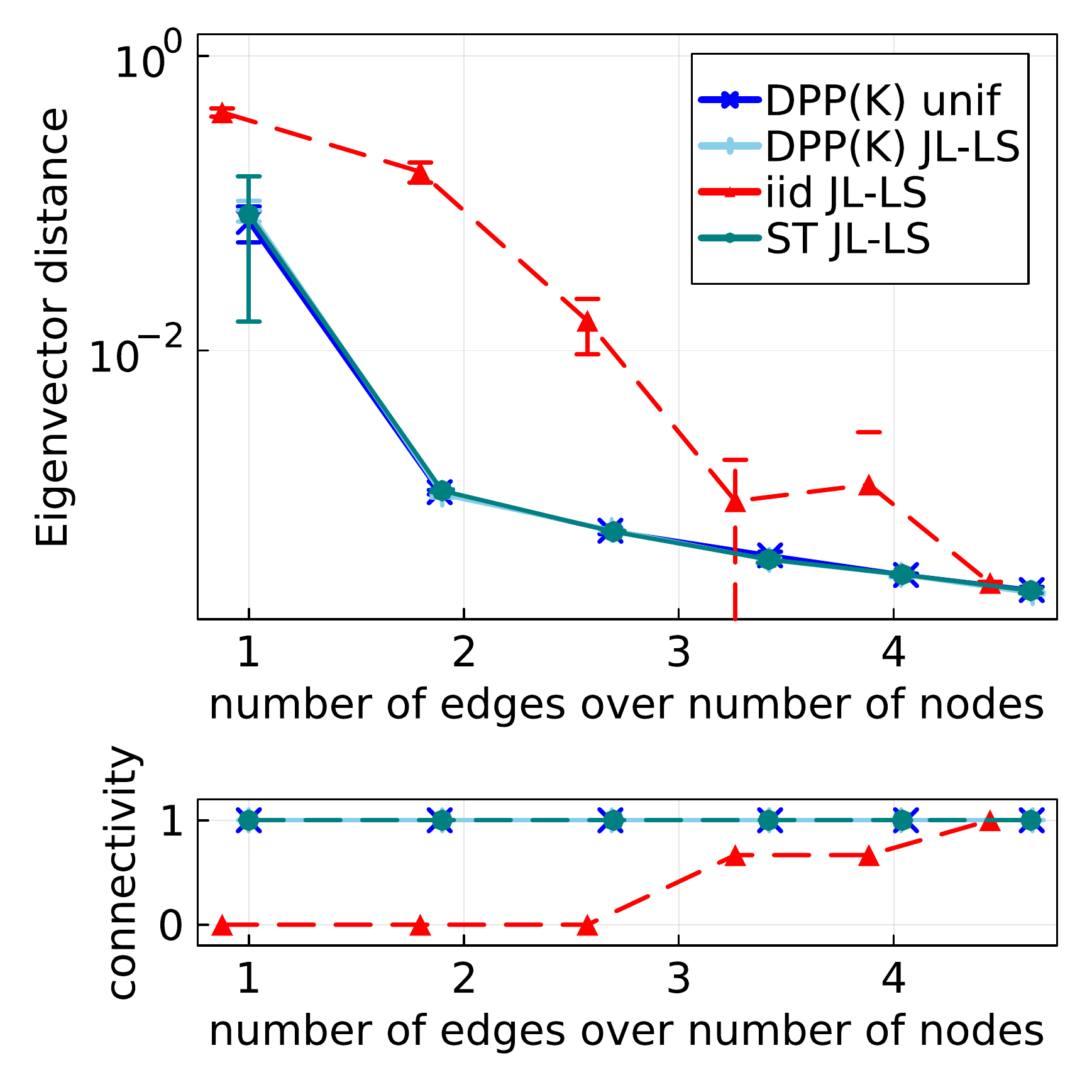}
        \caption{$1-|\tilde{f}_1^* f_1|$ (MUN).}
    \end{subfigure}
    \hfill
    \begin{subfigure}[b]{0.49\textwidth}
        \centering
        \includegraphics[scale=0.4]{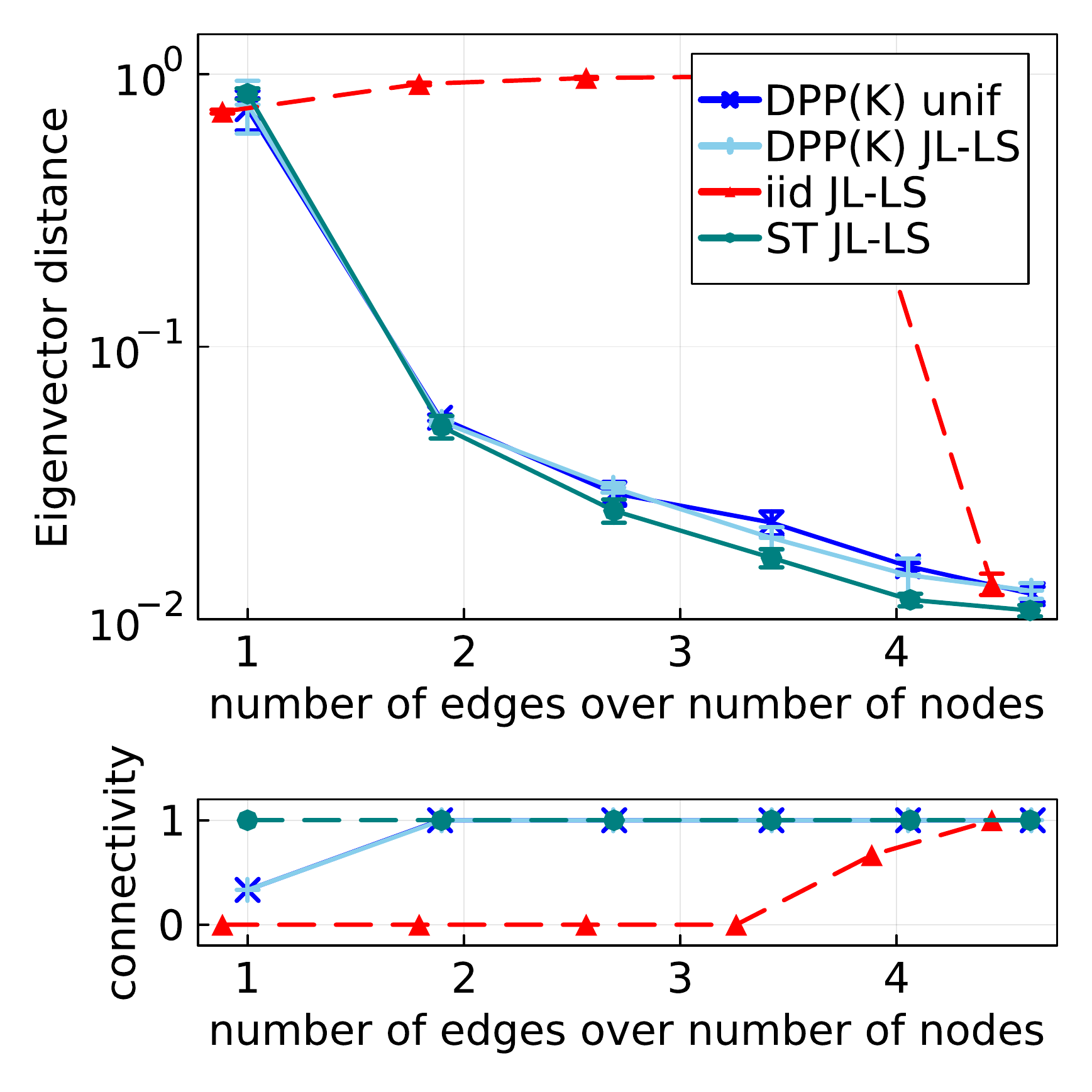}
        \caption{$1-|\tilde{f}_1^* f_1|$ (ERO).}
    \end{subfigure}
    \begin{subfigure}[b]{0.49\textwidth}
        \centering
        \includegraphics[scale=0.4]{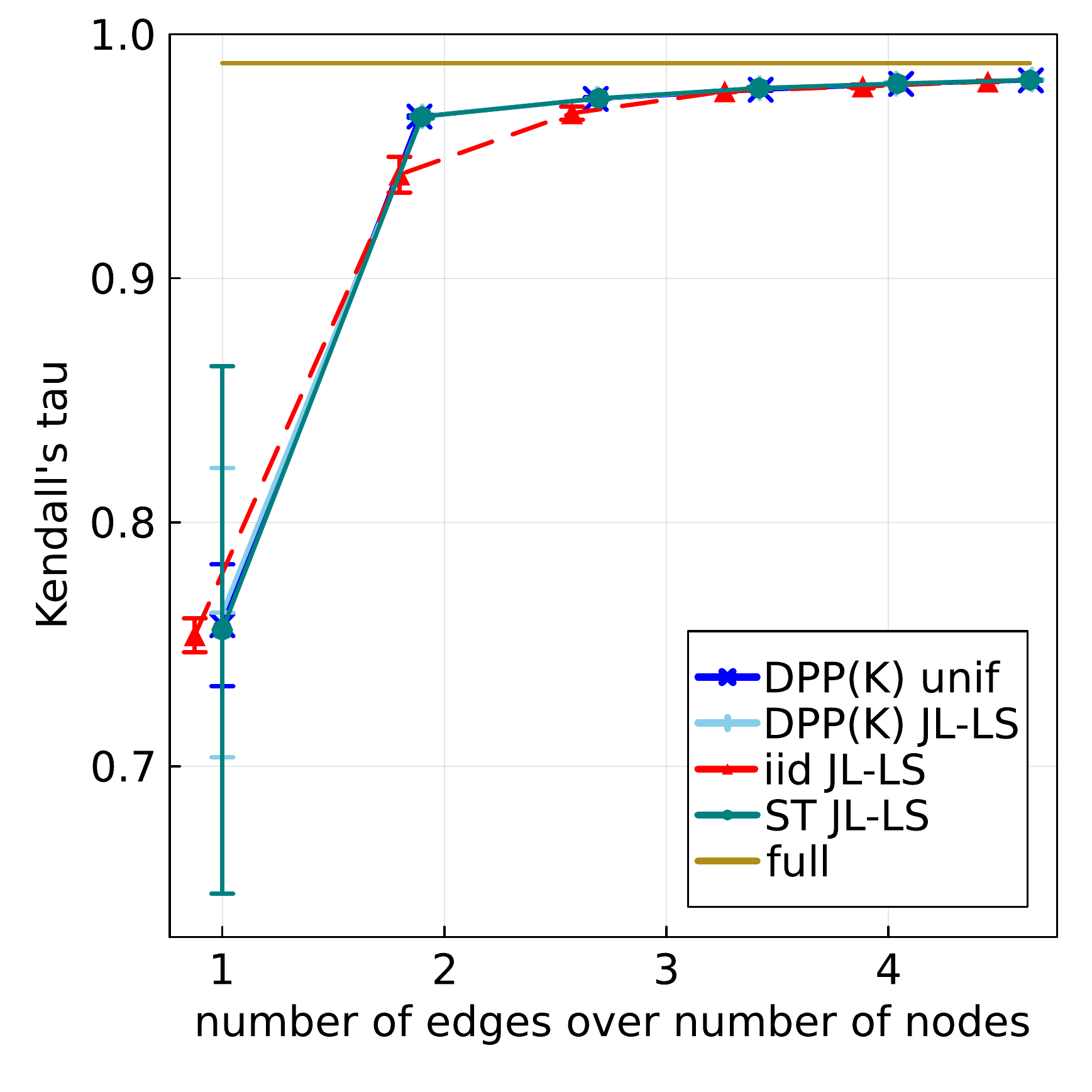}
        \caption{Ranking recovery (MUN). \label{fig:approx_mg_laplacian_MUN_ranking}}
    \end{subfigure}
    \hfill
    \begin{subfigure}[b]{0.49\textwidth}
        \centering
        \includegraphics[scale=0.4]{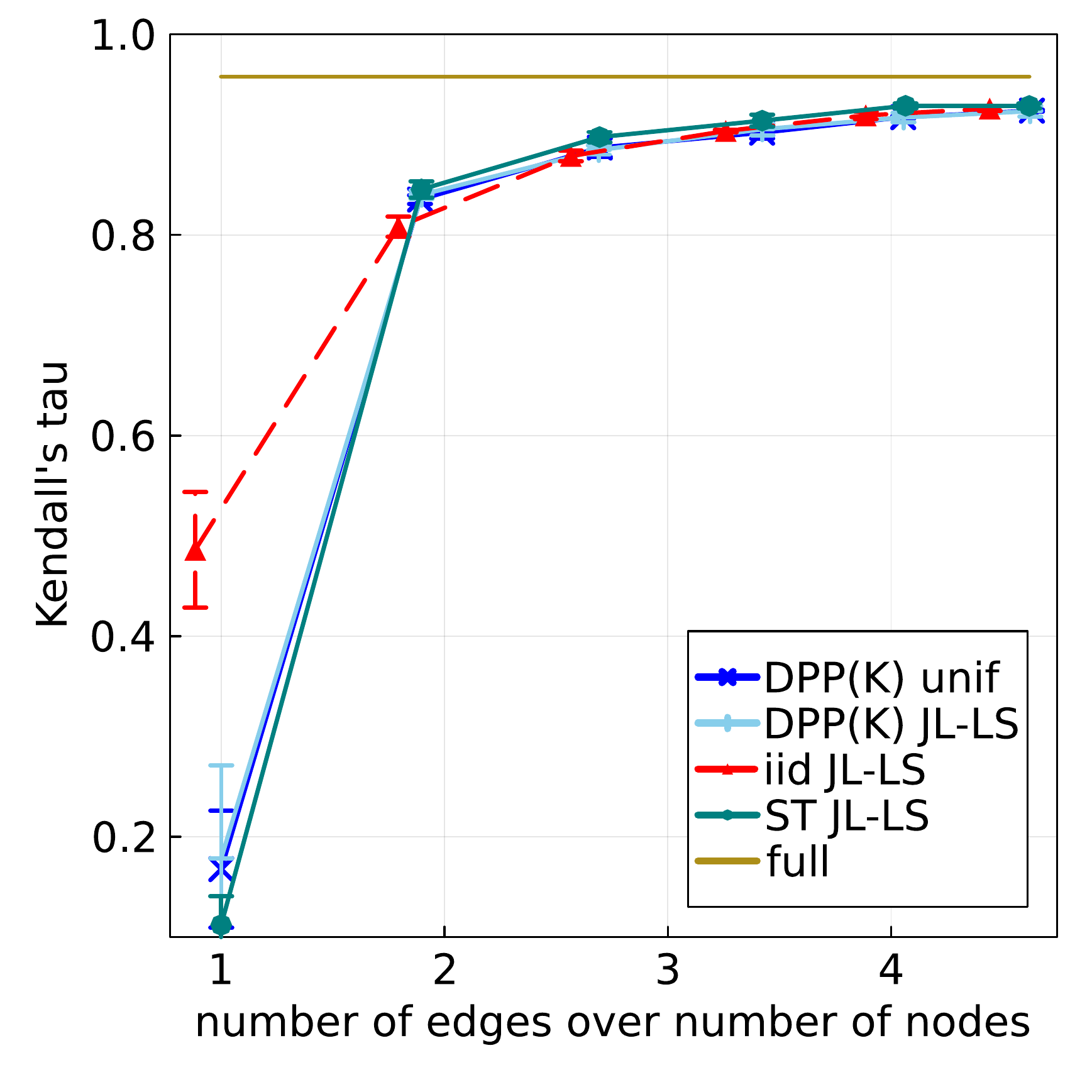}
        \caption{Ranking recovery (ERO). \label{fig:approx_mg_laplacian_ERO_ranking}}
    \end{subfigure}
    \caption{
        {Sparsify-and-eigensolve the magnetic Laplacian of MUN and ERO graphs. On the top row, we display the distance between the top eigenvectors of the magnetic Laplacian and its sparsifier. 
        On the bottom row, we report the ranking recovery with Sync-Rank based on the sparsifier, in terms of Kendall's tau distance to the planted ranking.}
        \label{fig:approx_mg_laplacian}
        }
\end{figure}
As expected, all algorithms reach the same performance with enough edges, whatever the metric.
Another immediate observation is that, for the random graphs we consider, the degree distribution is rather concentrated.
The uniform leverage score heuristic yields good results for the CRSF-based sparsifiers, in terms of the approximation of the least eigenvector of the magnetic Laplacian, as shown in \cref{fig:approx_mg_laplacian}.

We also observe that for the MUN and ERO random connection graphs considered, CRSF and ST sampling have a comparable performance, whereas i.i.d.\ sampling gives a poorer result.
For the chosen parameters, the magnetic Laplacian approximation requires fewer edges for MUN (\cref{fig:approx_mg_laplacian_MUN_ranking}) than for ERO (\cref{fig:approx_mg_laplacian_ERO_ranking}).
For the MUN graph, the number of sampled CRTs is close to minimal, namely there is about one CRT per CRSF; see \cref{fig:approx_mg_laplacian_MUN_CRTs}.
Also, the amount of CRTs sampled by our algorithm is larger in the case of the ERO model (\cref{fig:approx_mg_laplacian_ERO_CRTs}), which confirms that the latter corresponds to a more difficult problem.
{Intuitively, when the walker in $\cyclepopping{}$ meets more often inconsistent cycles, the ranking problem contains more inconsistencies to be resolved.}
In \cref{sec:cycle_inconsistencies}, the average inconsistency of the sampled cycles are also displayed to illustrate this observation.

In conclusion, for these well-connected graphs, the sparsifiers obtained with batches of CRSFs or batches of STs both reach approximately the same performance for sparsifying Sync-Rank, in terms of ranking recovery.

\subsection{Laplacian preconditioning} 
\label{sub:laplacian_preconditioning}
As motivated in \cref{sub:regularized_laplacian_systems_preconditioning_and_sparse_cholesky_factorization}, we give a few examples of condition number reduction by using sparsifiers. First, we illustrate the preconditioning of the magnetic Laplacian $\Delta$, and second, we analyse the preconditioning of the regularized combinatorial Laplacian $ \Lambda + q \I_n$.
\subsubsection{Magnetic Laplacian: sparsify-and-precondition}

In our simulations, we consider connection graph models with a planted ranking and a low level of noise since $\cond(\Delta)$ is larger in this case.
On the $y$-axis of the plots, we display $\cond(\widetilde{\Delta}^{-1}\Delta)$.
{This condition number is the ratio of the largest to the smallest eigenvalue of the generalized eigenvalue problem $\Delta \mathbf{v} = \lambda \widetilde{\Delta} \mathbf{v}$.
We compute these extremal eigenvalues thanks to \texttt{Arpack.jl}, which can deal with sparse matrices.}

The $x$-axis indicates the total number of edges contained in the union of batches.
The sampling is repeated $3$ times for a fixed connection graph and the mean is displayed.
The error bars are large for i.i.d.\ sampling and are not displayed to ease the reading of the figures.
We compare the condition number
$
    \cond(\widetilde{\Delta}^{-1}\Delta)
$
for sparsifiers $\widetilde{\Delta}$ obtained from different sampling methods.
This problem is motivated by applications to angular synchronization described in
\cref{sub:Eigenvalue problem and angular synchronization} and to semi-supervised learning described  in  \cref{sub:regularized_laplacian_systems_preconditioning_and_sparse_cholesky_factorization}.
An advantage of using CRSFs to construct the sparsifier $\widetilde{\Delta}$ is that the latter is invertible almost surely even in the case of one CRSF.
This is not necessarily true if i.i.d.\ samples or STs are used in which case we replace the sparsifier by $\widetilde{\Delta} + 10^{-12} \I_n$.

\paragraph*{{Sparsify-and-precondition $\Delta$ of random graphs.}} We consider the case of MUN$(n,p,\eta)$ and ERO$(n,p,\eta)$ connection graphs with $n= 2000$, $p=0.01$. 
For the MUN graph, we choose $\eta = 10^{-3}$ (\cref{fig:cond_number_mag_q_0_MUN_0.001}) and $\eta = 10^{-1}$ (\cref{fig:cond_number_mag_q_0_MUN_0.1}) since a low value of the multiplicative noise yields a small least eigenvalue for $\Delta$.
For the same reason, for generating the ERO graphs, we take $\eta=10^{-4}$
 (\cref{fig:cond_number_mag_q_0_ERO_0.0001}) and $\eta = 10^{-3}$ (\cref{fig:cond_number_mag_q_0_ERO_0.001})  so that a few percentage of edges (outliers) are affected by noise.
Both for MUN and ERO, results are averaged over $3$ independent runs.

We observe the sparsifiers obtained with CRSF sampling are more efficient that other sampling techniques.
For the MUN noise model where the noise is uniform over the edges, batches of spanning trees also provide a good approximation.

In the case of the ERO model, pure noise is corrupting certain edges whereas all the other edges are noiseless. 
In \cref{fig:cond_number_mag_ERO_q_0} for a small $\eta$, sparsifiers built with CRSFs yield a better approximation compared with STs or other approaches. Intuitively, we can understand this behaviour as follows: by construction, the ERO random graph contains only a few inconsistent edges, which are likely to be sampled in CRSFs.
Recall that an ERO$(n,p,\eta)$ random graph is a connection graph where an edge is present with probability $p$ and is corrupted by a uniform noise with probability $\eta$.
To confirm our intuition for this ERO model, we computed the average number of \emph{noisy} edges captured by the generated random subgraph.{
    We found $1.04 (0.2)$ noisy edges on average in a random CRSF and $0.17 (0.4)$ noisy edges in a random ST; note that the standard deviation over $100$ runs is given in parentheses.
}
%

{\paragraph*{Sparsify-and-precondition $\Delta$ of real graphs.} We also perform the same simulations on two connection graphs built on the Epinions and Stanford networks as described in \cref{sec:random_graph_real_nets}.}
{The choosen parameters are $\eta = 5\cdot 10^{-2}$ for Epinions-MUN and $\eta = 2 \cdot 10^{-5}$ for Epinions-O are such that the least eigenvalue of $\Delta$ is small.
Similarly, we take $\eta = 10^{-2}$ for Stanford-MUN and $\eta = 2\cdot 10^{-5}$ for Stanford-O.}

{
For these larger graphs, leverage score estimation is time consuming. Therefore, in \cref{fig:cond_number_Epinions_q_0} (Epinions) and  \cref{fig:cond_number_Stanford_q_0} (Stanford), we restrict our comparisons to ``DPP(K) unif'' and ``ST unif'' and run the simulation only once since the computation of the condition numbers by solving generalized eigenvalue problems is the bottleneck in this case.
In these two figures, we see an advantage of ``DPP(K) unif'' over ``ST unif'' in the ERO case.
For the MUN model, there is only a clear difference in favor of ``DPP(K) unif'' for the Stanford graph, while for the Epinions graph both methods have a similar performance.
We observe in \cref{fig:cond_number_MUN_epinions} that with a batch of two spanning subgraphs, the condition number of $\Delta$ is reduced by about three orders of magnitude.
Empirical sampling times of $\cyclepopping{}$ for CRSFs and Wilson's algorithm for STs (with uniformly sampled root node) are computed over $100$ runs and reported in \cref{fig:epinions_MUN_time_CRSF_vs_ST} and \cref{fig:epinions_ERO_time_CRSF_vs_ST}. 
The upshot is that there is no large difference between these two sampling times in this setting.}

\begin{remark}[Bound on the condition number]
    Note that the largest eigenvalue of $\Delta$ is always bounded from above by $2 \max_v d(v)$, as can be seen using Gershgorin's circles theorem. Thus, the condition number of $\Delta + q \I_n$ is bounded from above by $2 \max_v d(v) / (q + \lambda_{1}(\Delta))$. 
    A consequence is that ill-conditioning is more likely to arise for large $\Uone$-connection graphs with a low level of inconsistencies.
\end{remark}
\begin{figure}
    \centering
    \begin{subfigure}[b]{0.49\textwidth}
        \centering
        \includegraphics[scale=0.4]{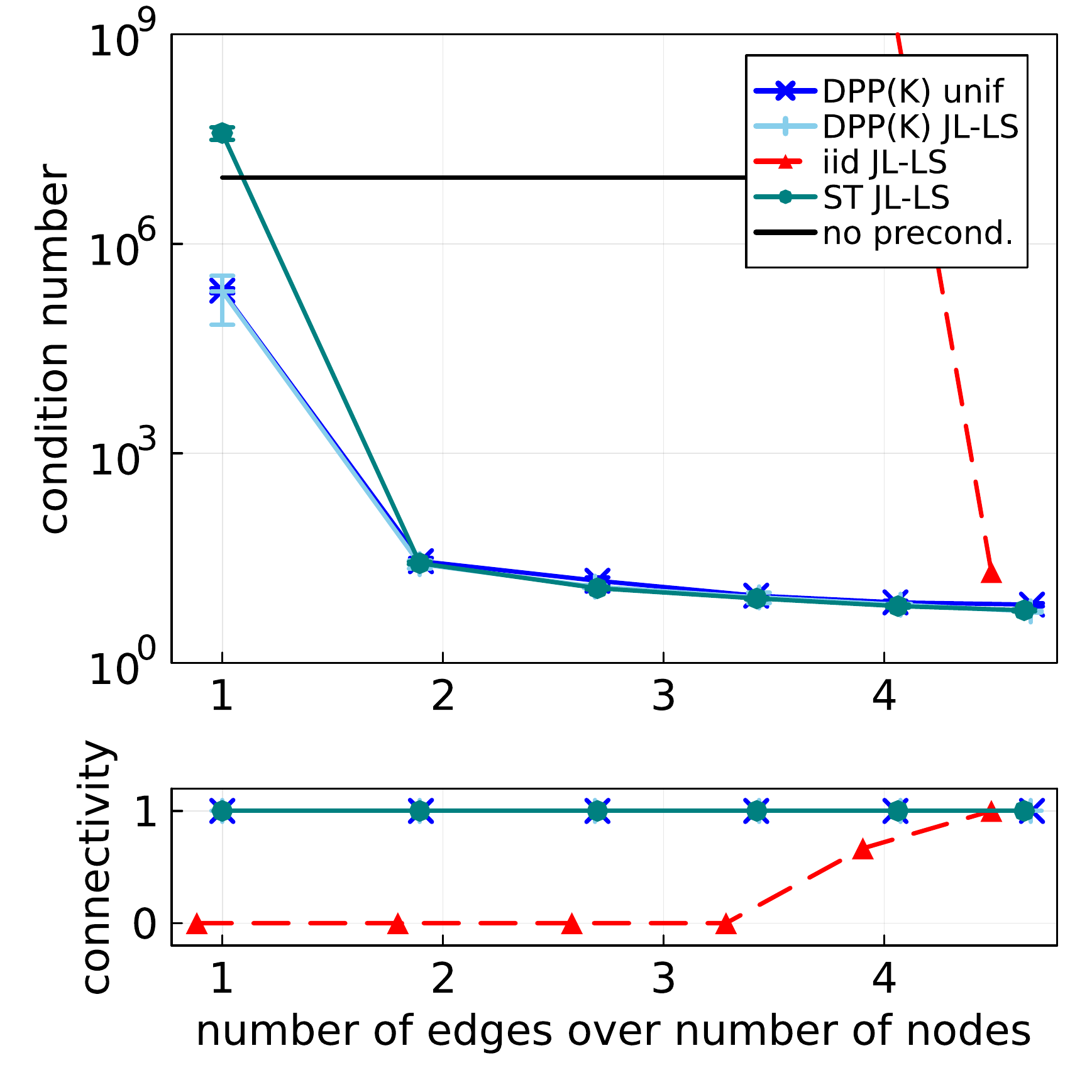}
        \caption{$\cond(\widetilde{\Delta} 
        ^{-1}\Delta)$ for MUN: $\eta=10^{-3}$. \label{fig:cond_number_mag_q_0_MUN_0.001}}
    \end{subfigure}
    \begin{subfigure}[b]{0.49\textwidth}
        \centering
        \includegraphics[scale=0.4]{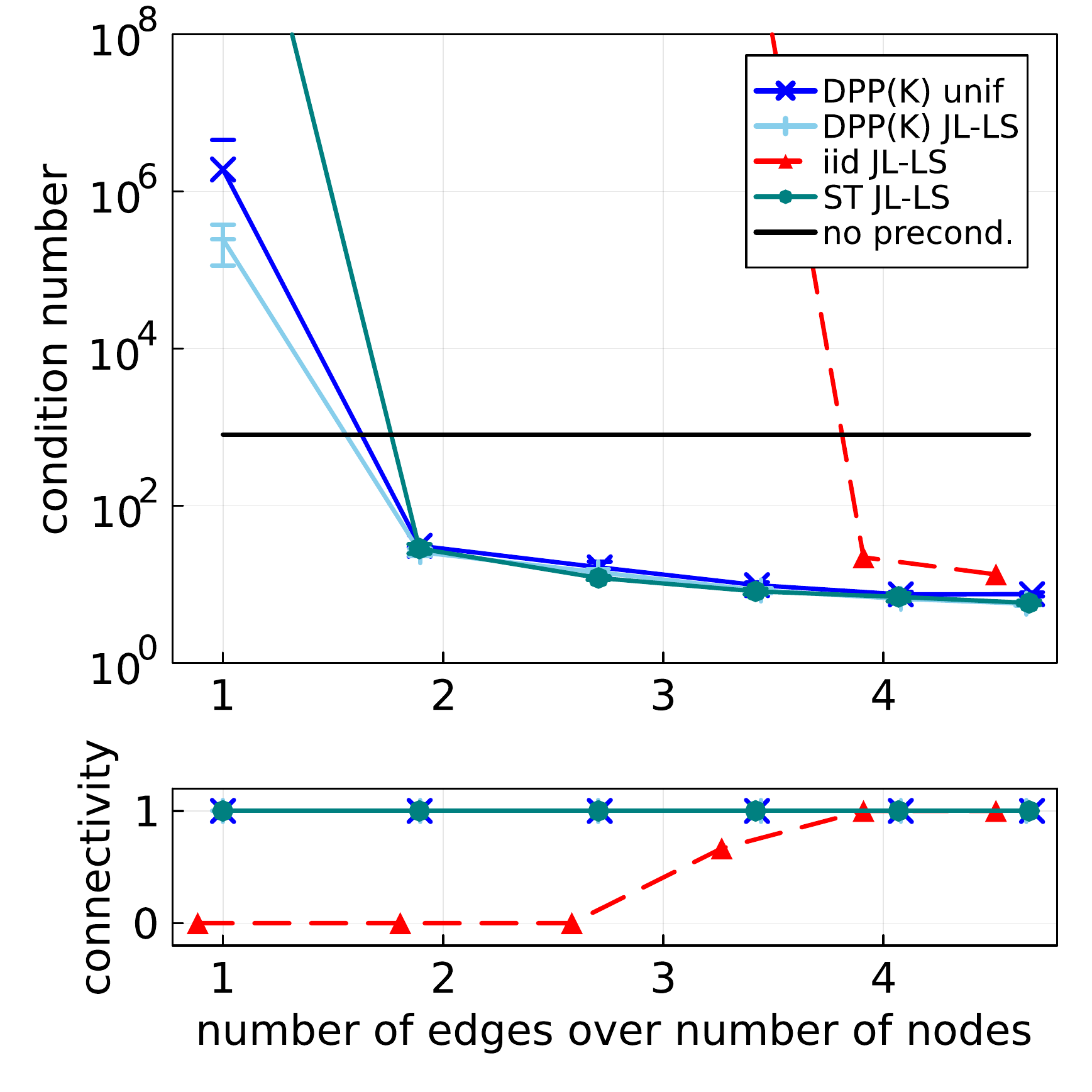}
        \caption{$\cond(\widetilde{\Delta} 
        ^{-1}\Delta)$  for MUN:  $\eta=10^{-1}$.\label{fig:cond_number_mag_q_0_MUN_0.1}}
    \end{subfigure}
    \caption{
        Sparsify-and-precondition the magnetic Laplacian of a MUN$(n,p,\eta)$ {graph with $n= 2000$ and $p=0.01$}. We display $\cond(\widetilde{\Delta} 
         ^{-1}\Delta)$.
        Results are averaged over $3$ independent estimates $\widetilde{\Delta}$ for a fixed MUN graph.}
    \label{fig:cond_number_Epinions_q_0}
\end{figure}

\begin{figure}
    \centering
    \begin{subfigure}[b]{0.49\textwidth}
        \centering
        \includegraphics[scale=0.4]{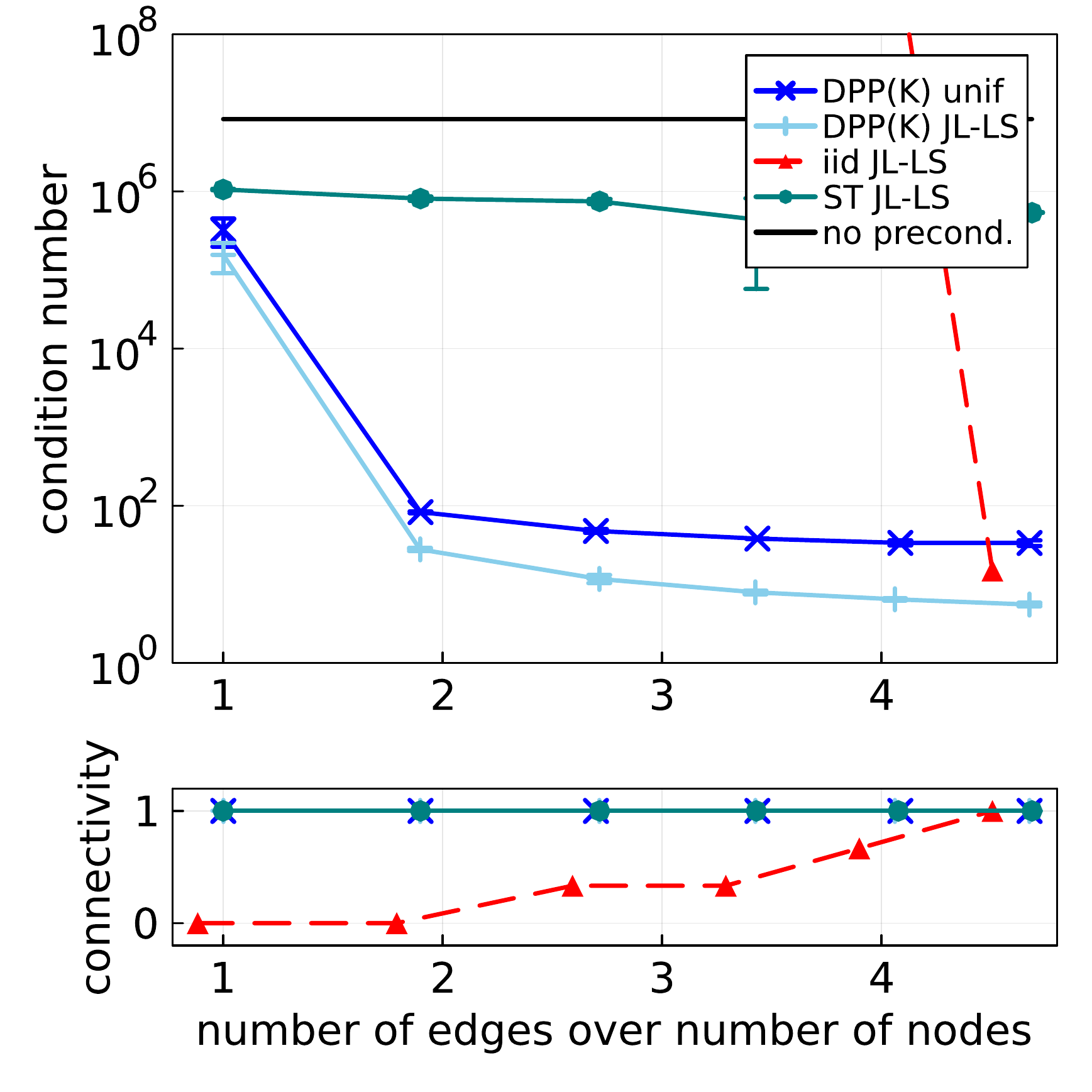}
        \caption{$\cond(\widetilde{\Delta} 
        ^{-1}\Delta)$ for ERO: $\eta=10^{-4}$.\label{fig:cond_number_mag_q_0_ERO_0.0001}}
    \end{subfigure}
    \begin{subfigure}[b]{0.49\textwidth}
        \centering
        \includegraphics[scale=0.4]{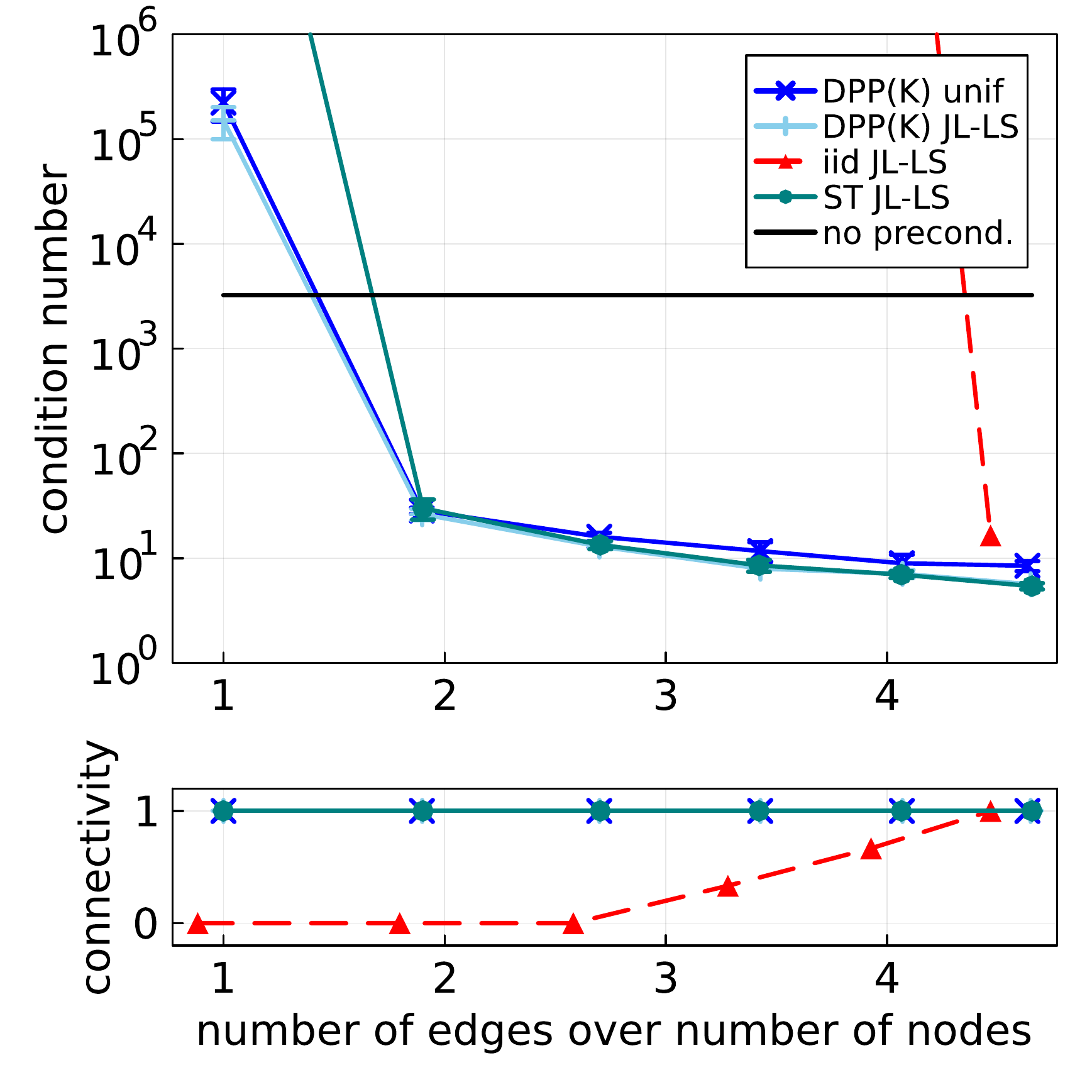}
        \caption{$\cond(\widetilde{\Delta} 
        ^{-1}\Delta)$ for ERO: $\eta=10^{-3}$. \label{fig:cond_number_mag_q_0_ERO_0.001}}
    \end{subfigure}
    \caption{
        Sparsify-and-precondition the magnetic Laplacian of ERO$(n,p,\eta)$ graph with $n= 2000$ and $p=0.01$. 
        We display $\cond(\widetilde{\Delta} 
         ^{-1}\Delta)$.
         Results are averaged over $3$ independent estimates $\widetilde{\Delta}$ for a fixed ERO graph.}
    \label{fig:cond_number_mag_ERO_q_0}
\end{figure}
\begin{figure}
    \centering
    \begin{subfigure}[b]{0.49\textwidth}
        \centering
        \includegraphics[scale=0.4]{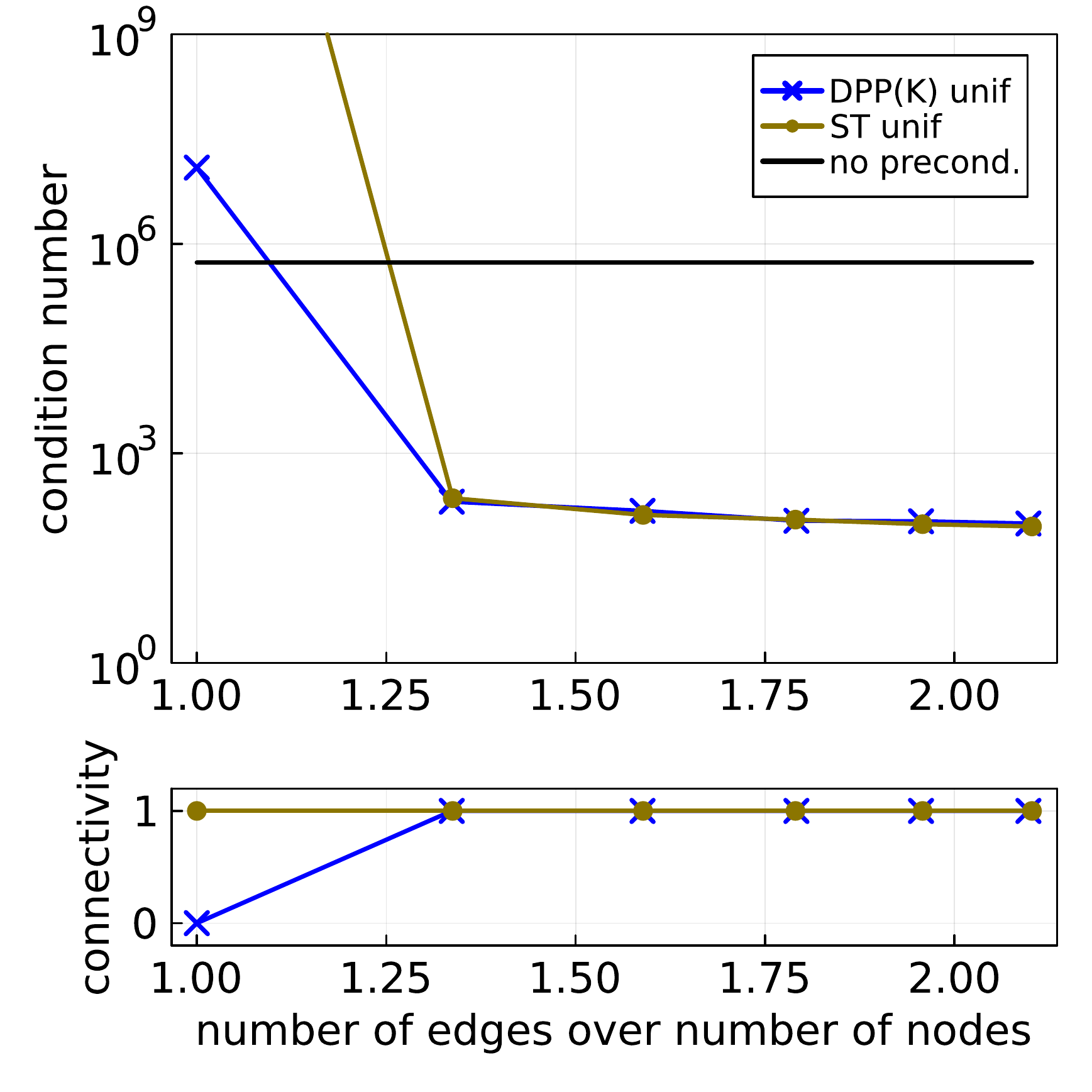}
        \caption{$\cond(\widetilde{\Delta} 
        ^{-1}\Delta)$ for Epinions-MUN: $\eta=5\cdot 10^{-2}$. \label{fig:cond_number_MUN_epinions}}
    \end{subfigure}
    \hfill
    \begin{subfigure}[b]{0.49\textwidth}
        \centering
        \includegraphics[scale=0.4]{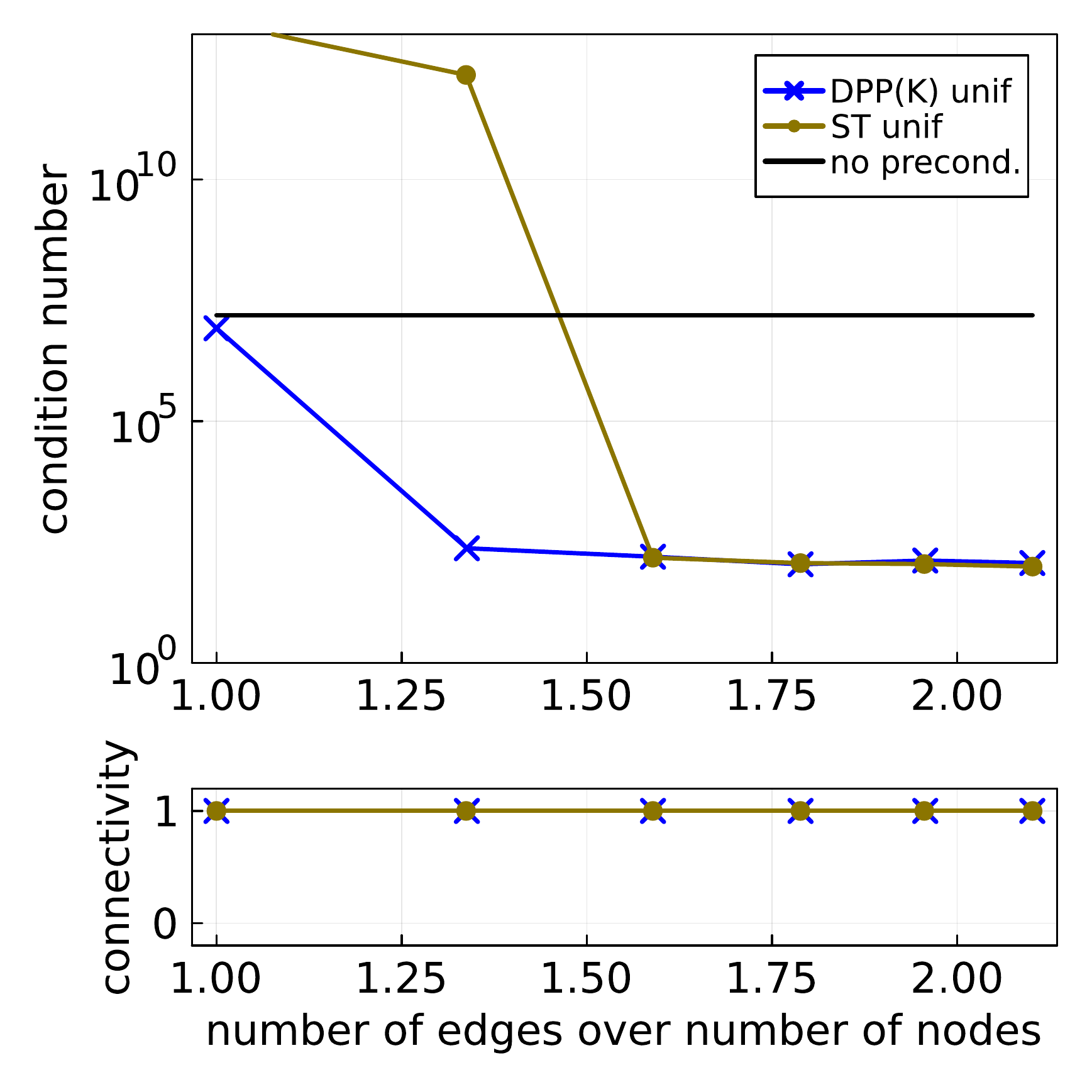}
        \caption{$\cond(\widetilde{\Delta} 
        ^{-1}\Delta)$ for Epinion-O: $\eta= 2\cdot 10^{-5}$. \label{fig:cond_number_ERO_epinions}}
    \end{subfigure}
    \begin{subfigure}[b]{0.49\textwidth}
        \centering
        \includegraphics[scale=0.4]{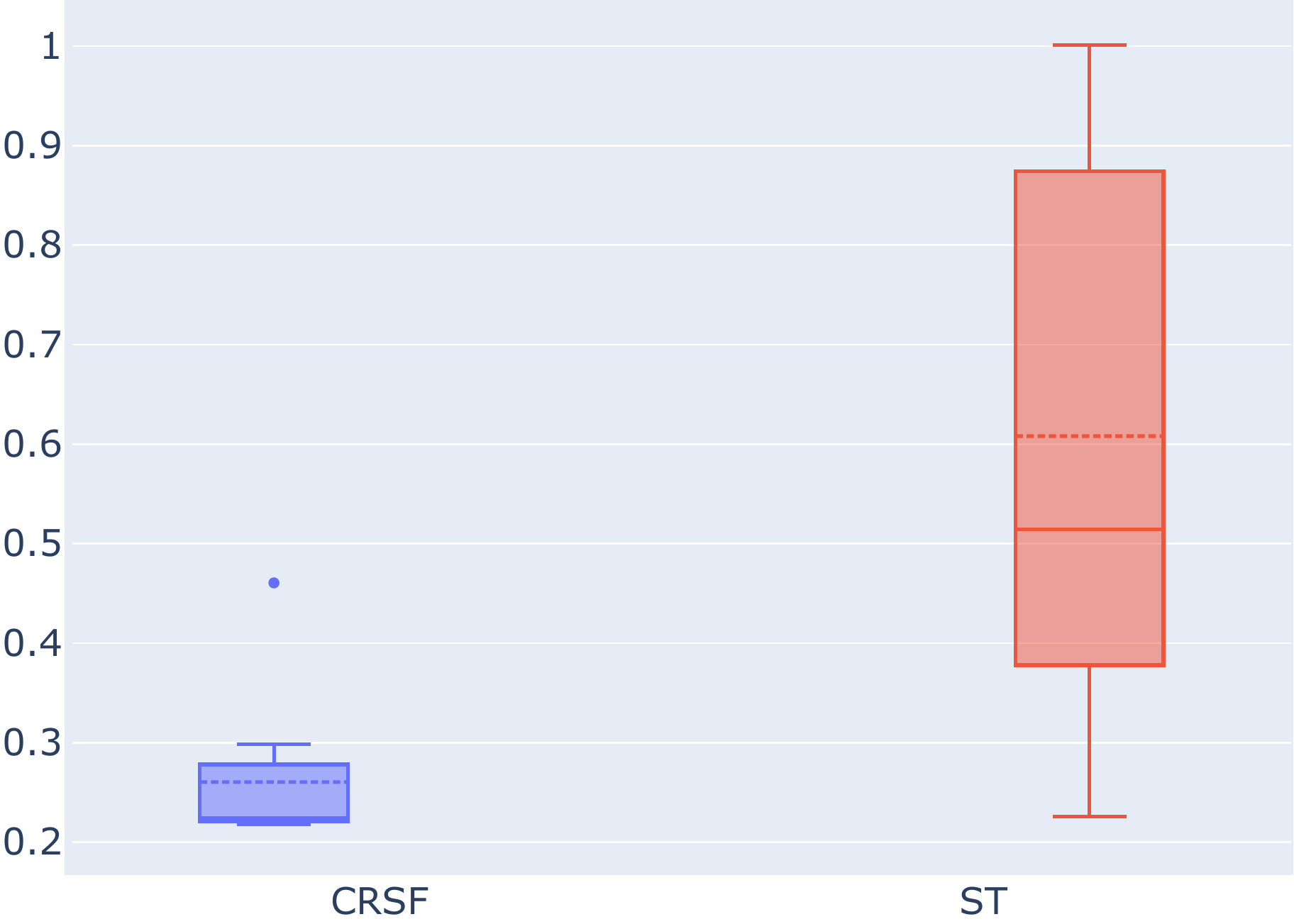}
        \caption{Sampling time (s) in Epinions-MUN: $\eta=5\cdot 10^{-2}$. \label{fig:epinions_MUN_time_CRSF_vs_ST}}
    \end{subfigure}
    \hfill
    \begin{subfigure}[b]{0.49\textwidth}
        \centering
        \includegraphics[scale=0.4]{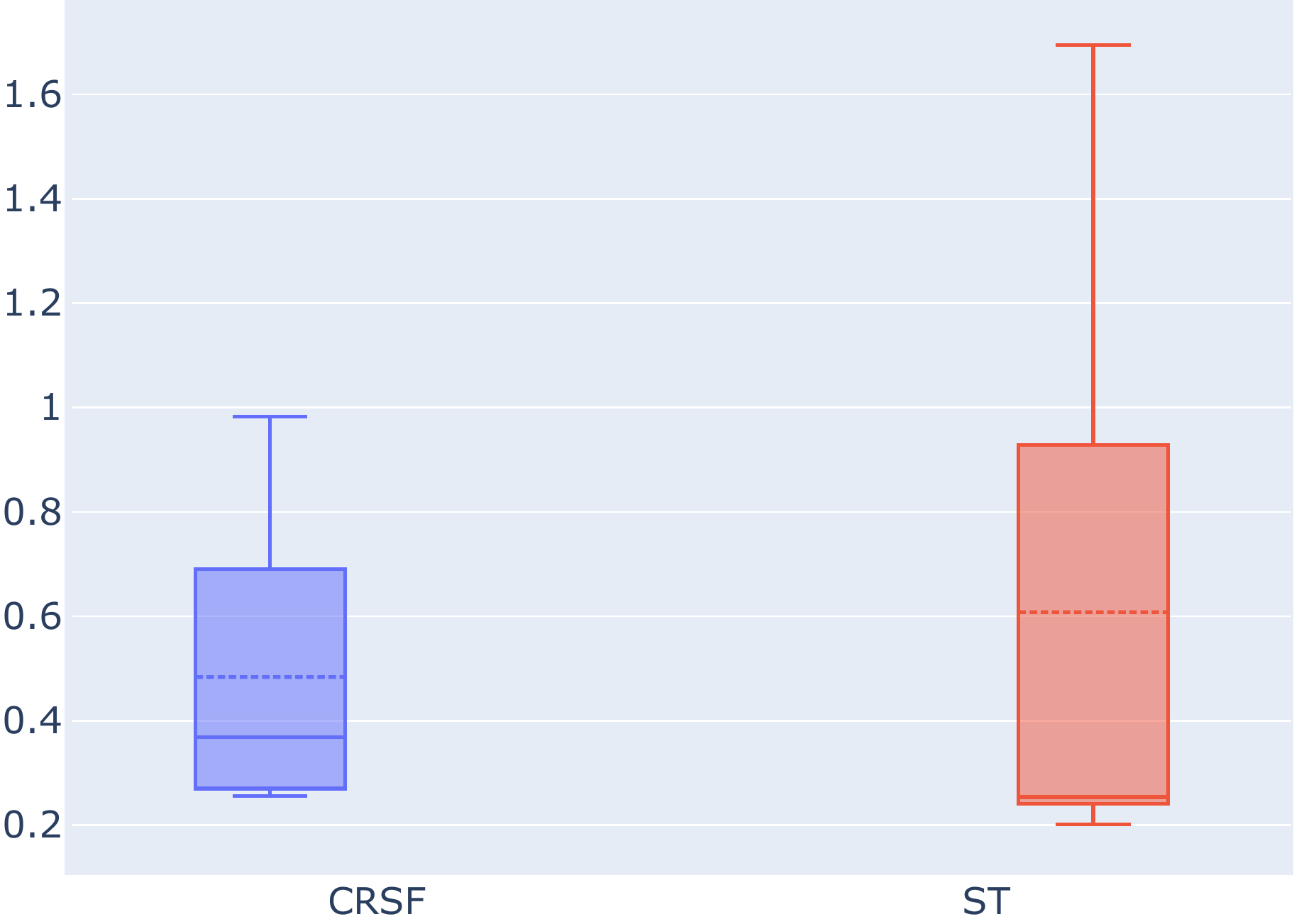}
        \caption{Sampling time (s) in Epinions-O: $\eta= 2\cdot 10^{-5}$. \label{fig:epinions_ERO_time_CRSF_vs_ST}}
    \end{subfigure}
    \caption{
        Sparsify-and-precondition the magnetic Laplacian of Epinions-MUN$(\eta)$ and Epinions-O$(\eta)$ graphs. 
        The first row of figures displays $\cond(\widetilde{\Delta} 
         ^{-1}\Delta)$ for only $1$ estimate $\widetilde{\Delta} $.
        The second row reports boxplots of sampling times (s) of $100$ executions of $\cyclepopping{}$ for CRSFs and Wilson's algorithm for STs. For each boxplot, the mean is the dashed line in the box.}
    \label{fig:cond_number_Epinions_q_0}
\end{figure}
\begin{figure}
    \centering
    \begin{subfigure}[b]{0.49\textwidth}
        \centering
        \includegraphics[scale=0.4]{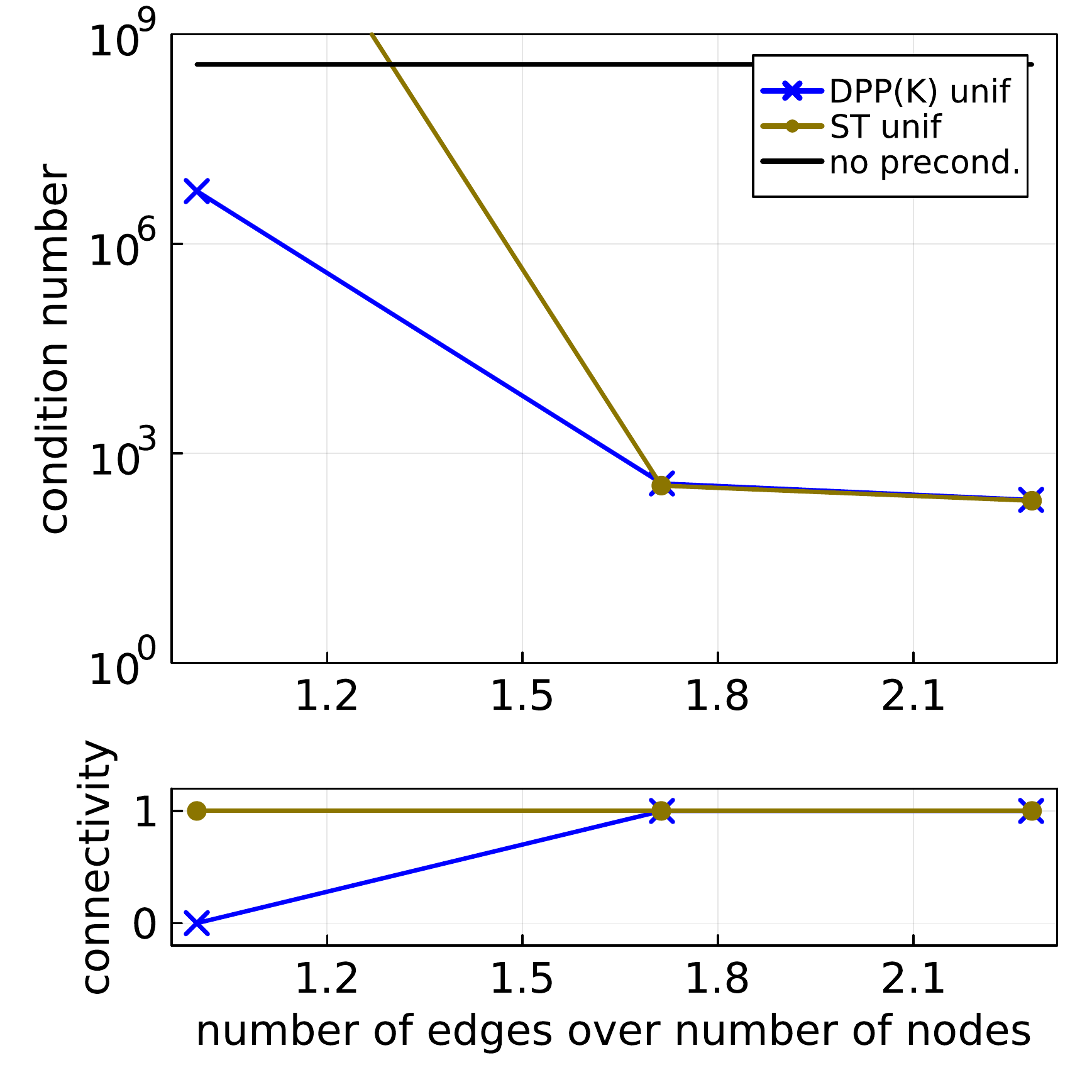}
        \caption{$\cond(\widetilde{\Delta} 
        ^{-1}\Delta)$ Stanford-MUN: $\eta=10^{-2}$. \label{fig:cond_number_MUN_epinions}}
    \end{subfigure}
    \hfill
    \begin{subfigure}[b]{0.49\textwidth}
        \centering
        \includegraphics[scale=0.4]{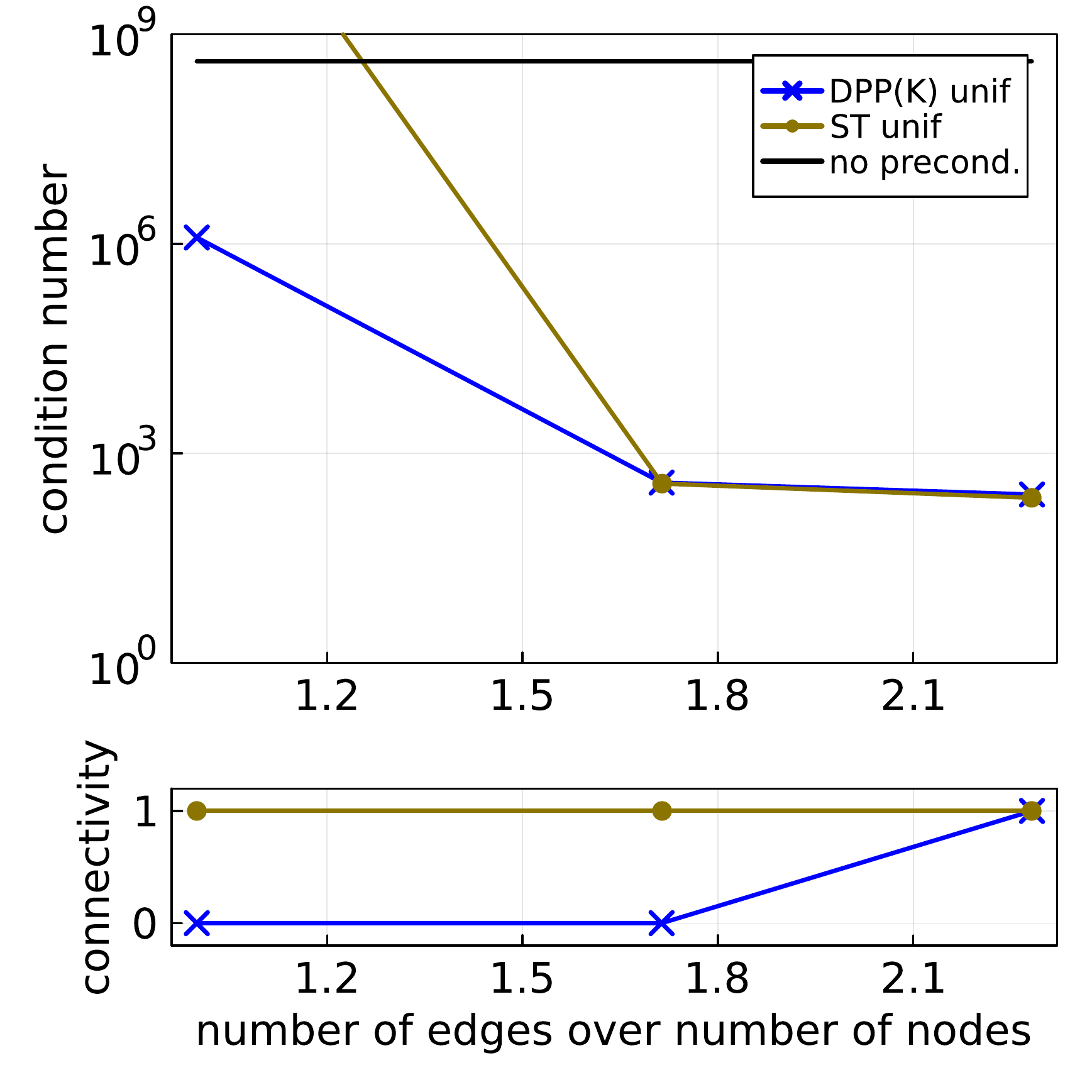}
        \caption{$\cond(\widetilde{\Delta} 
        ^{-1}\Delta)$ Stanford-O: $\eta=  10^{-5}$. \label{fig:cond_number_ERO_epinions}}
    \end{subfigure}
    \begin{subfigure}[b]{0.49\textwidth}
        \centering
        \includegraphics[scale=0.4]{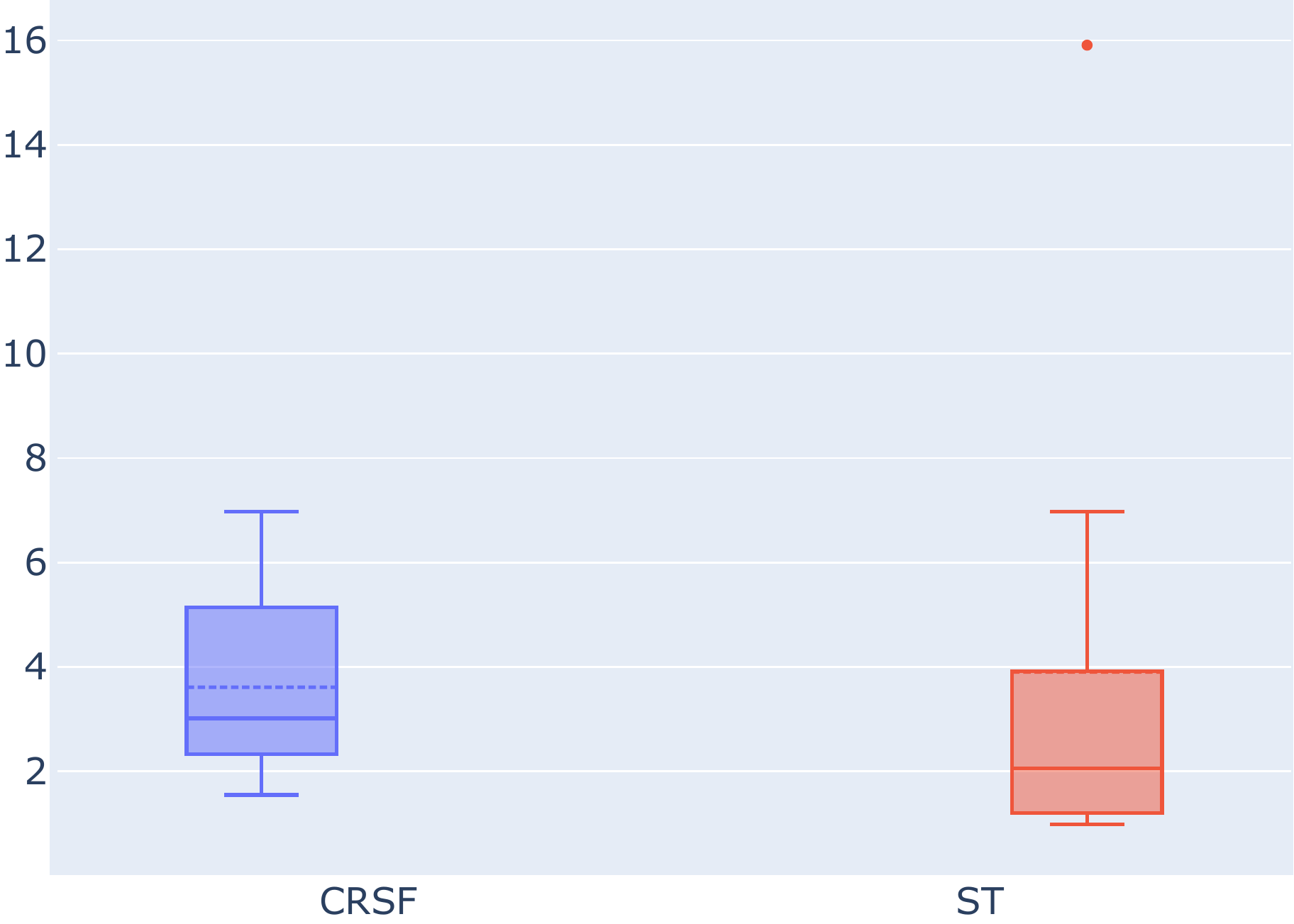}
        \caption{Sampling time (s) in Stanford-MUN: $\eta=10^{-2}$. \label{fig:epinions_MUN_time_CRSF_vs_ST}}
    \end{subfigure}
    \hfill
    \begin{subfigure}[b]{0.49\textwidth}
        \centering
        \includegraphics[scale=0.4]{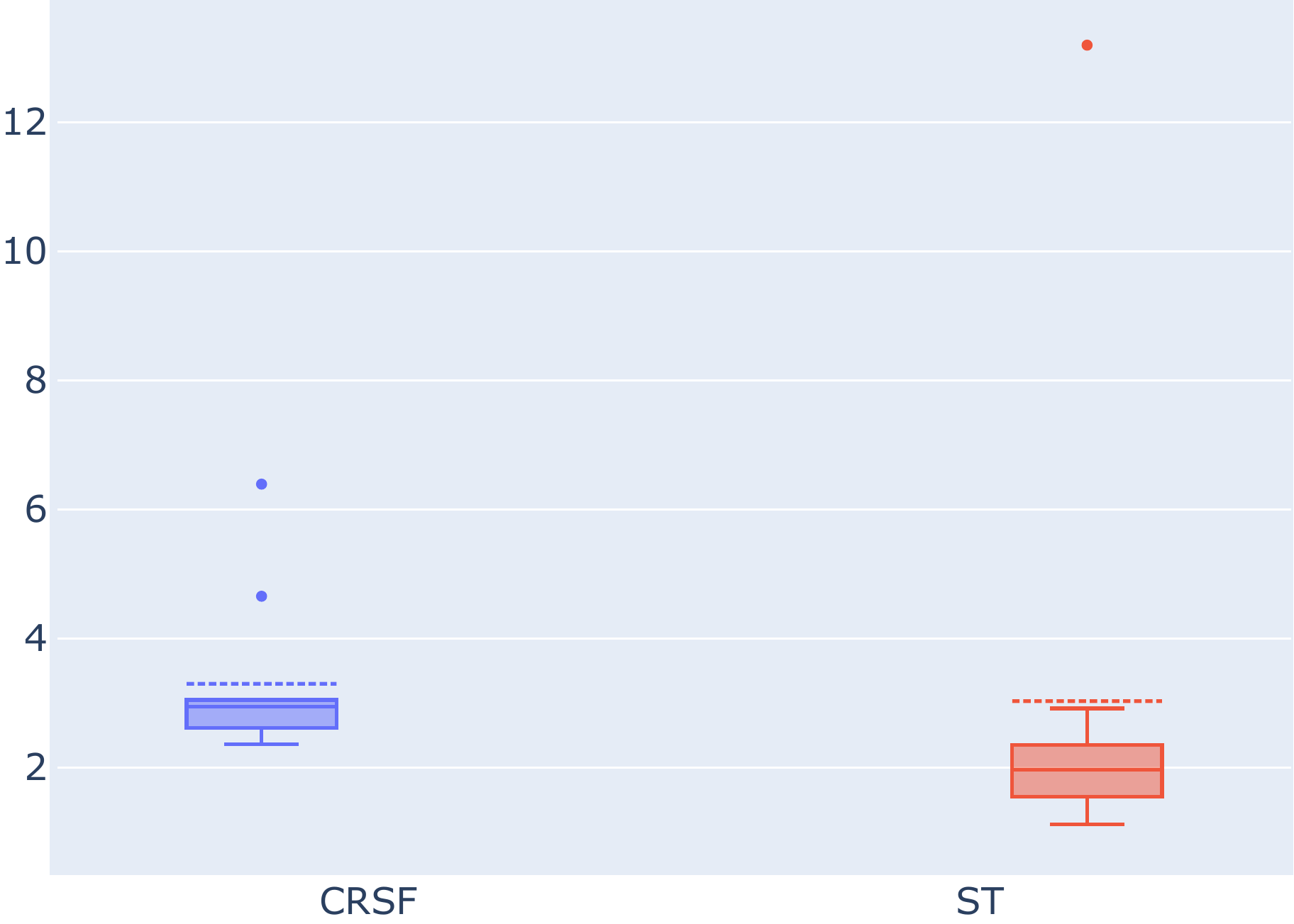}
        \caption{Sampling time (s) in Stanford-O: $\eta=  10^{-5}$. \label{fig:epinions_ERO_time_CRSF_vs_ST}}
    \end{subfigure}
    \caption{
        Sparsify-and-precondition the magnetic Laplacian of Stanford-MUN$(\eta)$ and Stanford-O$(\eta)$ graphs. 
        The first row of figures displays $\cond(\widetilde{\Delta} 
         ^{-1}\Delta)$ for only $1$ estimate $\widetilde{\Delta} $.
        The second row reports boxplots of sampling times (s) of $10$ executions of $\cyclepopping{}$ for CRSFs and Wilson's algorithm for STs. For each boxplot, the mean is the dashed line in the box.}
    \label{fig:cond_number_Stanford_q_0}
\end{figure}
\subsubsection{Regularized Combinatorial Laplacian}
We consider here the sparsified combinatorial Laplacian defined by $\widetilde{\Lambda} = B^*_0 SS^\top  B_0$ where $B_0$ is the (real) oriented incidence matrix -- see \cref{sec:magnetic_laplacian} -- and the sampling matrix
$
    S_{ee'}(\calC) = \delta_{ee'} 1_{\calC}(e') /\sqrt{\lev_0(e)},
$
where $\lev_0(e) = [B_0(B_0^* B_0)^{+}B_0^*]_{ee}$ for all $e\in \calE$.
In what follows, we simply display in the figures the condition number of $(\widetilde{\Lambda} + q \I_n)^{-1}(\Lambda + q \I_{n})$ for sparsifiers $\widetilde{\Lambda}$ obtained thanks to several sampling strategies.
\paragraph{Sparsify-and-precondition $\Lambda + q \I_{n}$ of Erd\H{o}s-Rényi graph.}  Our first simulation concerns an Erd\H{o}s-Rényi graph $\mathrm{ER}(n,p)$ with $n=2000$ and $p=0.01$.
The combinatorial Laplacian of this type of graphs has a simple spectrum and is reputated to have a simple topology as well as no clear community structure.
Hence, we expect a concentrated leverage score distribution.
To give an idea of the magnitude of its eigenvalues, the three smallest eigenvalues of a random $\mathrm{ER}(n,p)$ are {$\lambda_1(\Lambda)=0$, $\lambda_2(\Lambda)\approx 6$ and $\lambda_3(\Lambda)\approx 7$. 
In this academic example, which is designed to understand the role of leverage score estimation on the sparsifiers, we consider very small values of $q$ in comparison with $\lambda_2$.}
In \cref{fig:cond_number_std}, we display on the $y$-axis $\cond((\widetilde{\Lambda} + q \I_n)^{-1}(\Lambda + q \I_{n}))$ for $q = 10^{-3}$ (Left-hand side) and $q=10^{-1}$ (Right-hand side) for an $\widetilde{\Lambda}$ built with batchsizes ranging from $1$ to $6$, whereas the $x$-axis indicates the total number of edges over $n$.
Also, we report the empirical distribution of the number of roots to visualize the effect of $q$ on the number of connected components of the SFs.
The sampling is repeated $3$ times for a fixed Erd\H{o}s-R\'enyi graph and the mean is displayed.

We observe that the SFs sampling strategy outperforms i.i.d.\ sampling. 
Again, we presume that this is due to a lack of connectivity of the subgraphs with i.i.d.\ edges; see the bottom of \cref{fig:cond_number_std_cond_0.001} and \cref{fig:cond_number_std_cond_0.1}.
As it is expected, the SF and i.i.d.\ approaches tend to reach in \cref{fig:cond_number_std} the same accuracy as the number of sampled edges increases.
As displayed in \cref{fig:root_nb_0.001} and \cref{fig:root_nb_0.1}, for a small value of $q$, the sampled subgraphs are often spanning trees and have many more connected components for a larger value of $q$.

\paragraph{Sparsify-and-precondition $\Lambda + q \I_{n}$ of a real graph.} 
Now, we perform the same simulations for a real network, the Polblogs data set\footnote{\url{http://www-personal.umich.edu/~mejn/netdata/}}, which represents the connectivity of blogs about US politics~\citep{polblogs} and is considered here as an  undirected graph.
This graph has $n=1,490$ nodes and $m = 16,718$ edges.
{Only its largest connected component with $n=1,222$ nodes and $m = 16,717$ is considered here.
This network is known to have two densely connected clusters or communities, so that we can expect the leverage score distribution to be less concentrated than in the case of the $\mathrm{ER}$ graph.}
{Note that the three smallest eigenvalues are $\lambda_1(\Lambda)=0$, $\lambda_2(\Lambda)\approx 0.17$ and $\lambda_3(\Lambda)\approx 0.3$.
In the spirit of Tikhonov regularization, the value of the regularization parameter $q$ has to be chosen in order to filter out large frequencies; see \cref{eq:Laplacian_q_reg_system}.
Hence, for our simulations, we choose two values of $q$: one much smaller than $\lambda_2(\Lambda)$ and another with the same order of magnitude. 
These are still rather small values so that the preconditioning problem remains interesting.}
In \cref{fig:cond_number_Polblogs}, we display on the $y$-axis $\cond((\widetilde{\Lambda} + q \I_n)^{-1}(\Lambda + q \I_{n}))$ for $q=0.001$ (Left-hand side) and  for $q = 0.1$ (Right-hand side) with up to $6$ batches of edges, whereas the $x$-axis indicates the total number of edges.
On the right-hand side, we report the cumulated number of roots sampled for each set of batches. We display the average and standard deviations over $3$ runs.

We observe a larger difference between the uniform LS heuristics and the JL approximation, due to the more complex LS distribution of this real network.
Although it is obviously biased, the former heuristics is still able to reduce significantly the condition number of the full Laplacian, namely by about two orders of magnitude.

\begin{figure}
    \centering
    \begin{subfigure}[b]{0.49\textwidth}
        \centering
            \includegraphics[scale=0.4]{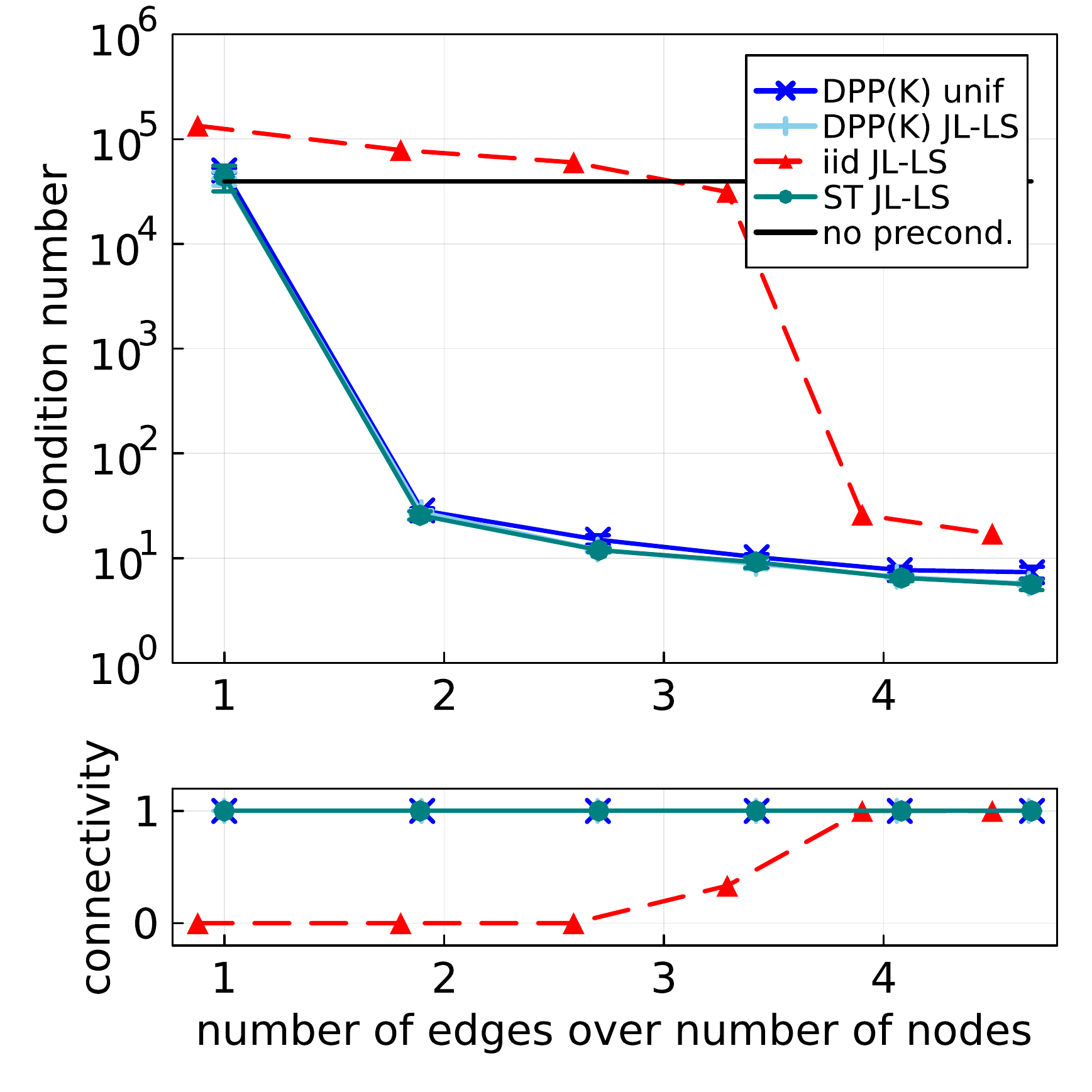}
            \caption{$\cond((\widetilde{\Lambda} + q \I)^{-1}(\Lambda + q \I))$ for $\mathrm{ER}$: $q = 0.001$. \label{fig:cond_number_std_cond_0.001}}
    \end{subfigure}
    \hfill
    \begin{subfigure}[b]{0.49\textwidth}
        \centering
            \includegraphics[scale=0.4]{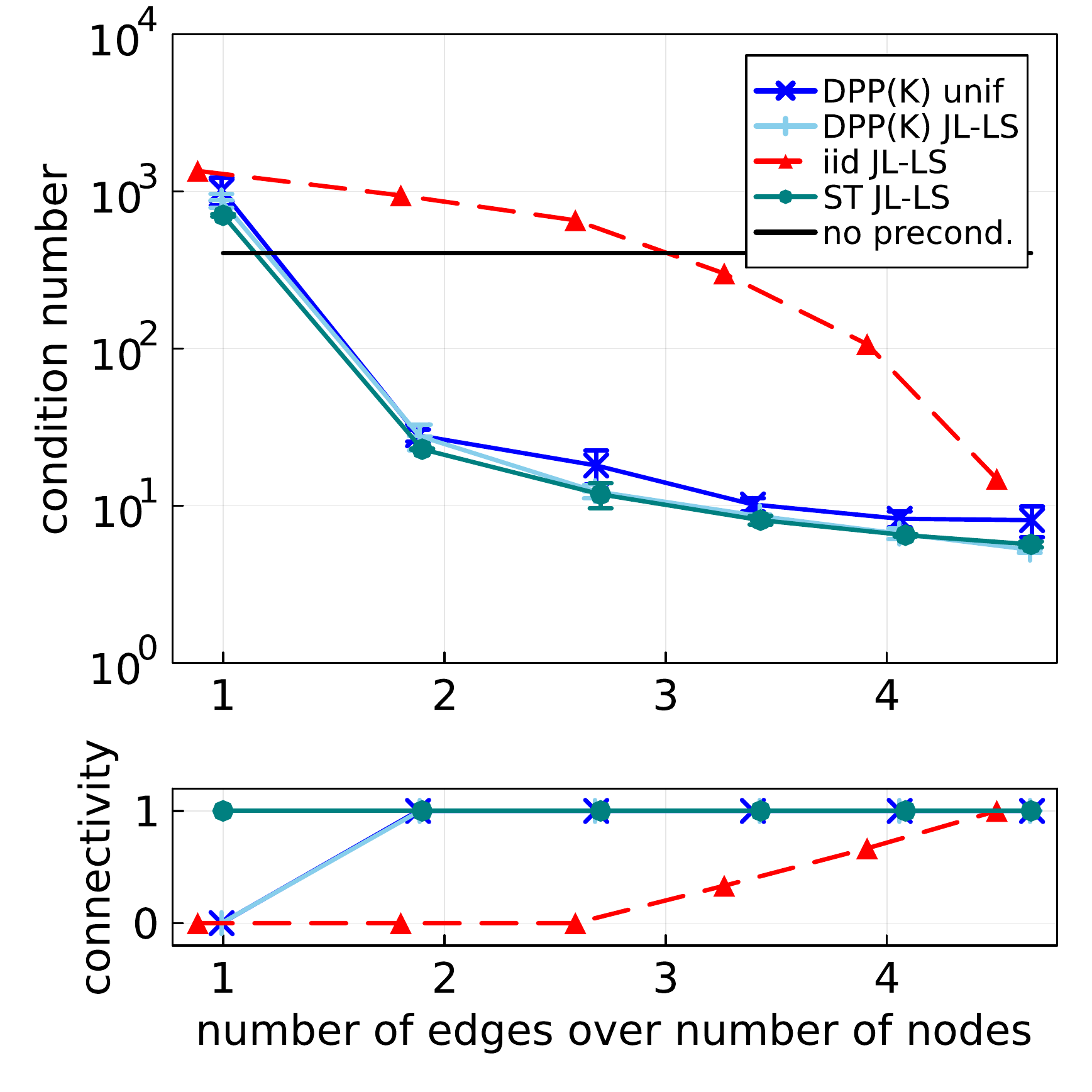}
            \caption{$\cond((\widetilde{\Lambda} + q \I)^{-1}(\Lambda + q \I))$ for $\mathrm{ER}$: $q = 0.1$. \label{fig:cond_number_std_cond_0.1}}
    \end{subfigure}
    \begin{subfigure}[b]{0.49\textwidth}
        \centering
        \includegraphics[scale=0.25]{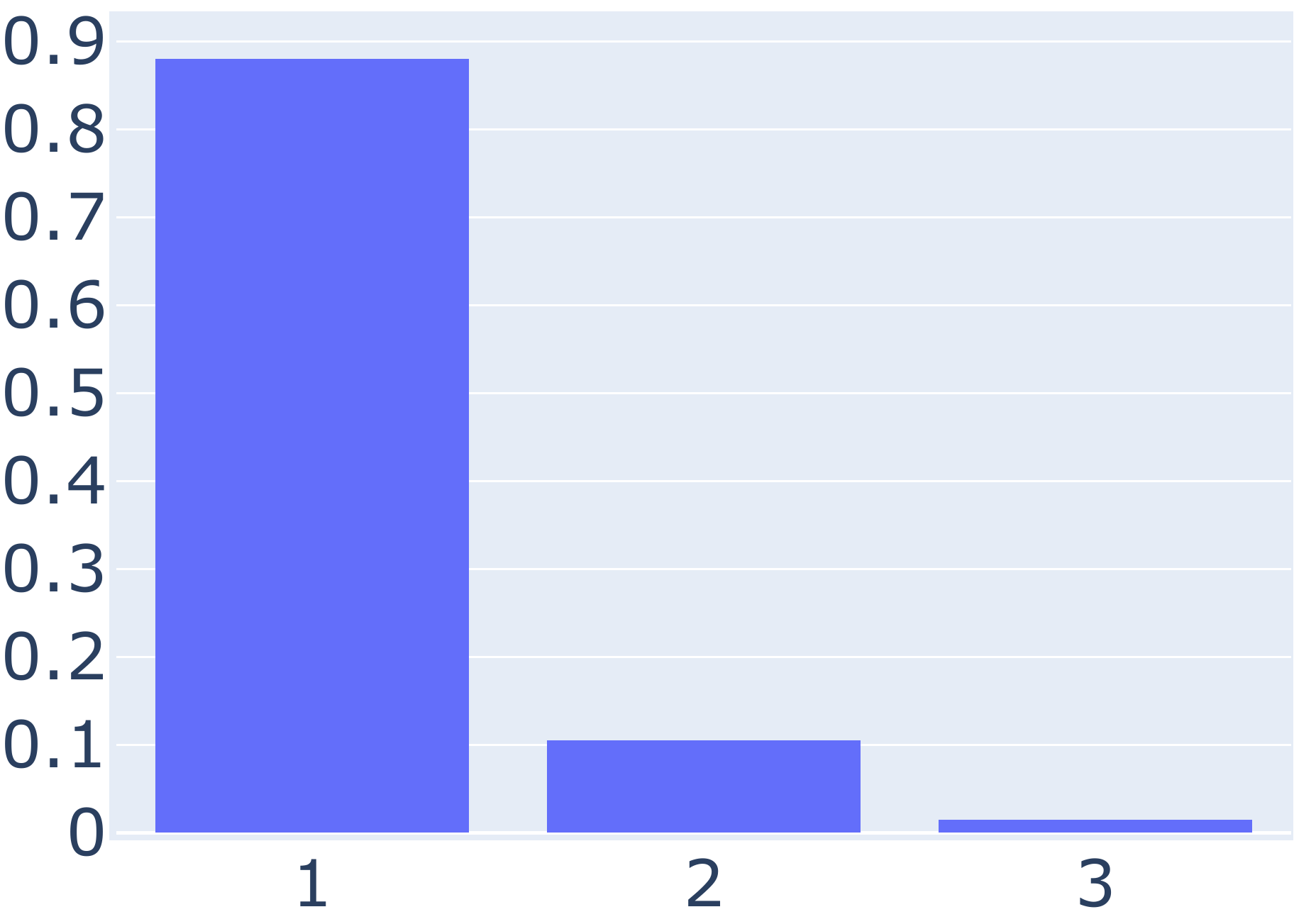}
        \caption{Root number distribution for $\mathrm{ER}$: $q = 0.001$. \label{fig:root_nb_0.001}}
    \end{subfigure}
    \hfill
    \begin{subfigure}[b]{0.49\textwidth}
        \centering
        \includegraphics[scale=0.25]{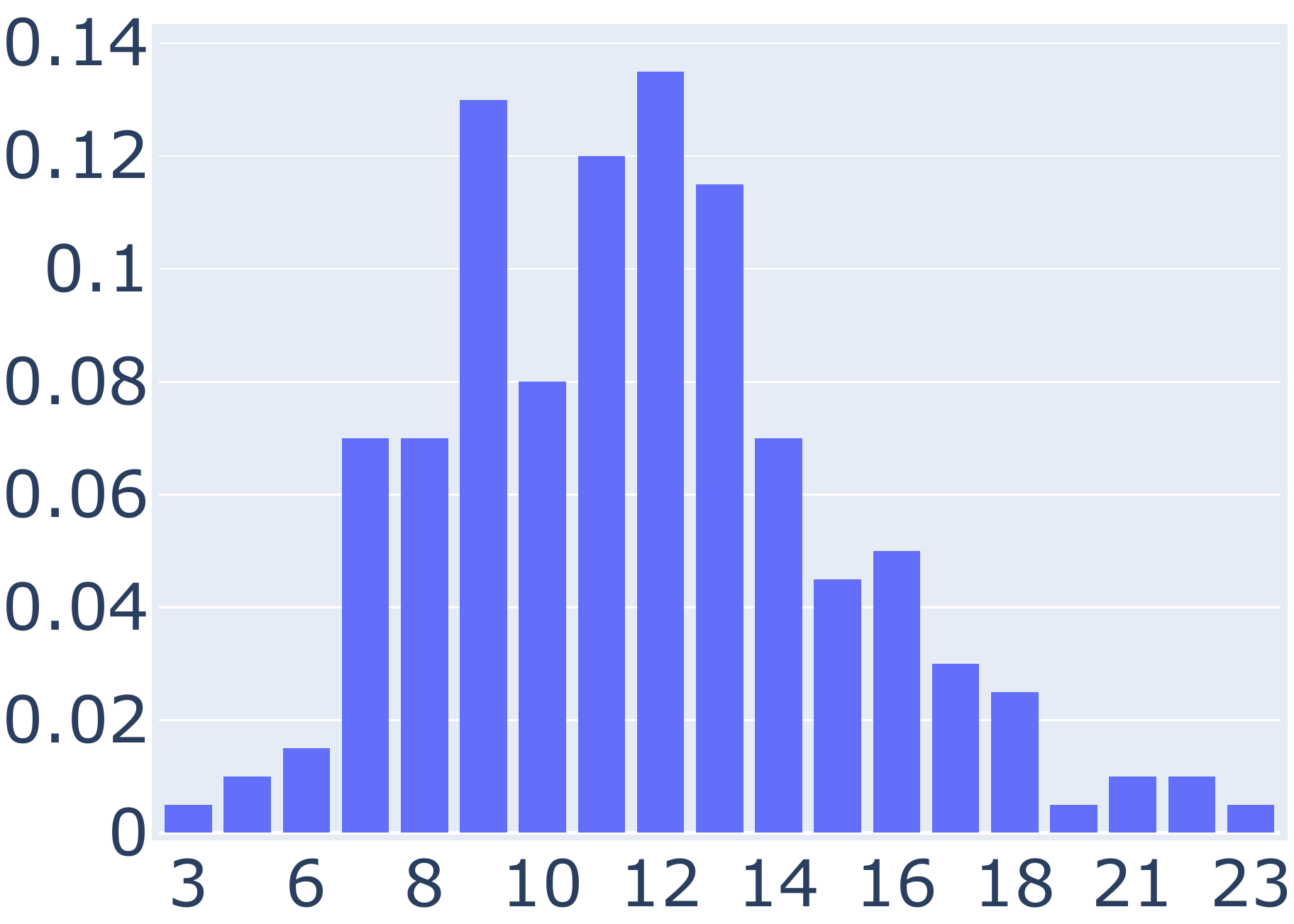}
        \caption{Root number distribution for $\mathrm{ER}$: $q = 0.1$. \label{fig:root_nb_0.1}}
    \end{subfigure}
    \caption{Sparsify-and-precondition the regularized combinatorial Laplacian of $\mathrm{ER}(n,p)$ graph for $n=2000$ and $p=0.01$. 
    At the first row of \cref{fig:cond_number_std_cond_0.001} and \cref{fig:cond_number_std_cond_0.1}, we report $\cond((\widetilde{\Lambda} + q \I)^{-1}(\Lambda + q \I))$ for several methods as a function of $n$ divided by the number of edges in the sparsifier, whereas the second row displays the connectivity of the subgraphs. 
    These figures report averages over $3$ independent runs.
    \cref{fig:root_nb_0.001} and \cref{fig:root_nb_0.1} display the empirical distribution of the number of roots in one determinantal subgraph. }
    \label{fig:cond_number_std}
\end{figure}
%
%
\begin{figure}
    \centering
    \begin{subfigure}[b]{0.49\textwidth}
        \centering
        \includegraphics[scale=0.4]{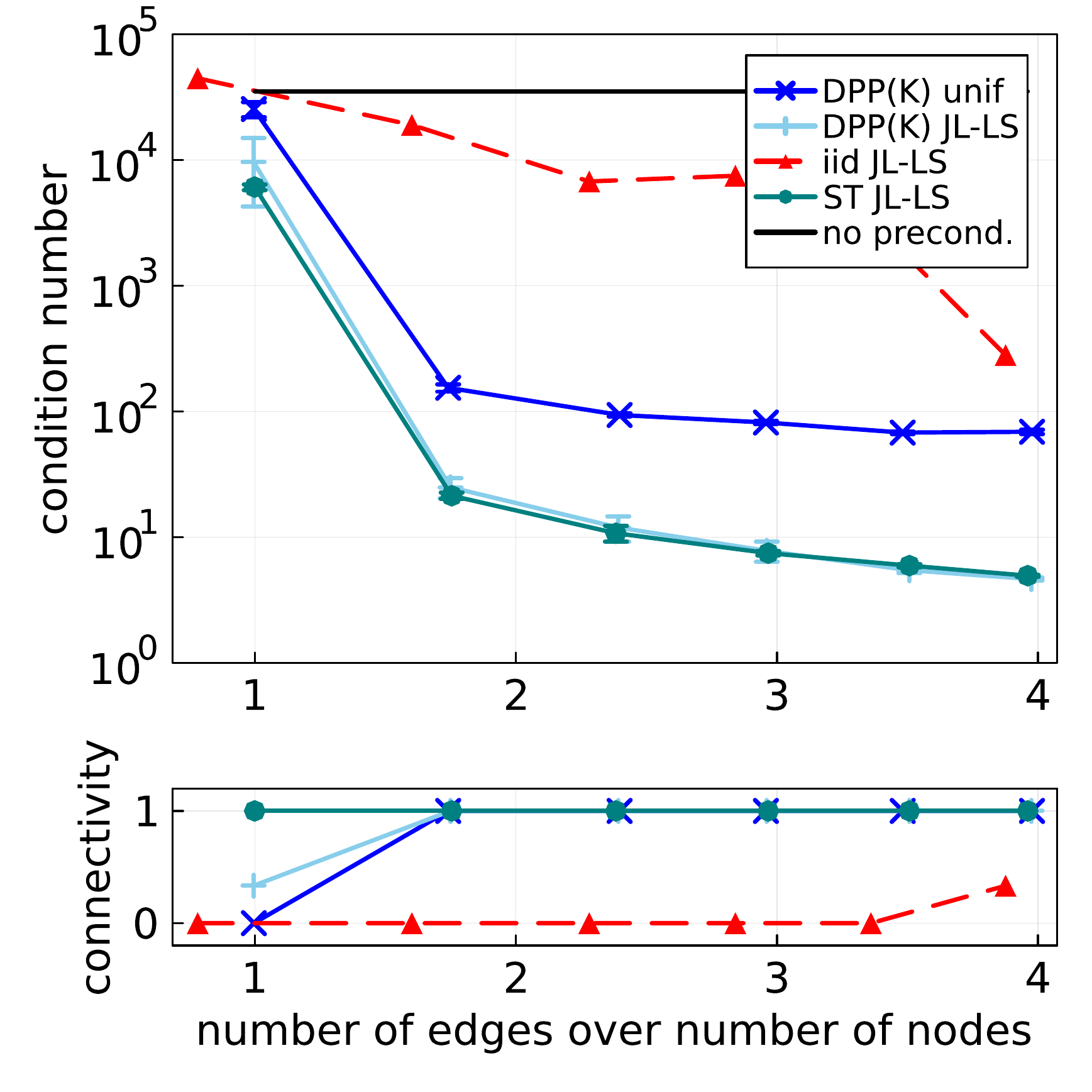}
        \caption{$\cond((\widetilde{\Lambda} + q \I_n)^{-1}(\Lambda + q \I_{n}))$ for Polblogs: $q=0.01$.}
    \end{subfigure}
    \begin{subfigure}[b]{0.49\textwidth}
        \centering
        \includegraphics[scale=0.4]{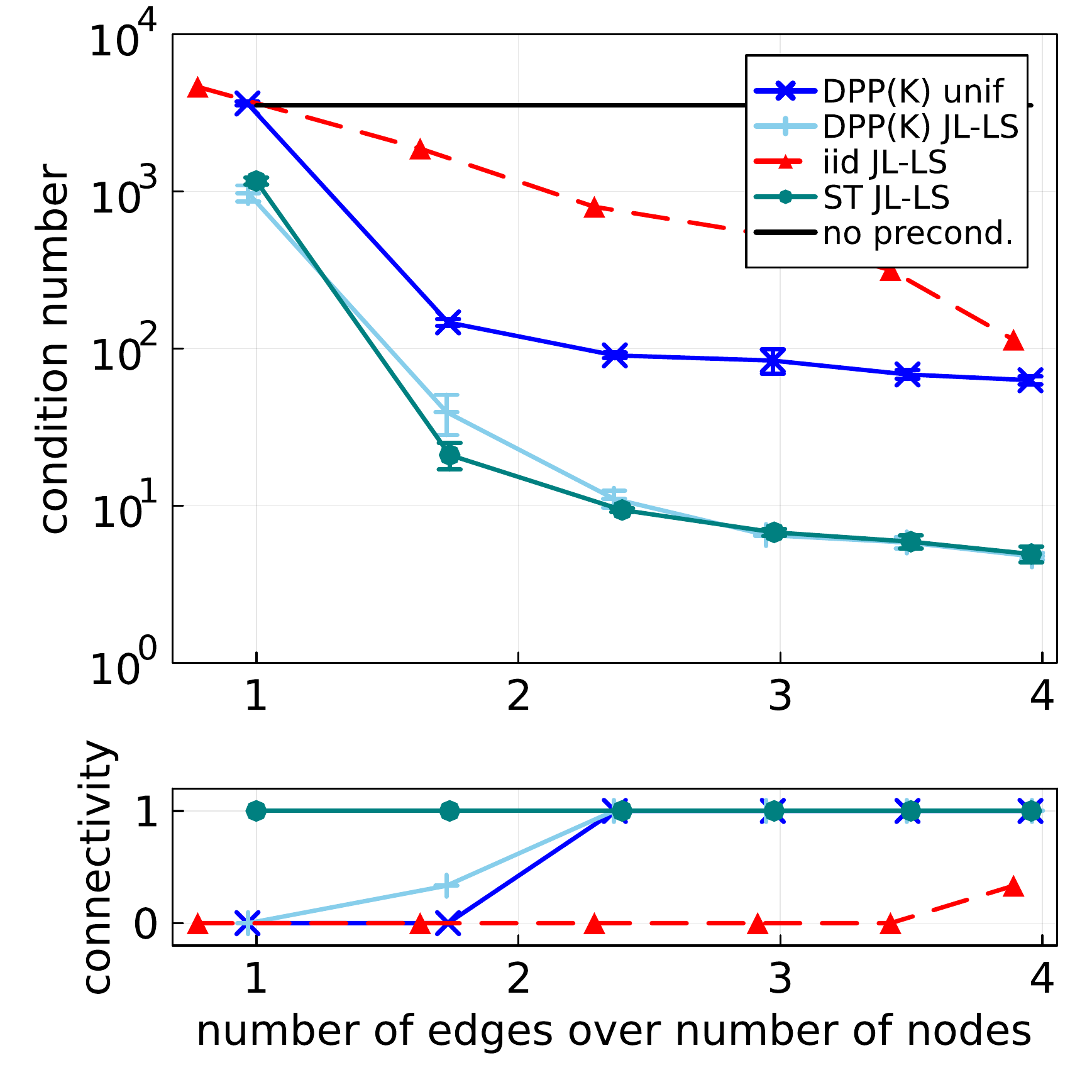}
        \caption{$\cond((\widetilde{\Lambda} + q \I_n)^{-1}(\Lambda + q \I_{n}))$ for Polblogs: $q=0.1$.}
    \end{subfigure}
    \caption{Sparsify-and-precondition the regularized combinatorial Laplacian of the Polblogs graph. 
    We report $\cond((\widetilde{\Lambda} + q \I_n)^{-1}(\Lambda + q \I_{n}))$ for several methods as a function of $n$ divided by the number of edges in the sparsifier, whereas the second row displays the connectivity of the subgraphs. 
    These figures report averages over $3$ independent runs.
    }
    \label{fig:cond_number_Polblogs}
\end{figure}

\subsection{Empirical sampling time for the cycle popping algorithm \label{sec:emp_sampling_time}}
In this section, we compare the sampling time of CRSFs and SFs thanks to $\cyclepopping{}$ to the sampling time of STs with Wilson's algorithm. 
{The aim of this section is to illustrate on an academic example of small size a scenario where sampling STs w.r.t.\ the uniform measure can be slower than sampling CRSFs or SFs w.r.t.\ the desired measure \cref{eq:proba_MTSF}.}
This relative comparison is interesting since we use a similar algorithm for sampling these subgraphs.
Nevertheless, the overall time scale might not be reliable since our implementation is not optimized.
A difference in sampling times can be visible for well-chosen graphs of rather small sizes.

We consider a graph -- which is a variant of  barbell graph that we denote by Barbell$(n)$ -- made of two cliques of $n/2$ nodes (for even $n$), namely two fully connected components, which are connected by \emph{only one} edge. This type of graphs with two cliques linked by a line graph is a worst case example for Wilson's cycle popping algorithm in terms of expected runtime\footnote{The mean hitting time for sampling a rooted spanning tree in a barbell graph with two cliques of size $n/3$ linked by a line graph with $n/3$ nodes is $\Omega(n^3)$ when the initial node and the root are sampled according to the stationary distribution of the walk; see \citet[Example 5.11]{aldous-fill-2014} for more details.}; see~\citet{GH20} for the partial rejection sampling viewpoint. This graph is considered for estimating the average time for sampling a random SF versus sampling a ST. In \cref{fig:cliques_time_SFs}, we report the average sampling time over $1000$ runs and the error bars are the standard errors of the mean.
\begin{figure}
    \centering
    \begin{subfigure}[b]{0.49\textwidth}
        \centering
        \includegraphics[scale=0.4]{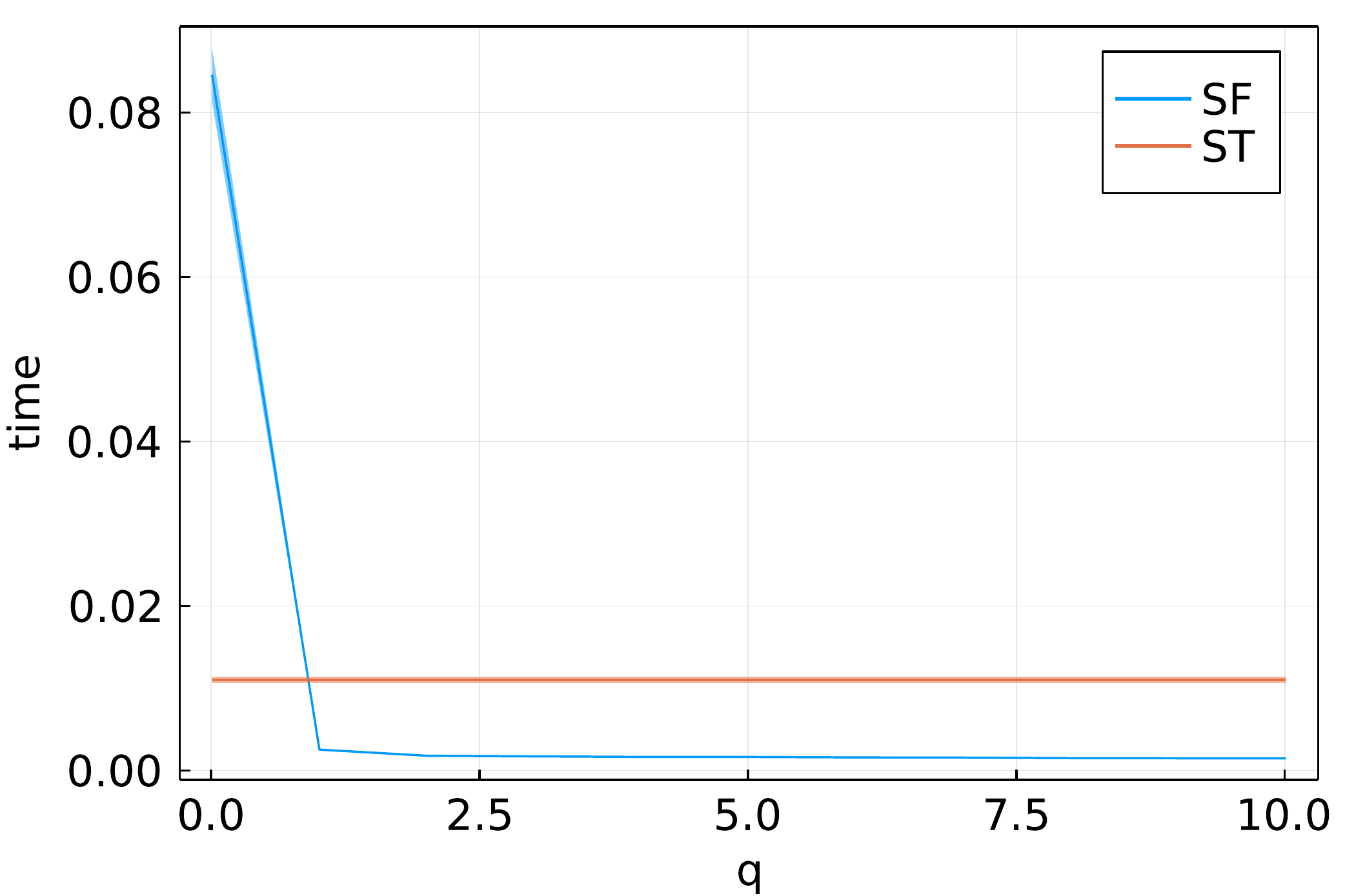}
        \caption{Average sampling time (s) \emph{vs}.\ $q$.}
    \end{subfigure}
    \hfill
    \begin{subfigure}[b]{0.49\textwidth}
        \centering
        \includegraphics[scale=0.4]{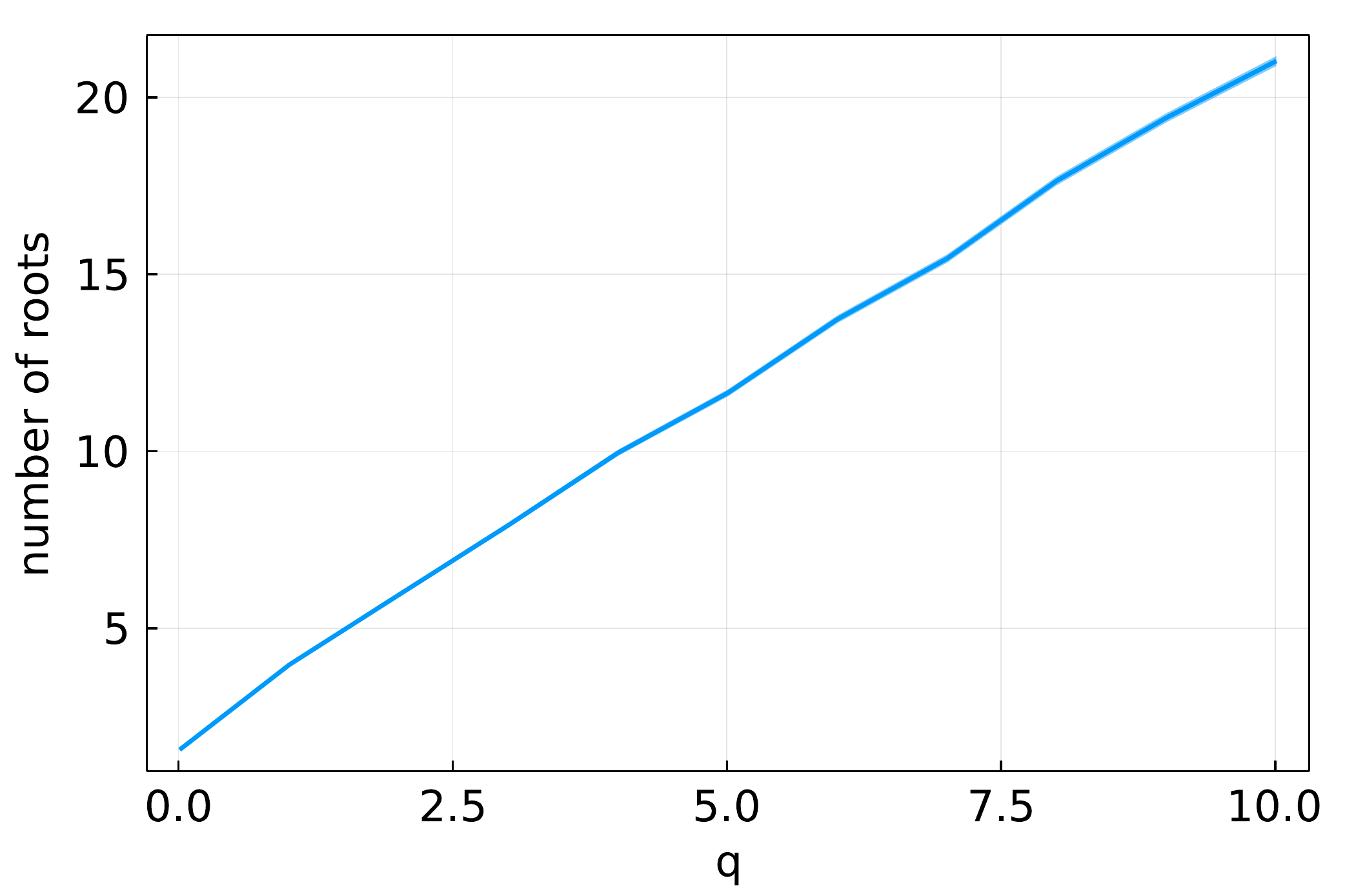}
        \caption{Number of roots in the SF \emph{vs}.\ $q$.}
    \end{subfigure}
    \caption{Comparison of sampling times for SFs w.r.t.\ \cref{eq:proba_MTSF}  and uniform STs for a Barbell$(n)$ graph with $n=500$.  \label{fig:cliques_time_SFs}}
\end{figure}
Our empirical observation is that sampling an SF might be faster than sampling an ST if the graph has dense communities linked with bottlenecks. This is particularly visible if $q$ has a large value. This phenomenon is observed in the extreme case of the barbell graph in \cref{fig:cliques_time_SFs}.

To perform a similar comparison of sampling times for CRSFs and SFs, the same barbell graph is endowed with a connection as follows. Let $\bm{h}$ be the planted ranking score such as defined in \cref{sub:settings_and_baselines}.
For each edge $uv$ in the barbell graph, with probability $1-\eta$, these edge comes with an angle $\vartheta(uv) = (h_u-h_v)/(\pi(n-1))$ and $\vartheta(uv)=\epsilon_{uv}/(\pi(n-1))$ with probability $\eta$ where $\epsilon_{uv}$ drawn from the discrete uniform distribution on $\{-n+1, \dots, n-1\}$. Once more, we take $\vartheta(vu) \triangleq -\vartheta(uv)$.
The resulting connection graph is denoted by Barbell$(n,\eta)$.

In \cref{fig:cliques_time_CRSFs}, on the left-hand side, we display the average sampling time over $5000$ runs for CRSFs and STs as a function of the noise parameter $\eta$. On the right-hand side, we report the average number of components (CRTs) in the sampled CRSFs as a function of $\eta$.
The error bars are the standard errors of the mean.

Empirically, we observe that the average sampling time for a CRSF indeed decreases as $\eta$ increases whereas the number of cycle-rooted trees also increases. Informally, for small values of $\eta$, the average time for sampling a CRSF in this case might be large since $\cyclepopping{}$ will pop many cycles before finally accepting one. This behavior is rather similar to the case of SFs in \cref{fig:cliques_time_SFs}.
\begin{figure}
    \centering
    \begin{subfigure}[b]{0.49\textwidth}
        \centering
        \includegraphics[scale=0.4]{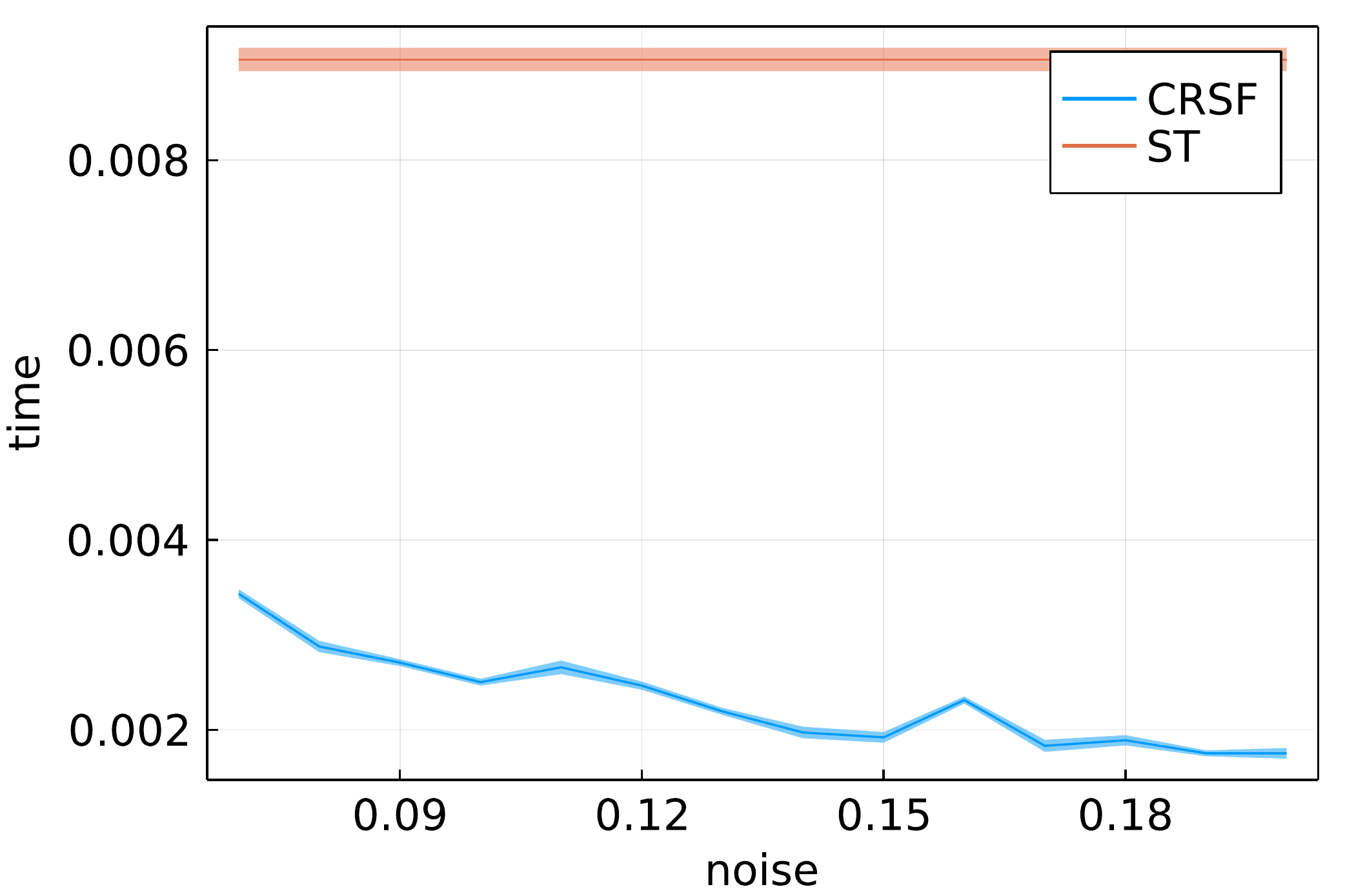}
        \caption{Average sampling time (s) \emph{vs}.\ $\eta$.}
    \end{subfigure}
    \hfill
    \begin{subfigure}[b]{0.49\textwidth}
        \centering
        \includegraphics[scale=0.4]{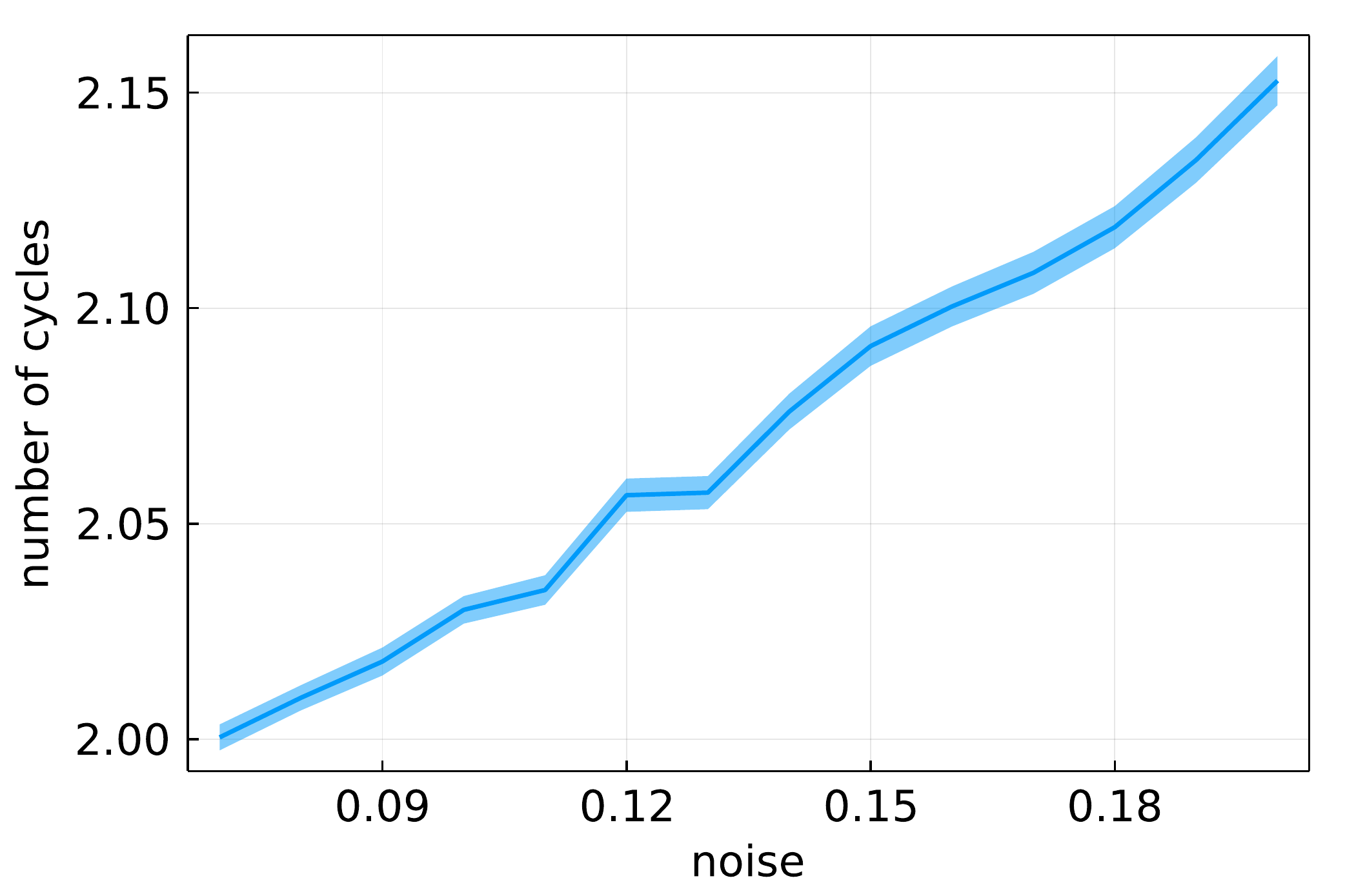}
        \caption{Number of CRTs \emph{vs}.\ $\eta$.}
    \end{subfigure}
    \caption{Comparison of sampling times of CRSFs w.r.t.\ \cref{eq:proba_MTSF}  and uniform  STs for a Barbell$(n,\eta)$ with $n = 500$.
     \label{fig:cliques_time_CRSFs}}
\end{figure}
\subsection{Estimation of leverage scores}
{In this section, we first provide a comparison between exact LSs and empirical estimates, which shows that our sampling algorithm has the correct marginal probabilities.
Second, we show that the Johnson-Lindenstrauss estimates of LSs are indeed good approximations of the exact LSs.}

\begin{figure}[ht]
    \centering
    \begin{subfigure}[b]{0.49\textwidth}
        \centering
        \includegraphics[scale=0.4]{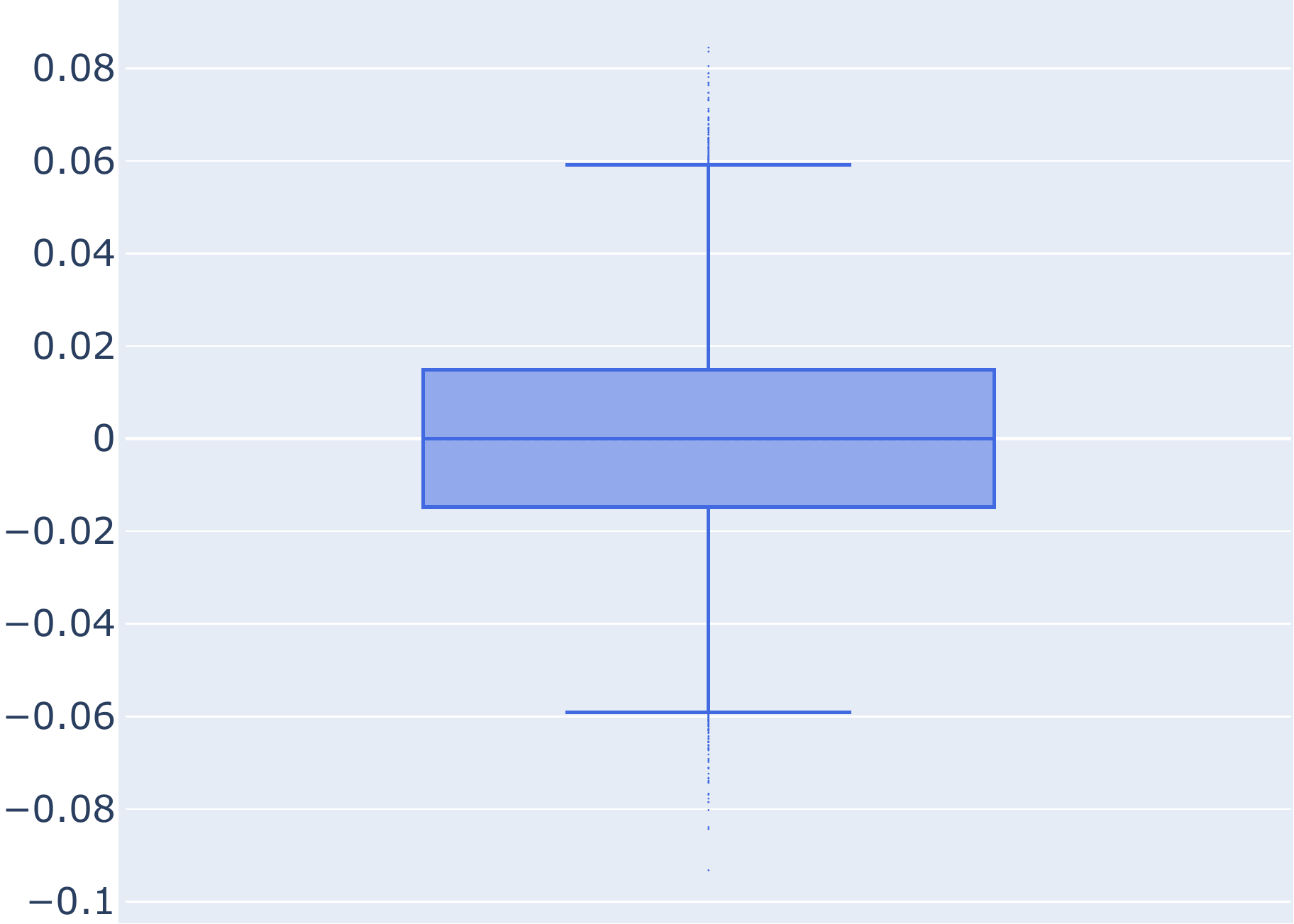}
        \caption{$(\lev(e) -\lev_{\mathrm{emp}}(e))/ \lev(e)$ for MUN: $q = 0$.}
    \end{subfigure}
    \begin{subfigure}[b]{0.49\textwidth}
        \centering
        \includegraphics[scale=0.4]{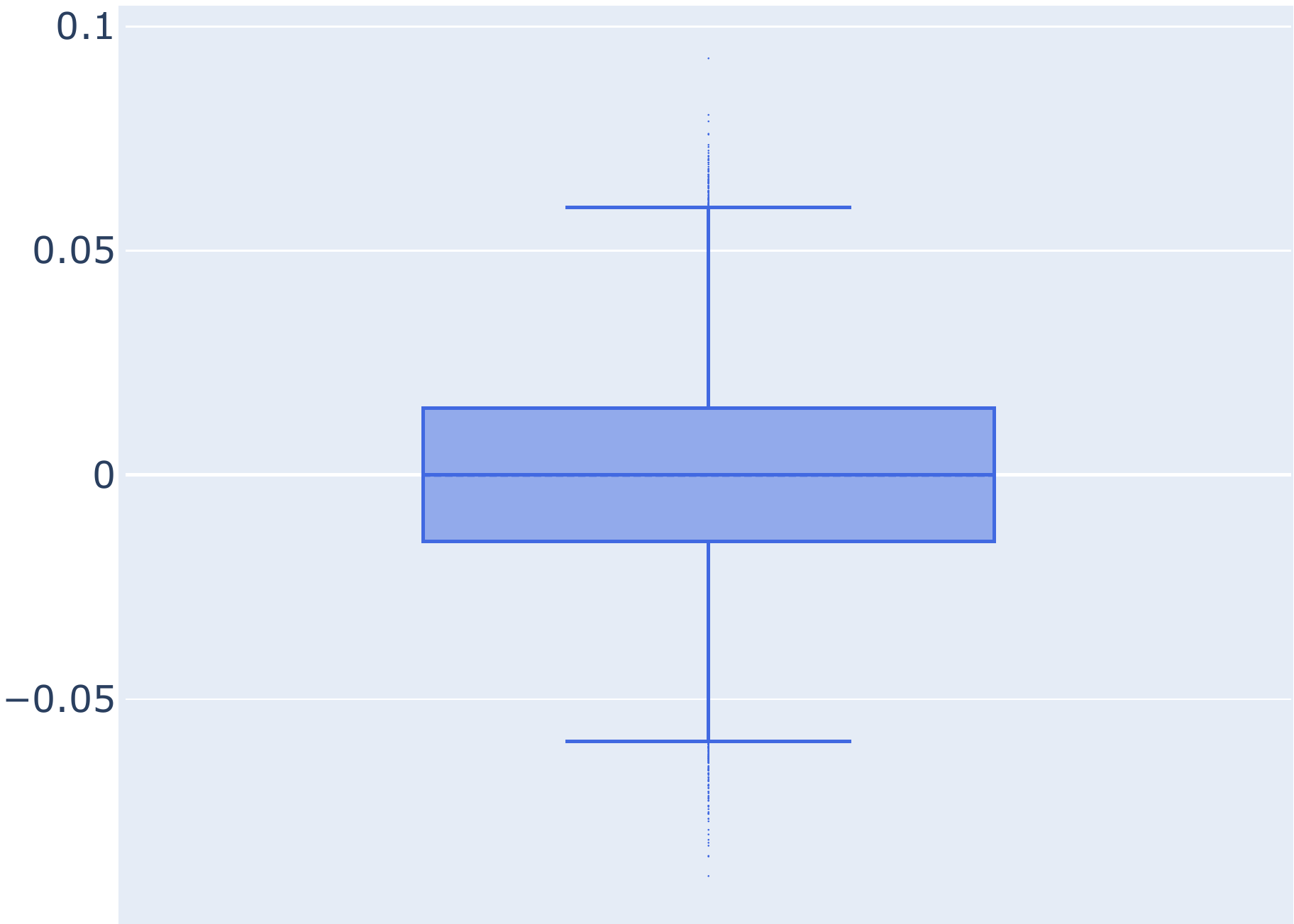}
        \caption{$(\lev(e) -\lev_{\mathrm{emp}}(e))/ \lev(e)$ for MUN: $q = 1$.}
    \end{subfigure}
    \caption{Relative differences between the exact leverage scores and their empirical estimates over $10^5$ Monte Carlo runs of $\cyclepopping{}$ for a MUN$(n,p,\eta)$ graph. }
    \label{fig:JL_lev}
\end{figure}
\subsubsection{Empirical estimates of leverage scores}
In \cref{fig:emp_lev}, we provide an empirical validation of the MTSF sampling algorithm given in \cref{sec:sampling_a_multitype_spanning_forest}. We empirically compute the probability that any edge $e$ belongs to an MTSF $\calC\sim \DPP(K)$, i.e., the leverage score $\Pr(e\in \calC) = K_{ee}$
with $K$ given in \cref{eq:K}; see \cref{sec:dpp_sampling_of_edges__multitype_spanning_forests}. For a fixed graph, we compute an empirical counterpart by calculating the frequencies
\begin{equation}
    \lev_{\mathrm{emp}}(e) = \frac{1}{t}\sum_{\ell=1}^t 1_{\calC_\ell}(e)\label{eq:emp_LS}
\end{equation}
where the MTSFs $\{\calC_\ell\}_{1\leq \ell\leq t}$ are sampled independently thanks to $\cyclepopping{}$ with capped cycle weights, i.e., the weight of a cycle $\eta$ reads $1 \wedge \big(1 - \cos \theta(\eta)\big)$.
No importance weight is considered in \cref{eq:emp_LS}.
\Cref{fig:emp_lev} provides a comparison of the exact and empirical leverage scores by computing a boxplot of relative differences $(\lev(e) -\lev_{\mathrm{emp}}(e))/ \lev(e)$,
for CRSFs ($q=0$) and MTSFs ($q >0$) for a MUN$(n,p,\eta)$ random graph  with $n=500$, $p=0.2$ and $\eta=0.1$.
The empirical leverage scores are computed over $10^5$ runs of $\cyclepopping{}$ for a fixed MUN random graph for $q=0$ (Left-hand side) and $q = 1$ (Right-hand side).
Each boxplot displays the mean (dashed line) and median as well as the $25$-th and $75$-th percentile in the blue box.

We observe that the CRSF and MTSF sampling algorithms yield samples with empirical edge marginal probabilities close to the exact edge marginal probabilities since the distribution of relative differences is centered at $0$ and is rather peaked.
\begin{figure}[ht]
    \centering
    \begin{subfigure}[b]{0.49\textwidth}
        \centering
        \includegraphics[scale=0.4]{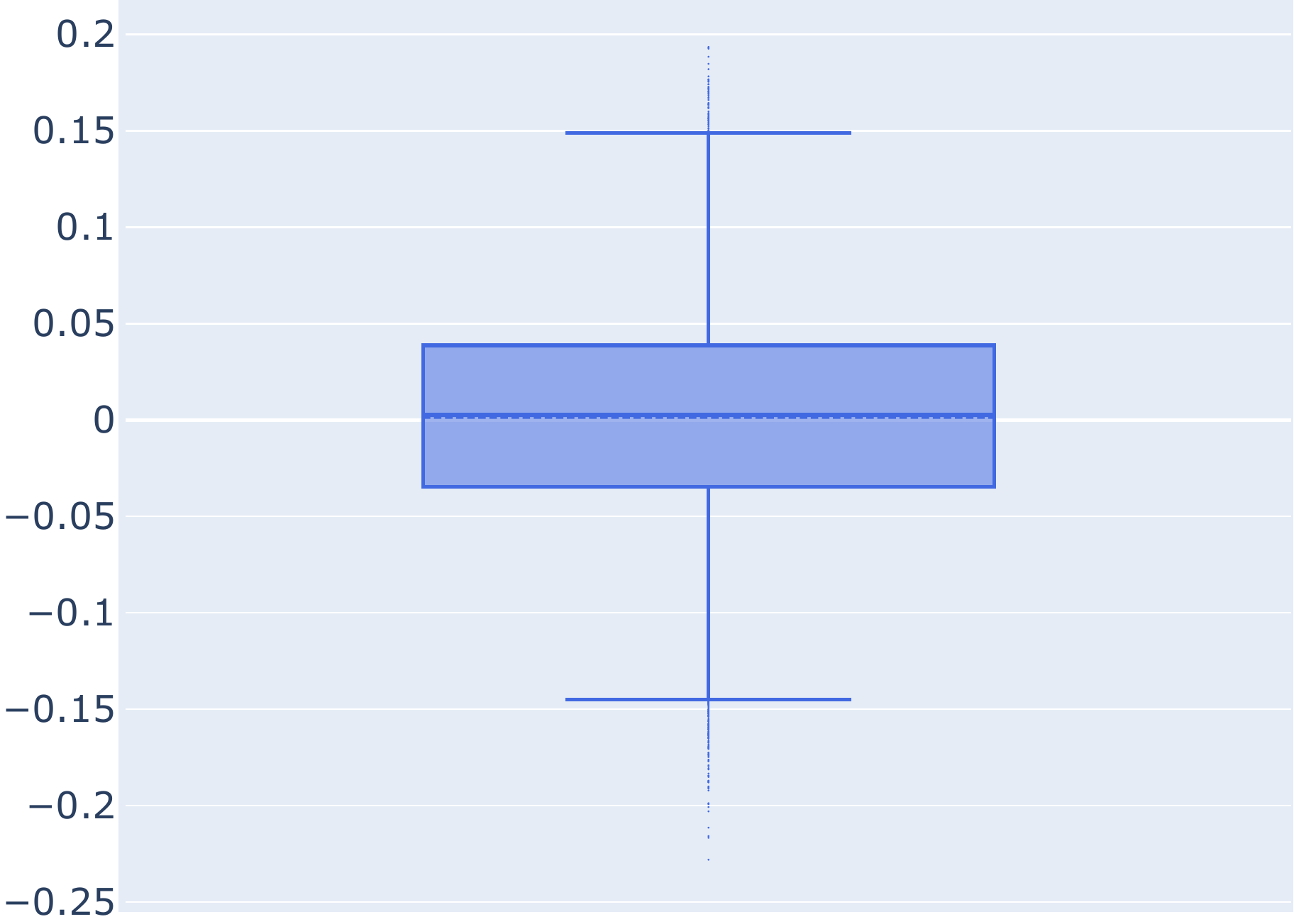}
        \caption{$(\lev(e) -\lev_{\mathrm{JL}}(e)(e))/ \lev(e)$ with $q = 0$.}
    \end{subfigure}
    \begin{subfigure}[b]{0.49\textwidth}
        \centering
        \includegraphics[scale=0.4]{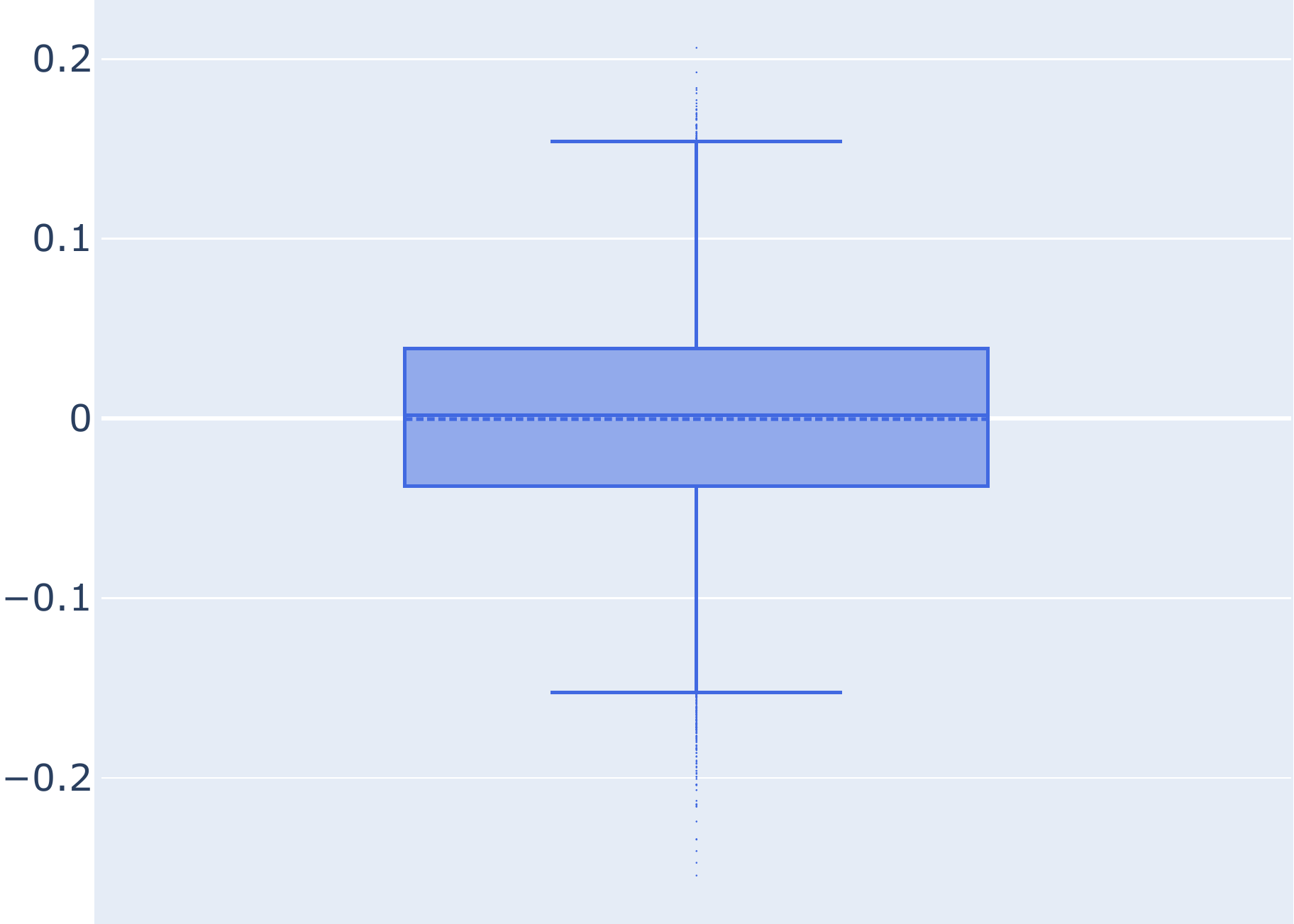}
        \caption{$(\lev(e) -\lev_{\mathrm{JL}}(e)(e))/ \lev(e)$ with $q = 1$.}
    \end{subfigure}
    \caption{Relative differences between sketched and exact leverage scores for all the edges of a MUN$(n,p,\eta)$ graph with $n=500$, $p=0.2$ and $\eta=0.1$.}
    \label{fig:emp_lev}
\end{figure}

\subsubsection{JL estimates of leverage scores}

{To illustrate the accuracy of the LS approximation, we here compare the sketched LSs obtained with JL lemma with their exact values.
The estimates $\lev_{\mathrm{JL}}(e)$ are computed by using \eqref{eq:linear_system_sketch} and \eqref{eq:lev_sketch}.
We compute a boxplot of the relative differences $(\lev(e) -\lev_{\mathrm{JL}}(e))/ \lev(e)$ for a MUN$(n,p,\eta)$ random graph with $n=500$, $p=0.2$ and $\eta=0.1$.
This boxplot is constructed in the same way as in \cref{fig:emp_lev}.
To give an order of magnitude, the number of columns in the JL approximation is then roughly equal to $k = 4\cdot 10^2$ (see \cref{eq:k_for_JL}) over a total of about $2.5\cdot 10^4$ edges.
In \cref{fig:JL_lev}, we observe that JL-estimates give a good approximation since the relative differences have a mean of $2 \cdot 10^{-3}$ and a standard deviation of $6 \cdot 10^{-2}$.
Therefore, in this case, the relative absolute error on the LSs is typically often less than $20\%$.}

\section{Discussion} 
\label{sec:conclusion}
Similarly to how ensembles of uniform spanning trees sparsify the combinatorial Laplacian, we have investigated ensembles of multi-type spanning forests to sparsify the magnetic Laplacian.
The distribution we use to draw one such MTSF is a determinantal point process (DPP), like uniform spanning trees.
Our theoretical results include a Chernoff bound with intrinsic dimension for DPPs, which might be of independent interest, and a confidence ``interval'' based on self-normalized importance sampling to treat graphs with large inconsistencies.

Our experiments suggest that sparsifiers based on CRSFs and MTSFs are the solution of choice when only a few inconsistent cycles are present, so that it is important not to miss them. In this circumstance, the least eigenvalue of the magnetic Laplacian is expected to be small, yielding a large condition number for the system associated with $\Delta$.
In other cases, and in particular when the graph is well-connected, a simpler approach using spanning forests with comparable leverage-score heuristics actually performs on par with our MTSFs.
This good performance of STs may appear surprising at first, given that the distribution used for STs does not rely on the graph's connection.
However, in our experiments, the number of cycles in CRSFs remained relatively low, and it would be interesting to further investigate graph structures that favor a large number of cycles in a CRSF.

From the perspective of the sampling algorithm, the sampling time for the cycle-popping algorithm of Section~\ref{sec:sampling_a_multitype_spanning_forest} is also expected to be large when the least eigenvalue of $\Delta$ is small, i.e., for a low level of inconsistency.
Indeed, if the cycles have low inconsistencies, the random walk is likely to spend a long time popping cycles before one cycle is successfully accepted.
In this case, the aforementioned ``generic'' DPP sampling procedure by \citet{HKPV06} might also be costly or be prone to numerical errors since it requires the full eigenvalue decomposition of the correlation kernel $K = B\Delta^{-1}B^*$; the latter eigenvalue problem being also possibly ill-conditioned due to the factor $\Delta^{-1}$.
{Although they do not need an eigendecomposition, the algebraic DPP sampling algorithms such as in \citep{launay_galerne_desolneux_2020} or \citep{poulson2020high} are also likely to be prone to numerical errors since they used slight modifications of $LU$ and $LDL^\top$ decompositions.}

Let us list below a few remarks about technological aspects of MTSF sampling.
\begin{itemize}
    \item A computational advantage of $\cyclepopping{}$ over the algebraic algorithm of \citet{HKPV06} is that it is decentralized in the sense that the random walker only needs to query the knowledge of the neighbouring nodes and edges.
    Hence, the knowledge of the full connection graph is not necessarily needed for running $\cyclepopping{}$ and the random walker can discover the graph on the fly. This may  reduce memory requirements.
    \item The output of $\cyclepopping{}$ is also \emph{necessarily} an MTSF \emph{by design}.
    This makes $\cyclepopping{}$ less sensitive to numerical errors compared with the generic $\HKPV$ algorithm of \citet{HKPV06} in this specific case, which relies on an iterative linear algebraic procedure. In other words, in case of an ill-conditionned $\Delta$, it is not excluded that an algebraic sampler outputs a graph which is not an MTSF due to numerical errors. 
    A few numerical results in this direction are given in \cref{sec:HKPV_extra_results}.
    \item $\cyclepopping{}$ relies on the condition that all  cycles in the connection graph are weakly inconsistent, which seems a rather arbitrary limitation.
    Currently, we have no strategy for checking that this condition is satisfied.
    As a possible future work, extending $\cyclepopping{}$ beyond the case of weakly inconsistent connection graphs would be an interesting contribution to avoid using self-normalized importance sampling for which only weaker statistical guarantees for sparsification are given in \cref{prop:asymp_confidence}.
    \item As illustrated in our simulations, leverage scores are important in order to sparsify Laplacians of graphs with a non-concentrated degree distribution, such as in the case of the Polblogs network of \cref{fig:cond_number_Polblogs}.
    A prospect for future work consists in developping a leverage score approximation scheme for which some statistical guarantees can be derived.
\end{itemize}

Our implementation of $\cyclepopping{}$ is currently far from optimal. As a perspective, an optimized implementation of $\cyclepopping{}$ would have to be developed in order to compare standard linear algebraic solvers with our simulation results in terms of computational time.
As another prospect, we can investigate the use of $\cyclepopping{}$ in the context of cycle-based synchronization in the light of \citep{lerman2022robust}.

\paragraph*{Acknowledgements.}
We thank Guillaume Gautier for numerous discussions, and his precious support for Julia implementations.
The paper further benefited from comments of Simon Barthelmé, and from discussions on partial rejection sampling.
We also thank Nicolas Tremblay for discussions about spectral sparsification and suggestions concerning the numerical simulations. Furthermore, we thank Hugo Jaquard for pointing out typos.

We acknowledge support from ERC grant BLACKJACK (ERC-2019-STG-851866) and ANR AI chair
BACCARAT (ANR-20-CHIA-0002).

%
\appendix
\FloatBarrier

\section{Deferred results and  proofs} 
\label{sec:deferred_proofs}

\subsection{Influence of $q$} 

\begin{proposition}\label{prop:decreasing_deff}
    {
    Let \cref{assump:non-singularity} hold.
    The functions $\deff = \Tr(\Delta(\Delta + q \I)^{-1})$,  $\mumax = \|\Delta(\Delta + q \I)^{-1}\|_{\op}$, and $\deff/\kappa$ 
    are decreasing functions of $q$.}
\end{proposition}
\begin{proof}
    {Let $\lambda_\ell$ for $1\leq \ell \leq n$ be the eigenvalues of $\Delta$, including possible multiplicities and sorted in increasing order. 
    Denote $\lambda_{\max} = \lambda_n$.
    The following function
    $
        \deff = \frac{\lambda_1}{\lambda_1 + q} + \dots + \frac{\lambda_{\max}}{\lambda_{\max} + q}
    $
    is clearly decreasing. Similarly,
    $
        \kappa = \lambda_{\max}/(\lambda_{\max} + q)
    $
    is also decreasing.}
    {Furthermore, we have 
    \[
        \deff/\kappa = \sum_{\ell} \frac{\lambda_{\ell}/(\lambda_{\ell} + q)}{\lambda_{\max}/(\lambda_{\max} + q)} = \sum_{\ell} \frac{\lambda_{\ell}}{\lambda_{\max}} \frac{\lambda_{\max} + q}{\lambda_{\ell} + q} = \sum_{\ell} \frac{\lambda_{\ell}}{\lambda_{\max}} \left(1+ \frac{\lambda_{\max} - \lambda_{\ell} }{\lambda_{\ell} + q}\right).
    \]
    Since $\lambda_{\max} \geq \lambda_{\ell}$, the function $\deff/\kappa$ is also decreasing.}
\end{proof}

\subsection{Proof of \cref{prop:CholeskyMTSF}} 
\label{sec:cholesky_decomposition_of_the_magnetic_laplacian_of_a_mtsf}
The proof technique is directly adapted from~\citet[Theorem 13.3]{Vishnoi2013}. It relies on the observation that the Cholesky decomposition of a matrix $A$ follows from the Schur complement formula
\begin{equation}
    \label{eq:Schur}
    \begin{pmatrix}
        d      & \bm{u}^* \\
        \bm{u} & A
    \end{pmatrix} = \begin{pmatrix}
        1        & 0^\top \\
        \bm{u}/d & \I
    \end{pmatrix}\begin{pmatrix}
        d & 0^\top              \\
        0 & A-\bm{u}\bm{u}^* /d
    \end{pmatrix}
    \begin{pmatrix}
        1 & \bm{u}^*/d \\
        0 & \I
    \end{pmatrix},
\end{equation}
which corresponds to the elimination of the first row and column~\citep[Proof of Theorem 13.1]{Vishnoi2013}. In what follows, we shall associate to a sparse symmetric matrix $A$ a graph which has an edge for each non-zero off-diagonal element of $A$. This matrix is equal to a block of the regularized Laplacian at the first iteration of the decomposition algorithm and is updated at the subsequent iterations.
A formula akin to~\cref{eq:Schur} holds for the elimination of the $i$-th row and column.
Recursively applying this formula, the Cholesky decomposition appears as a product of triangular matrices.

In our case, we first note that each connected component of the MTSF can be treated separately since, by an appropriate permutation, the regularized magnetic Laplacian of the MTSF $\widetilde{\Delta} + q \mathbb{I}$ can be reduced to have a diagonal block structure.
We thus henceforth assume that $\calC$ has a unique cycle, of length $c > 0$.
The idea of the proof is to find a permutation matrix $Q$ such that the Cholesky decomposition of $Q(\widetilde{\Delta} + q \mathbb{I})Q^\top$ has at most $\mathcal{O}(n + c)$ non-zero off-diagonal entries.
This permutation corresponds to the desired reordering of the nodes.

    We now construct the permutation matrix $Q$.
    Each connected component of a MTSF is either a tree or a cycle-rooted tree.
    Consider first the case of a rooted tree.
    We can proceed with the elimination by \emph{peeling off} the leaves of the tree.
    Notice that after eliminating a leaf, the resulting structure of $A-\bm{u}\bm{u}^* /d$ is a tree with one missing leaf.
    This procedure specifies an ordering of the nodes, and therefore the permutation matrix Q.
    Since each leaf only has one neighbor, the vector $u$ only has one non-zero entry and therefore, the corresponding column of the triangular matrix $\left(
        \begin{smallmatrix}
                1           & 0^\top \\
                \bm{u}/d    & \I
            \end{smallmatrix}
        \right)$
    in \cref{eq:Schur} has only one non-zero off-diagonal entry.
    For the same reason, the subtraction $A-\bm{u}\bm{u}^* /d$ requires $\mathcal{O}(1)$ operations.
    By computing the product of triangular matrices obtained by eliminating one leaf after the other, we see that each elimination corresponds to the update of one column in the final triangular matrix.
    Thus, the Cholesky decomposition of the Laplacian of a rooted tree of $n_{t}$ nodes has at most $n_{t} -1$ non-zero off-diagonal entries, and requires $\mathcal{O}(n_{t})$ operations.

    Second, consider the case of a cycle-rooted tree, with one cycle of $n_i$ nodes and $n_{crt}$ nodes in total.
    We proceed similarly by eliminating one leaf after the other until only the cycle remains.
    This first stage yields $n_{crt}-n_i$ non-zero off-diagonal entries.
    Next, we eliminate the nodes of the cycle.
    For each such node, the vector $\bm{u}$ in \cref{eq:Schur} has two non-zero entries since each node in a cycle has two neighbors.
    Furthermore, if $\bm{u}$ has two non-zero entries, the matrix $\bm{u}\bm{u}^*$ has $4$ non-zero entries and the corresponding update costs $\mathcal{O}(1)$ operations.
    Hence, if $n_{i}>3$, eliminating a node in a cycle with $n_{i}$ nodes yields another cycle with $n_{i}-1$ nodes.
    This continues until only two nodes remain.
    Eliminating one of them yields one non-zero entry in the triangular matrix.
    Thus, in total, eliminating a cycle with $n_i\geq 3$ nodes yields $2 (n_i - 2) + 1$ non-zero entries.
    So, for a cycle-rooted tree of $n_{crt}$ nodes in total with a cycle of $n_i$ nodes, we have at most: $n_{crt}-n_i$ non-zero entries for the nodes in the trees and $2 n_i -3$ non-zero entries for the nodes in the cycle.
    Thus, in total, there are at most $n_{crt} + n_i -3$ non-zero entries in the triangular matrix  and this costs $\mathcal{O}(n_{crt} + n_i)$ operations.

    By combining the counts for the rooted trees and cycle-rooted trees, we obtain the desired result.
    Note that the diagonal of $\widetilde{\Delta} + q \mathbb{I}$ does not influence the sparsity of the decomposition.
\subsection{Proof of \cref{lem:consequence_of_avg_mult_indpt}} \label{sec:proof_of_inf_am}
Let $\pp$ be a $k$-homogeneous point process on $[m]$ for some integer $k \geq 1$.
Recall the intensity of $\pp$ is the probability mass function $\nu_\pp(e) = \rho_\pp(e)/k$ for all $e\in [m]$, as given in \cref{sec:Intensity_Density}.
First, we notice that for all  $i$ and $j$ with $i\neq j$ such that $\rho_\pp(i) > 0$ and $\rho_\pp(j) > 0$, we have
\begin{equation}
    \nu_\pp(i) \rho_{\pp_i}(j) = \nu_\pp(j) \rho_{\pp_j}(i).\label{eq:bayes_p_nu}
\end{equation}
Now, recall the notation
\[
    \delta_{\pp,e}(i) = \begin{cases}
        1 - \rho_\pp(e) & \text{ if } i = e,\\
        \rho_{\pp_e}(i) - \rho_{\pp}(i) & \text{ otherwise.}
    \end{cases}
\]
If $\rho_\pp(i) > 0$, it holds that
\begin{align*}
    \E_{e \sim \nu_\pp} | \delta_{\pp,e}(i) | &= \nu_\pp(i) |1 - \rho_{\pp}(i)| +  \sum_{\substack{e\neq i\\ \rho_{\pp}(e) >0}} \nu_\pp(e) | \rho_{\pp_e}(i) - \rho_{\pp}(i) | \\
    &= \nu_\pp(i) |1 - \rho_{\pp}(i)| + \sum_{\substack{e \neq i\\ \rho_{\pp}(e) >0}} |\nu_\pp(e) \rho_{\pp_e}(i) - \nu_\pp(e)\rho_{\pp}(i) | \\
    &= \nu_\pp(i) |1 - \rho_{\pp}(i)| + \sum_{\substack{e \neq i\\ \rho_{\pp}(e) >0}} |\nu_\pp(i) \rho_{\pp_i}(e) - \frac{\rho_{\pp}(e)\rho_{\pp}(i)}{k} | \quad \text{(by using \cref{eq:bayes_p_nu})}\\
    & = \nu_\pp(i) |1 - \rho_{\pp}(i)| + \nu_\pp(i)\sum_{\substack{e \neq i\\ \rho_{\pp}(e) >0}} | \rho_{\pp_i}(e) - \rho_{\pp}(e) | \\
    & = \nu_\pp(i) |1 - \rho_{\pp}(i)| + \nu_\pp(i)\sum_{e \neq i} | \rho_{\pp_i}(e) - \rho_{\pp}(e) |\\
    &\leq \nu_\pp(i) \dinf
\end{align*}
where the last inequality is due to \cref{def:pp_parameter}. This concludes the first part of the proof.

For the second part of the proof, we simply use that $K^2 =K$ since $\DPP(K)$ is projective and we directly compute the
 desired expectation by recalling that $|\delta_e(e)| = 1 - K_{ee}$ and $|\delta_e(i)| = K_{ei} K_{ie}/K_{ee}$ for $e\neq i$.
The conclusion follows from simple algebra.

\subsection{Proof of \cref{lem:Fact5.4}}\label{sec:proof_of_Fact5.4}
We only introduce minor changes in a proof from \citet[arxiv v2]{KKS22}.
Let $e\in [m]$ such that $\rho_\pp(e) >0$.
Recall the following expression
\[
    Z_e = \E_{\pp_e} \left[Y_e + \sum_{i \in \pp_e}  Y_i\right] - \E_{\pp} \left[\sum_{i\in \pp} Y_i\right].
\]
(i) First, we simply rewrite $Z_e$ as follows
\begin{align*}
    Z_e = Y_e + \sum_{i \neq e} \rho_{\pp_e}(i) Y_i - \sum_{i} \rho_\pp(i) Y_i  = (1-\rho_\pp(e))Y_e + \sum_{i \neq e} \left(\rho_{\pp_e}(i) - \rho_\pp(i)\right) Y_i.
\end{align*}
The latter expression is used to upper bound $|Z_e|$,
where $|\cdot|$ is the matrix absolute value. We begin by noticing that $ \lambda_{\max}(|Z_e|) \leq \|Z_e\|_{\op}$.
Next, in light of \cref{def:pp_parameter}, we can control the following quantity as follows
\[
    |1-\rho_\pp(e)| + \sum_{i \neq e} |\rho_{\pp_e}(i) - \rho_\pp(i)| \leq \dinf.
\]
By using $Y_i \preceq r \I$ for all $i\in[m]$  and the triangular inequality, we find
\begin{align*}
    |Z_e| \preceq \|Z_e\|_{\op} \I \preceq |1-\rho_\pp(e)|\cdot \|Y_e\|_{\op} \I+ \sum_{i \neq e} |\rho_{\pp_e}(i) - \rho_\pp(i)| \cdot \|Y_i\|_{\op} \I \preceq  \dinf r \I.
\end{align*}
(ii) We now consider the square of $Z_v$ and we bound its expression by using the identity $A B + B A \preceq A^2 + B^2$. To simplify the expressions, we recall the notation
\[
    \delta_{\pp,e}(i) = \begin{cases}
        1 - \rho_\pp(e) & \text{ if } i = e,\\
        \rho_{\pp_e}(i) - \rho_{\pp}(i) & \text{ otherwise.}
    \end{cases}
\]
Define $s_i = {\rm sign}(\delta_{\pp,e}(i))$. By expanding the square, we find
\begin{align*}
    |Z_e|^2 &\preceq  \sum_{i} \left|\delta_{\pp,e}(i) \right|^2 Y_i^2 + \sum_{i} \sum_{j < i} |\delta_{\pp,e}(i)|\cdot |\delta_{\pp,e}(j)| (s_i Y_i s_j Y_j + s_j Y_j s_i Y_i)\\
    & \preceq \sum_{i} \left|\delta_{\pp,e}(i) \right|^2 Y_i^2 + \sum_{i} \sum_{j < i} |\delta_{\pp,e}(i)|\cdot |\delta_{\pp,e}(j)| (Y_i^2 + Y_j^2)\\
    & \preceq \sum_{i} \left|\delta_{\pp,e}(i) \right|^2 r Y_i + \sum_{i}\sum_{j < i}|\delta_{\pp,e}(i)| \cdot  |\delta_{\pp,e}(j)| rY_i + \sum_{i}\sum_{j < i} |\delta_{\pp,e}(i)| \cdot |\delta_{\pp,e}(j)|  rY_j\\
    &= \sum_{i} \left|\delta_{\pp,e}(i) \right|^2 rY_i + \sum_{i}\sum_{j < i}|\delta_{\pp,e}(i)| \cdot  |\delta_{\pp,e}(j)| r Y_i  + \sum_{j}\sum_{i < j} |\delta_{\pp,e}(i)| \cdot |\delta_{\pp,e}(j)|  r Y_i\\
    &= \sum_{j} \left|\delta_{\pp,e}(j) \right| \sum_i  |\delta_{\pp,e}(i)|  r Y_i \\
    &\preceq \dinf r \sum_i  |\delta_{\pp,e}(i)|  Y_i,
\end{align*}
where the indices were renamed in the next-to-last equality, and where we used $Y_i\preceq r \I$.
Now, recall that $\E_{e\sim \nu }|\delta_{\pp,e}(i)| \leq \dinf\nu_\pp(i)$ for all $i$ in the support of $\rho_\pp$ in light of \cref{lem:consequence_of_avg_mult_indpt}.
Finally, we compute the following expectation of the last bound,
\begin{align*}
    \E_{e\sim \nu } [|Z_e|^2] &\preceq r\dinf  \sum_{i}\E_{e\sim \nu }|\delta_{\pp,e}(i)| Y_i \preceq r\dinf^2 \sum_{i}\nu_\pp(i) Y_i = r\dinf^2 \E_{e\sim \nu_\pp } [Y_e].
\end{align*}
This is the desired result.
\FloatBarrier
\section{Additional numerical simulations} 
\label{sec:additional_numerical_simulations}
\subsection{Consistency of cycles for the random graphs of \cref{sec:exp_SyncRank}}\label{sec:cycle_inconsistencies}

In \cref{fig:MUN_weights} and \cref{fig:ERO_weights}, we display in boxplots the average cycle importance weights
$
\prod_{\text{cycles } \eta} \left(1 \vee (1 - \cos \theta(\eta))\right)
$
of the CRSFs sampled by $\cyclepopping{}$ and used for the sparsification of the magnetic Laplacian in \cref{fig:approx_mg_laplacian}.
These weights are used in the self-normalized importance sampling estimate; see \cref{eq:importance_weight}.
A weakly inconsistent cycle has a unit weight, whereas the largest possible weight is two.
Thus, the MUN example is likely to be a case where the assumption of \emph{weak inconsistencies} is satisfied and $\cyclepopping{}$ output is correctly distributed.

Clearly, the inconsistent cycles in the MUN model have smaller weights compared with the case of the ERO model, since the distribution is centered at $1$ in \cref{fig:MUN_weights} and is close to $2$ in \cref{fig:ERO_weights}.
Also, we computed the distribution of cycles (or CRTs) in the CRSFs sampled in the same setting as for the simulation of \cref{fig:approx_mg_laplacian}, thanks a Monte Carlo estimation over $100$ independent executions. 
This empirical distribution is reported in \cref{fig:cycle_weights}.
In \cref{fig:approx_mg_laplacian_MUN_CRTs}, we see that CRSFs sampled in the MUN graph often have only one component, whereas CRSFs sampled in the ERO graph have much more components; see \cref{fig:approx_mg_laplacian_ERO_CRTs}.
This confirms that inconsistencies are larger in the former case.

\begin{figure} 
    \begin{subfigure}[b]{0.49\textwidth}
        \centering
        \includegraphics[scale=0.3]{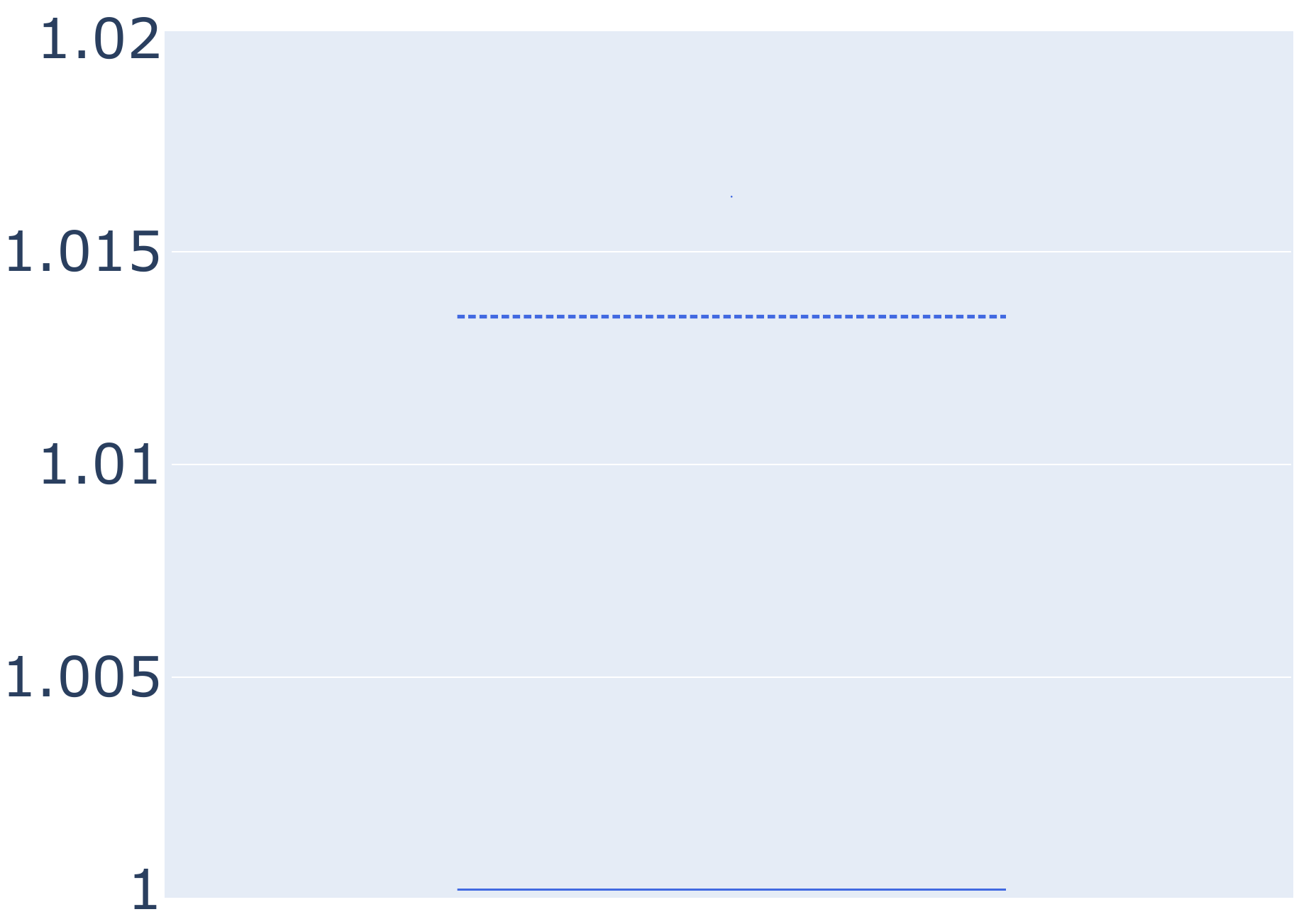}
        \caption{Average cycle weight (MUN). \label{fig:MUN_weights}}
    \end{subfigure}
    \hfill
    \begin{subfigure}[b]{0.49\textwidth}
        \centering
        \includegraphics[scale=0.3]{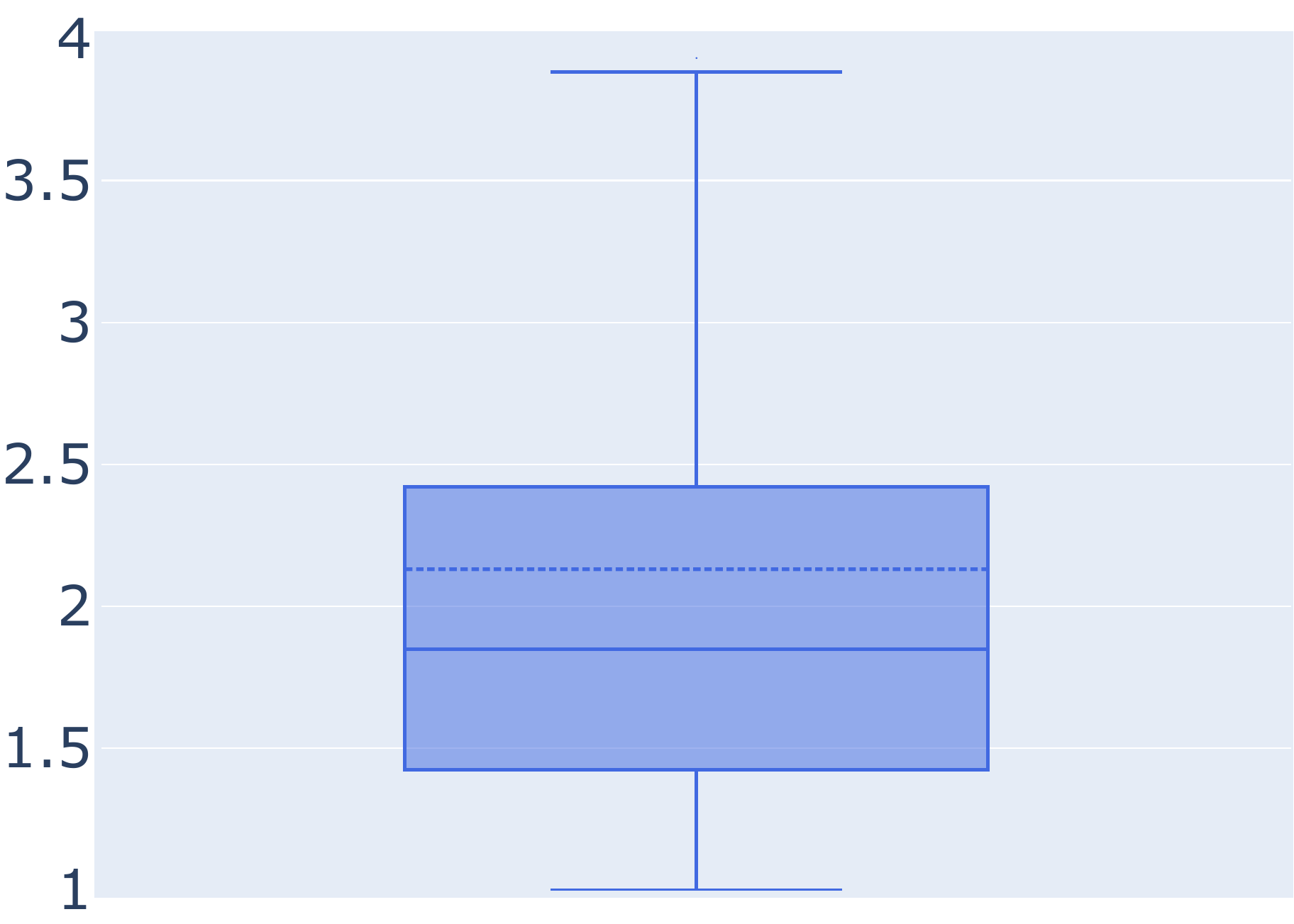}
        \caption{Average cycle weight (ERO).\label{fig:ERO_weights}}
    \end{subfigure}
    \begin{subfigure}[b]{0.49\textwidth}
        \centering
        \includegraphics[scale=0.3]{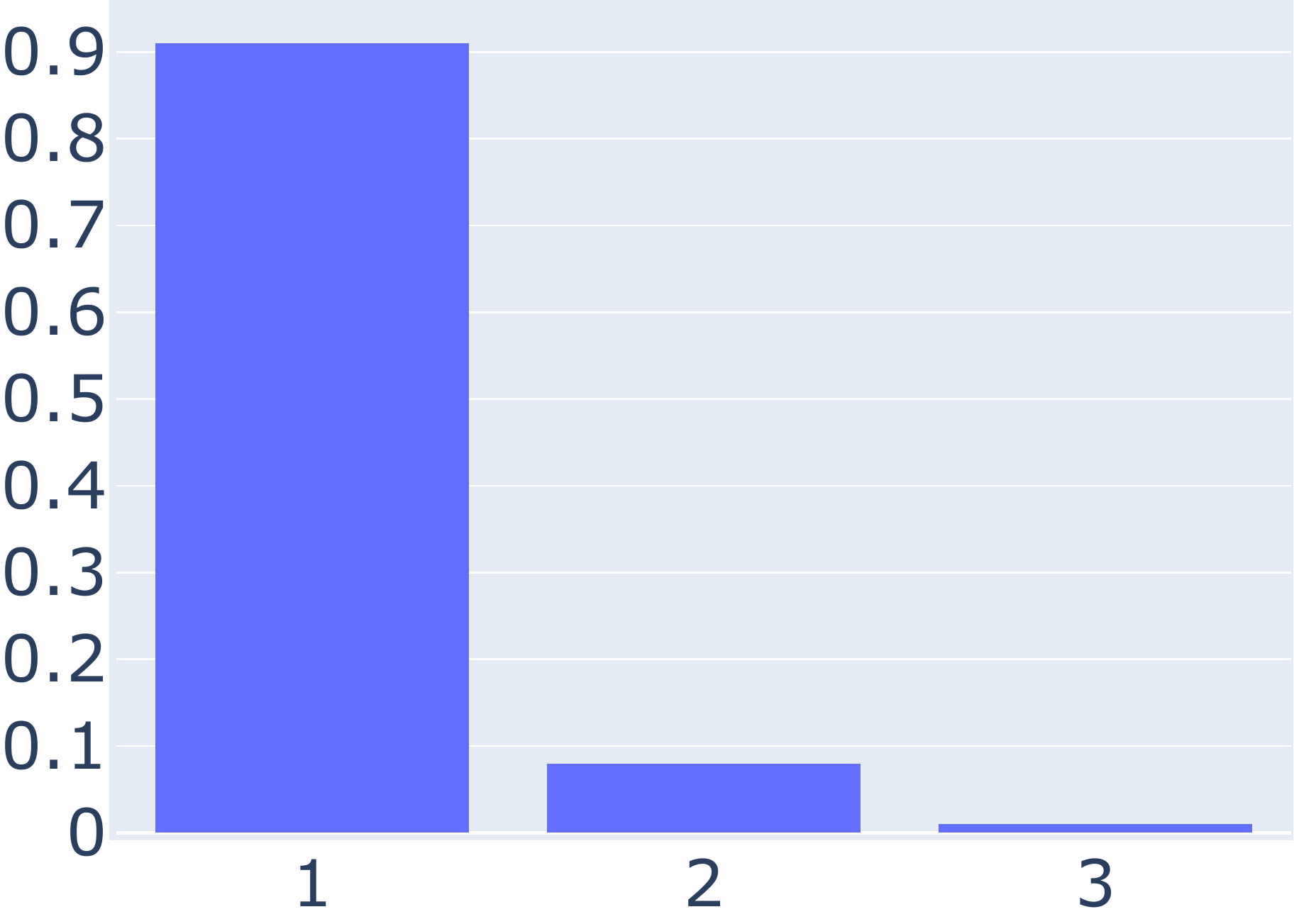}
        \caption{Number of CRTs in a CRSF (MUN). \label{fig:approx_mg_laplacian_MUN_CRTs}}
    \end{subfigure}
    \hfill
    \begin{subfigure}[b]{0.49\textwidth}
        \centering
        \includegraphics[scale=0.3]{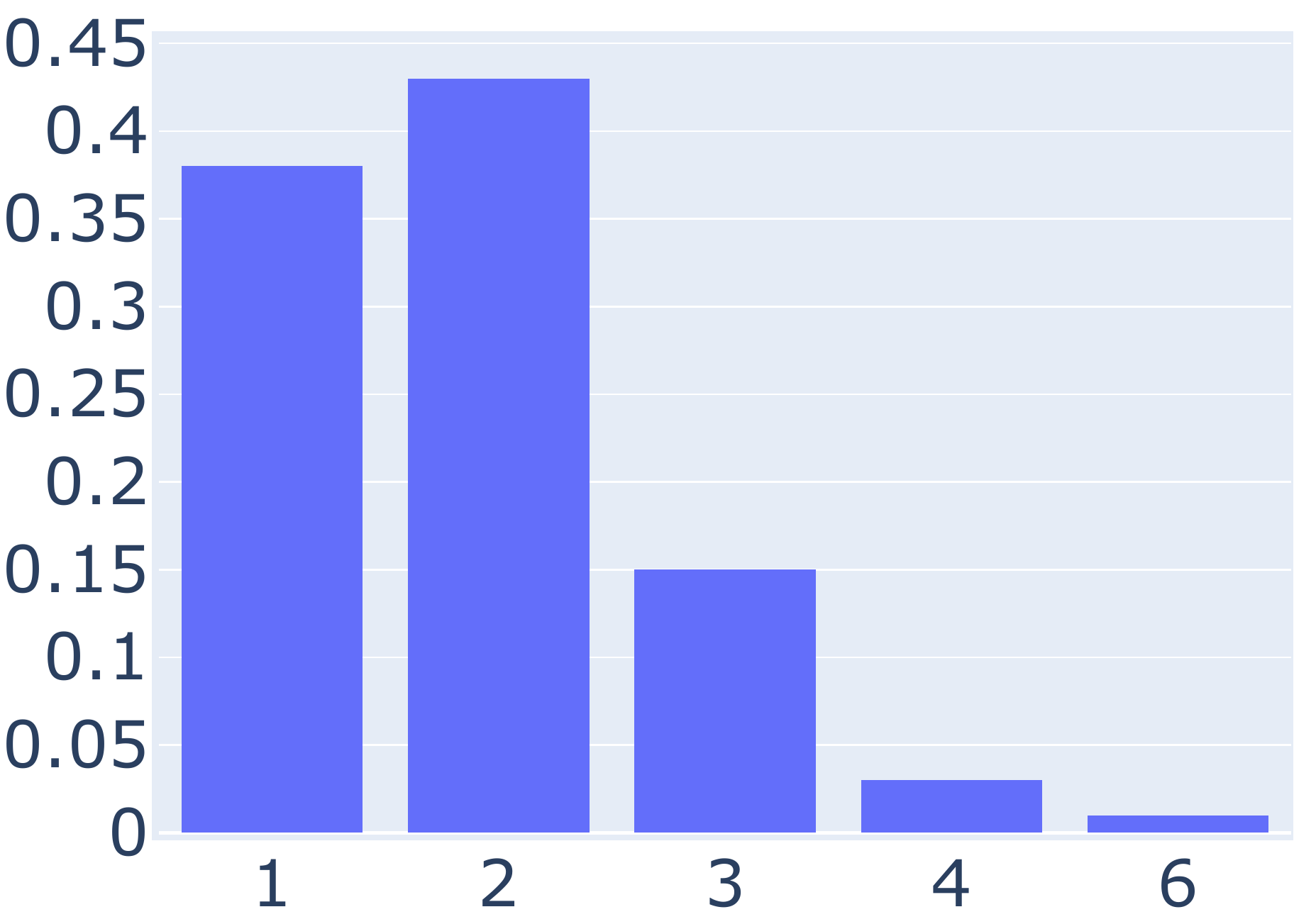}
        \caption{Number of CRTs in a CRSF  (ERO).\label{fig:approx_mg_laplacian_ERO_CRTs}}
    \end{subfigure}
    \caption{On the top row, average cycle weight $
    1 \vee (1 - \cos \theta(\eta))$ for the sparsifiers obtained with self-normalized importance sampling in \cref{fig:approx_mg_laplacian}.
    On the bottom row, distribution of the number of CRTs in one CRSF samples for the setting of \cref{fig:approx_mg_laplacian}. 
    \label{fig:cycle_weights}}
\end{figure}



\subsection{Comparison with algebraic sampling \label{sec:HKPV_extra_results}}
In this section, we report a brief empirical evaluation of the generic DPP sampler, called here $\HKPV$ algorithm \citep*{HKPV06},  for sampling determinantal CRSFs.
\subsubsection{Subgraphs sampled with HKPV}
\begin{figure}[ht]
    \centering
    \begin{subfigure}[b]{0.49\textwidth}
        \centering
        \includegraphics[scale=0.25]{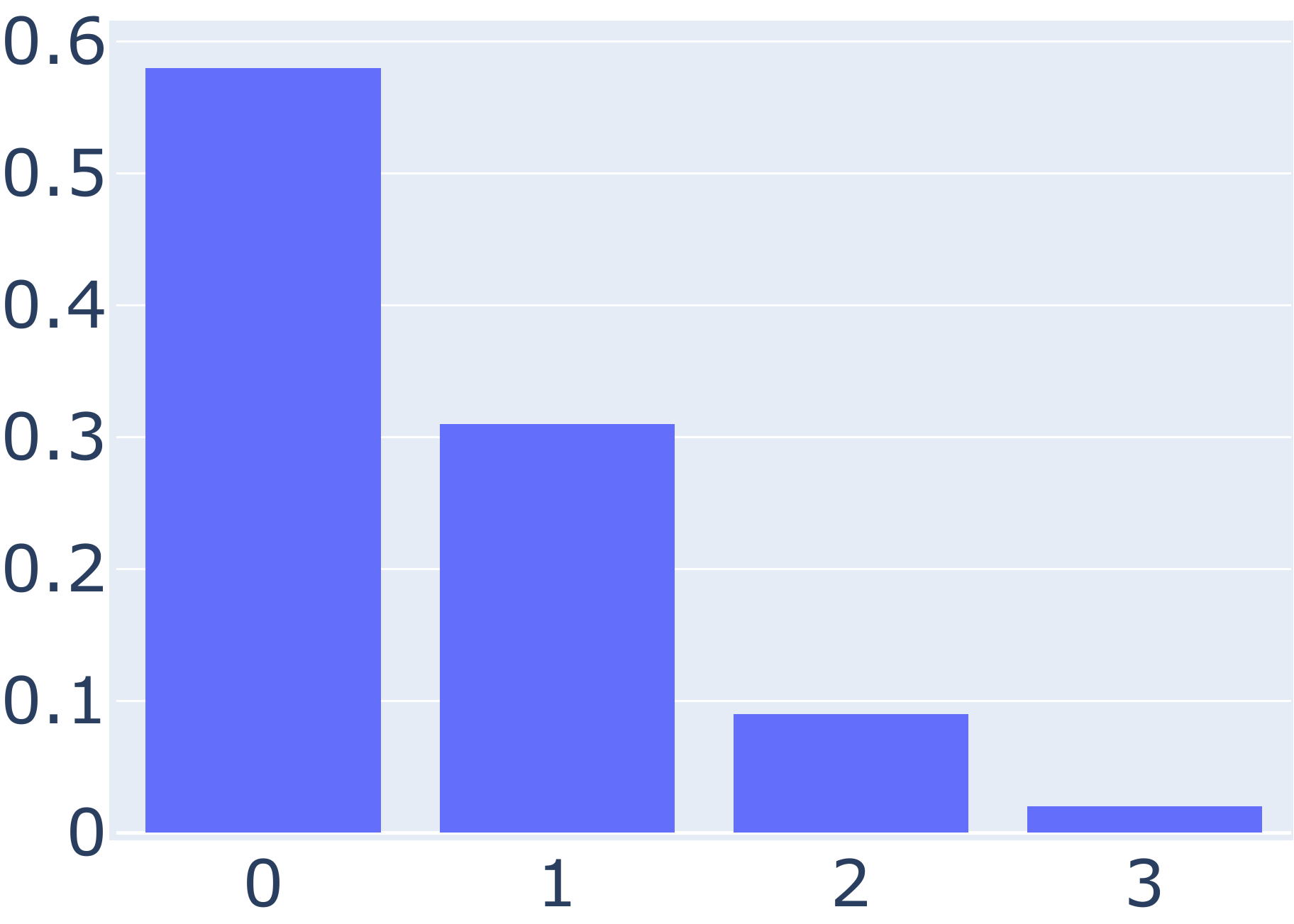}
        \caption{Subgraphs with no cycle ($\eta=0.1$).} 
    \end{subfigure}
    \begin{subfigure}[b]{0.49\textwidth}
        \centering
        \includegraphics[scale=0.25]{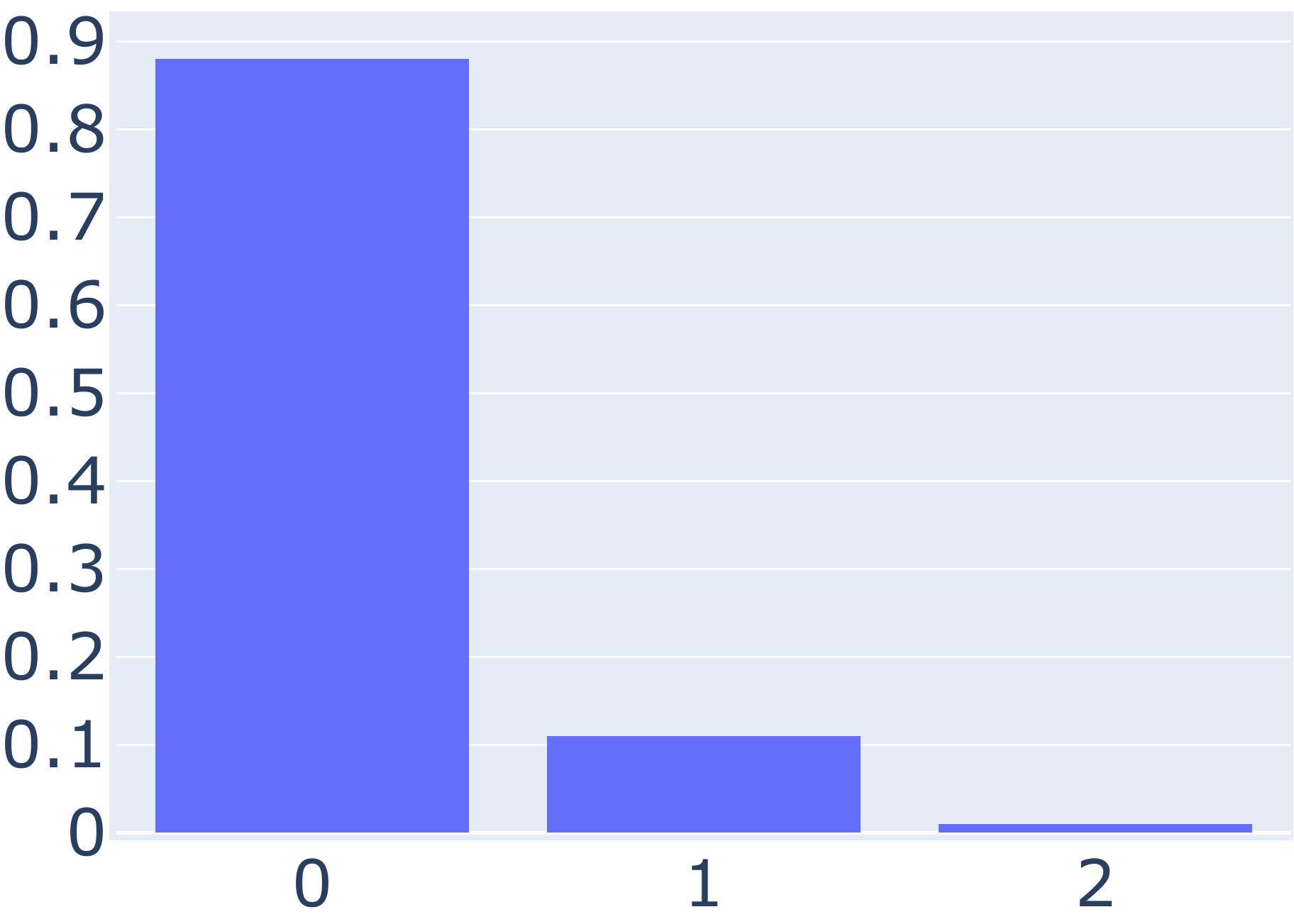}
        \caption{Subgraphs with no cycle ($\eta=0.7$).}
    \end{subfigure}
    \begin{subfigure}[b]{0.49\textwidth}
        \centering
        \includegraphics[scale=0.25]{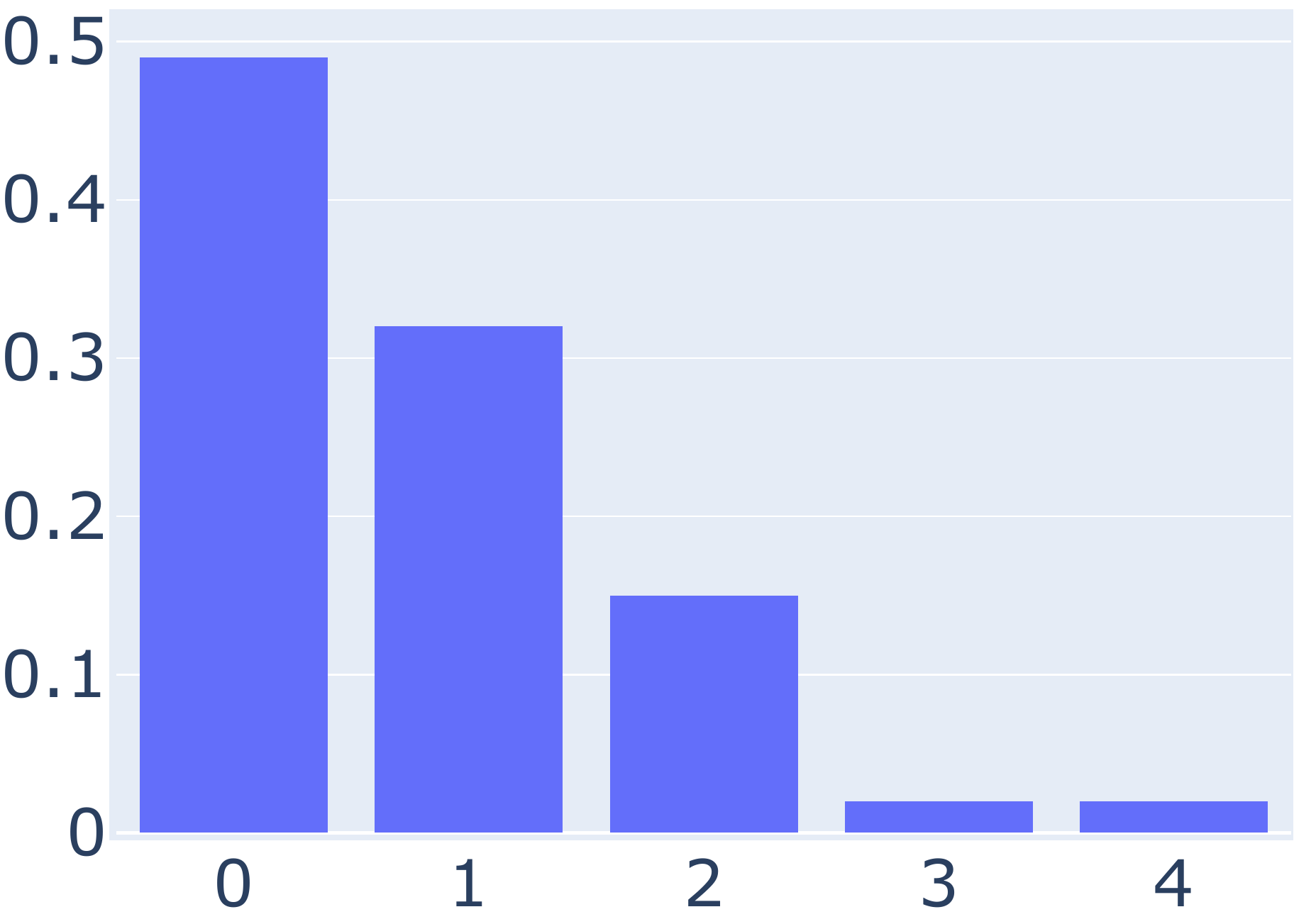}
        \caption{Subgraphs with exactly $1$ cycle ($\eta=0.1$).} 
    \end{subfigure}
    \begin{subfigure}[b]{0.49\textwidth}
        \centering
        \includegraphics[scale=0.25]{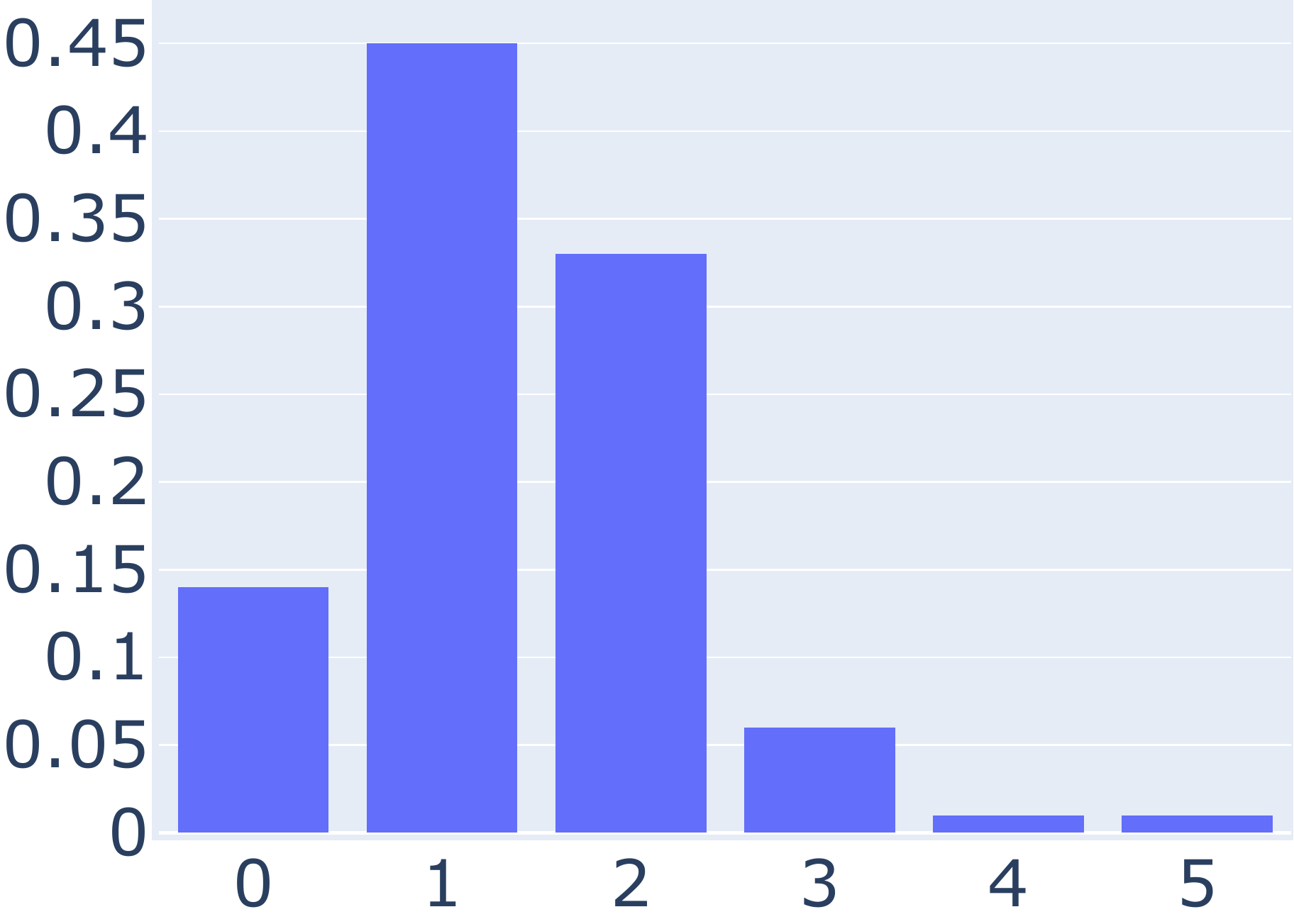}
        \caption{Subgraphs with exactly $1$ cycle ($\eta=0.7$).}
    \end{subfigure}
    \begin{subfigure}[b]{0.49\textwidth}
        \centering
        \includegraphics[scale=0.25]{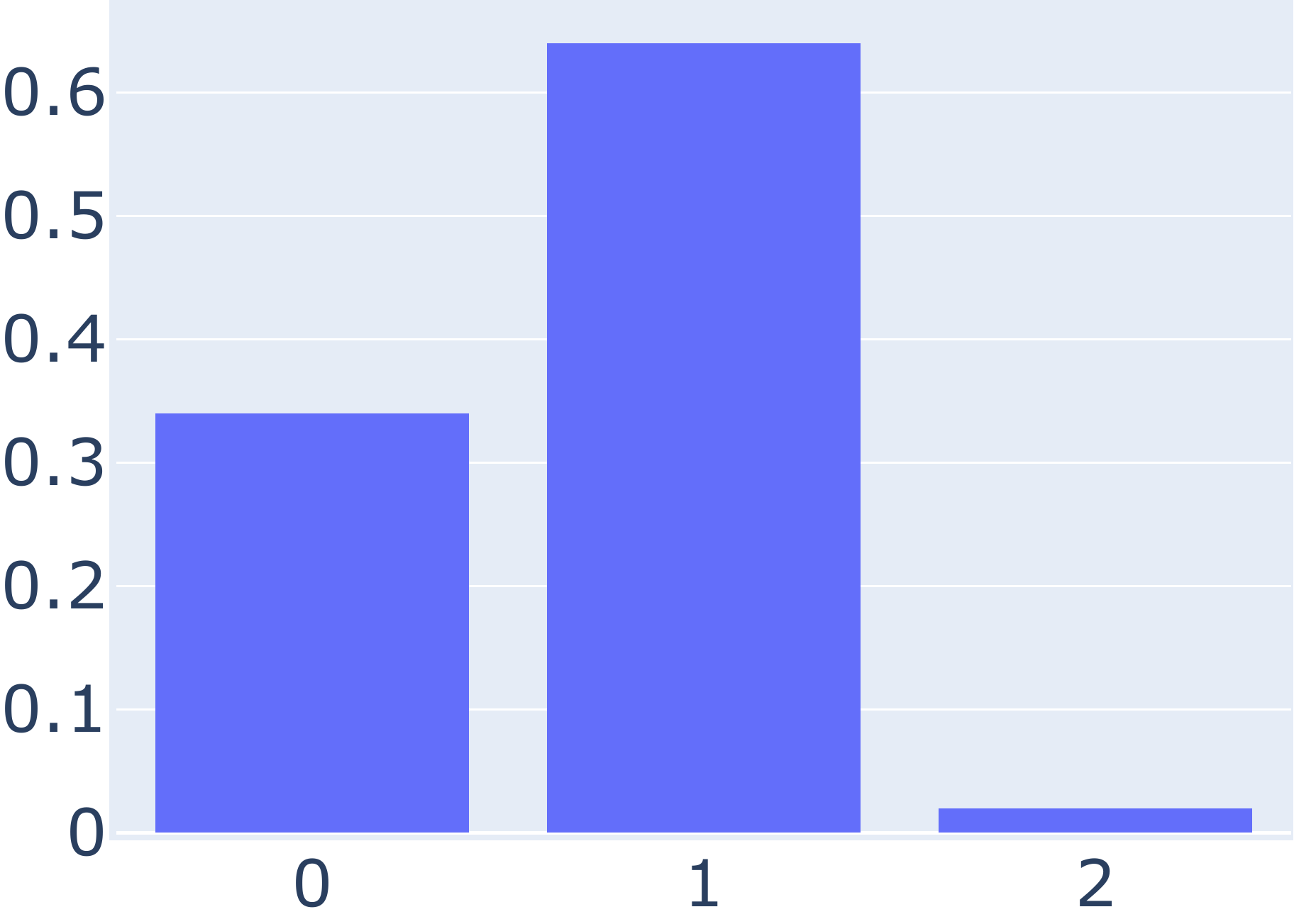}
        \caption{Subgraphs at least $2$ cycles ($\eta=0.1$).} 
    \end{subfigure}
    \begin{subfigure}[b]{0.49\textwidth}
        \centering
        \includegraphics[scale=0.25]{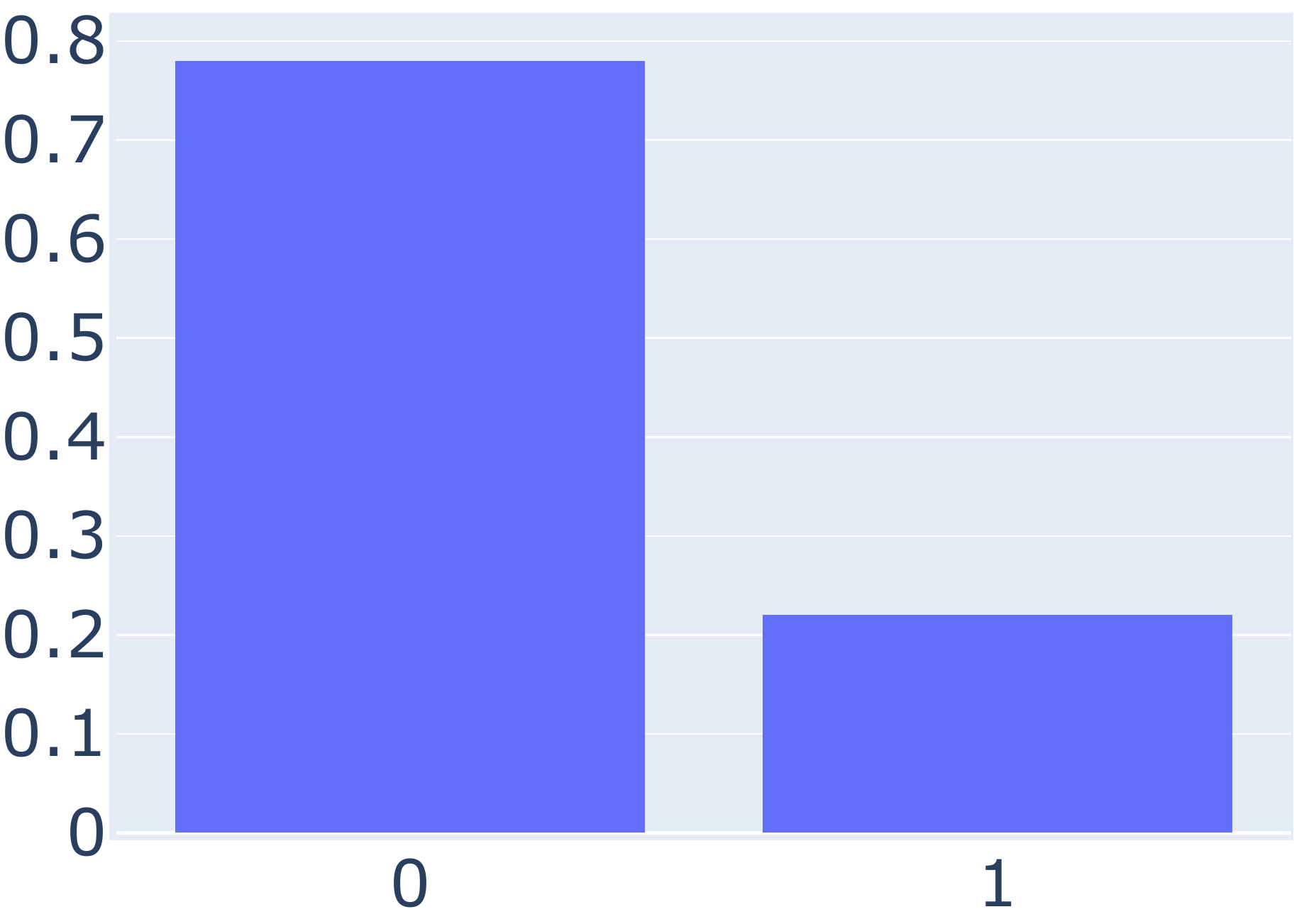}
        \caption{Subgraphs at least $2$ cycles ($\eta=0.7$).}
    \end{subfigure}
    \caption{Histogram of number of connected subgraphs in samples  obtained thanks to HPKV algorithm for $\DPP(K)$ with $K = B\Delta^{-1}B^*$ where $B$ is the incidence matrix of a MUN$(n,p,\eta)$ graph with $n=100$, $p=0.6$  and where  $\eta=0.1$ (Left-hand side) and $\eta=0.7$ (Right-hand side).
    A subgraph in the support of this determinantal measure should have no connected component with no cycle and no connected component with at least $2$ cycles.
     \label{fig:HKPV_subgraphs}}
\end{figure}

{With the help of the $\HKPV$ algorithm, we sample the determinantal measure over CRSFs on a MUN$(n,p,\eta)$ graph with $n=100$, $p=0.6$ and for two values of the noise parameter $\eta$, i.e., $\eta=0.1$ and $\eta=0.7$.
The implementation goes as follows. 
First, a QR decomposition of the $m\times n$ magnetic incidence matrix $B$ is performed, so that the correlation kernel of the DPP associated with the orthogonal projection onto the column space of $B$ reads $K = Q Q^*$. 
QR decomposition is often a useful pre-processing step for DPP sampling \citep{barthelme2023faster}.
Second, a chain rule is applied to the $m\times n$ matrix $Q$ which gives the desired sample; see \citep{HKPV06}. At this stage we use the implementation of \texttt{Determinantal.jl}\footnote{\url{https://github.com/dahtah/Determinantal.jl}}.}

{$\HKPV$ yields then by construction a sample of $n$ edges.
Next, we sample $100$ subgraphs with this algorithm. 
For each subgraph, we count the number of connected components with (i) no cycle or trees (illegitimate), (ii) exactly one cycle or CRT (legitimate) and (iii) at least two cycles (illegitimate).
The corresponding histograms are given in \cref{fig:HKPV_subgraphs}.}

{Recall that a CRSF is defined as a spanning subgraph whose connected components have exactly $1$ cycle.
By inspection of \cref{fig:HKPV_subgraphs}, we observe that several subgraphs sampled in this simulation contain trees or event subgraphs with at least $2$ cycles.
The latter subgraphs cannot in principle appear in a CRSF.
Notwithstanding, the distribution of subgraphs is closer to the target distribution (i.e., associated with a CRSF) for $\eta=0.7$ as compared with $\eta = 0.1$.
We therefore presume that this effect is due to numerical errors in this algebraic algorithm -- since a large value of $\eta$ intuitively makes the columns of the magnetic incidence matrix $B$ less collinear.}

{As a sanity check, we compared the empirical LS distribution $\lev_{\mathrm{emp}}(e)$ obtained with $10^7$ samples of $\HKPV$ -- computed with \cref{eq:emp_LS} -- to the exact LS for a rather small MUN$(n,p,\eta)$ graph with $n = 20$, $p = 0.3$, and $\eta = 0.9$. 
In the boxplot of \cref{fig:LS_est_HKPV}, we observe that the empirical inclusion probabilities of edges obatined with $\HKPV$ are indeed close to the exact inclusion probabilities.
The mean value and standard deviation of the relative differences $(\lev(e) -\lev_{\mathrm{emp}}(e))/ \lev(e)$ for all edges $e$ are respectively $-0.004$ and $0.05$.}

\subsubsection{Time to sample CRSFs with $\HKPV$ and $\cyclepopping{}$\label{sec:comp_time_HKPV_cyclepopping}}

{For completeness, we also report in \cref{fig:timings_HKPV_vs_CyclePopping} boxplots corresponding to the run times for $100$ executions of $\HKPV$ and $\cyclepopping{}$.
The setting is the same as for the simulation reported in \cref{fig:HKPV_subgraphs}.
In this example, we observe that the mean time to sample with $\HKPV$ is about one order of magnitude larger w.r.t.\ $\cyclepopping{}$.} 
\begin{figure}[ht]
    \centering
    \begin{subfigure}[b]{0.49\textwidth}
        \centering
        \includegraphics[scale=0.25]{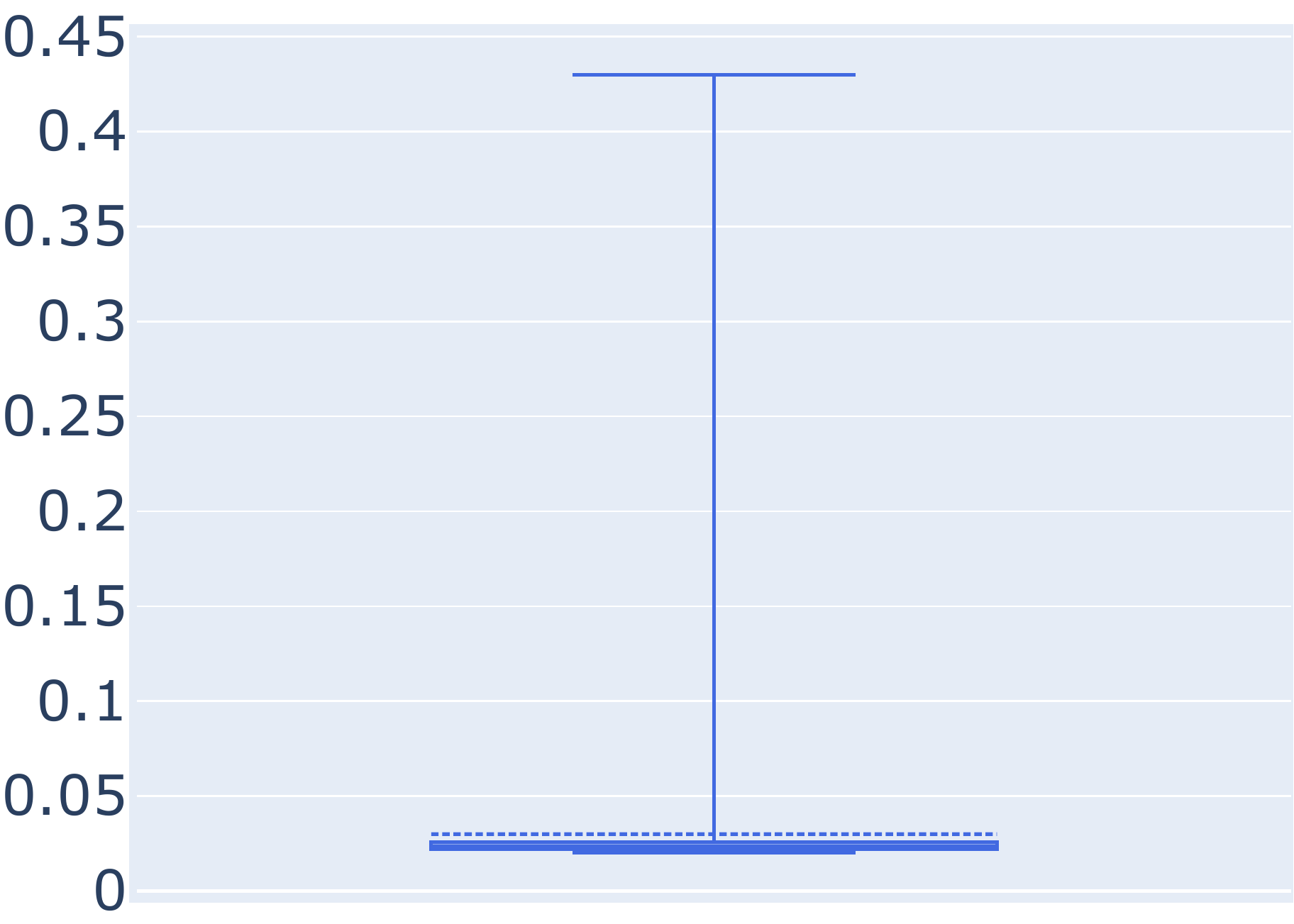}
        \caption{$\HKPV$ ($\eta=0.1$).} 
    \end{subfigure}
    \begin{subfigure}[b]{0.49\textwidth}
        \centering
        \includegraphics[scale=0.25]{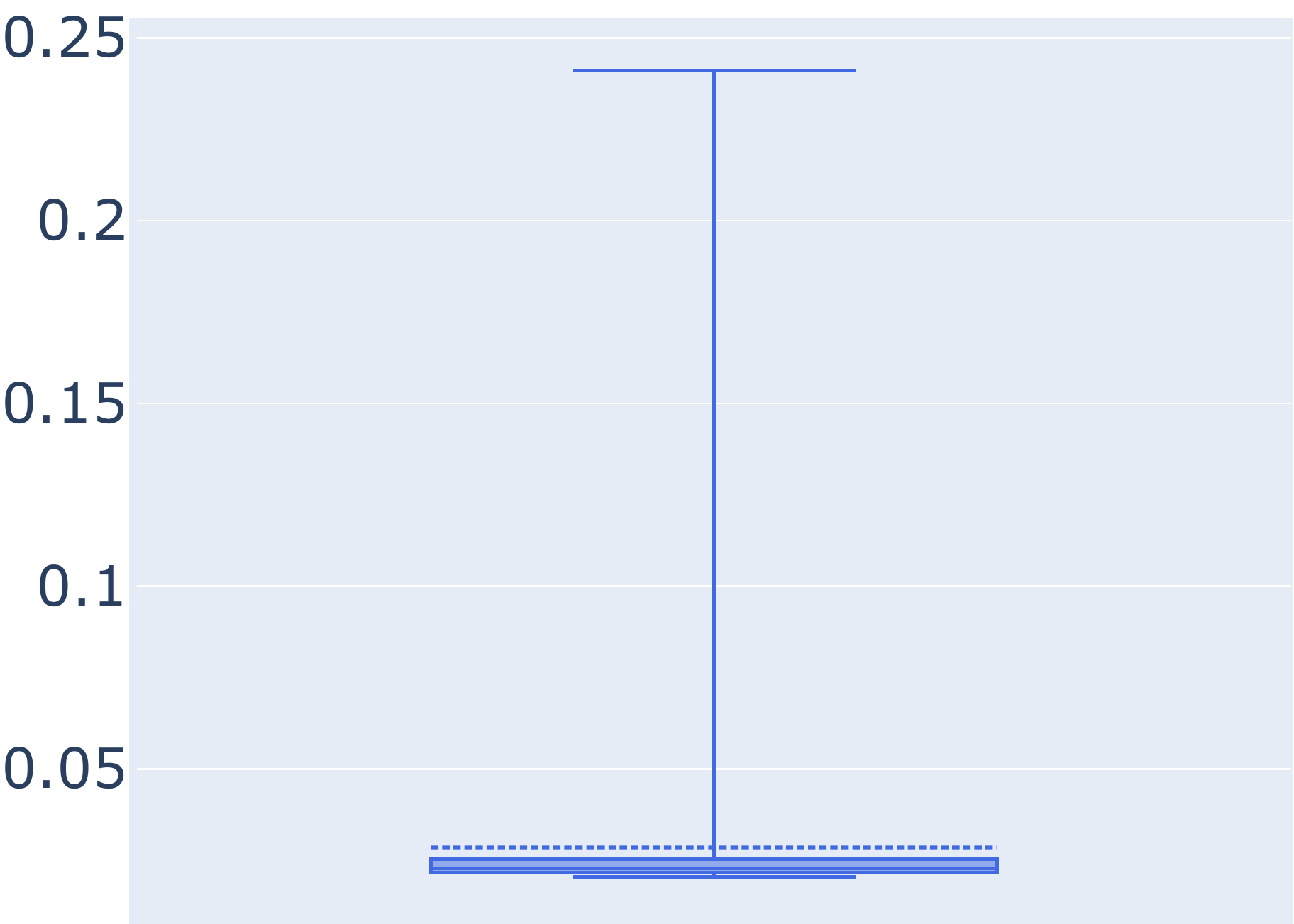}
        \caption{$\HKPV$ ($\eta=0.7$).}
    \end{subfigure}
    \begin{subfigure}[b]{0.49\textwidth}
        \centering
        \includegraphics[scale=0.25]{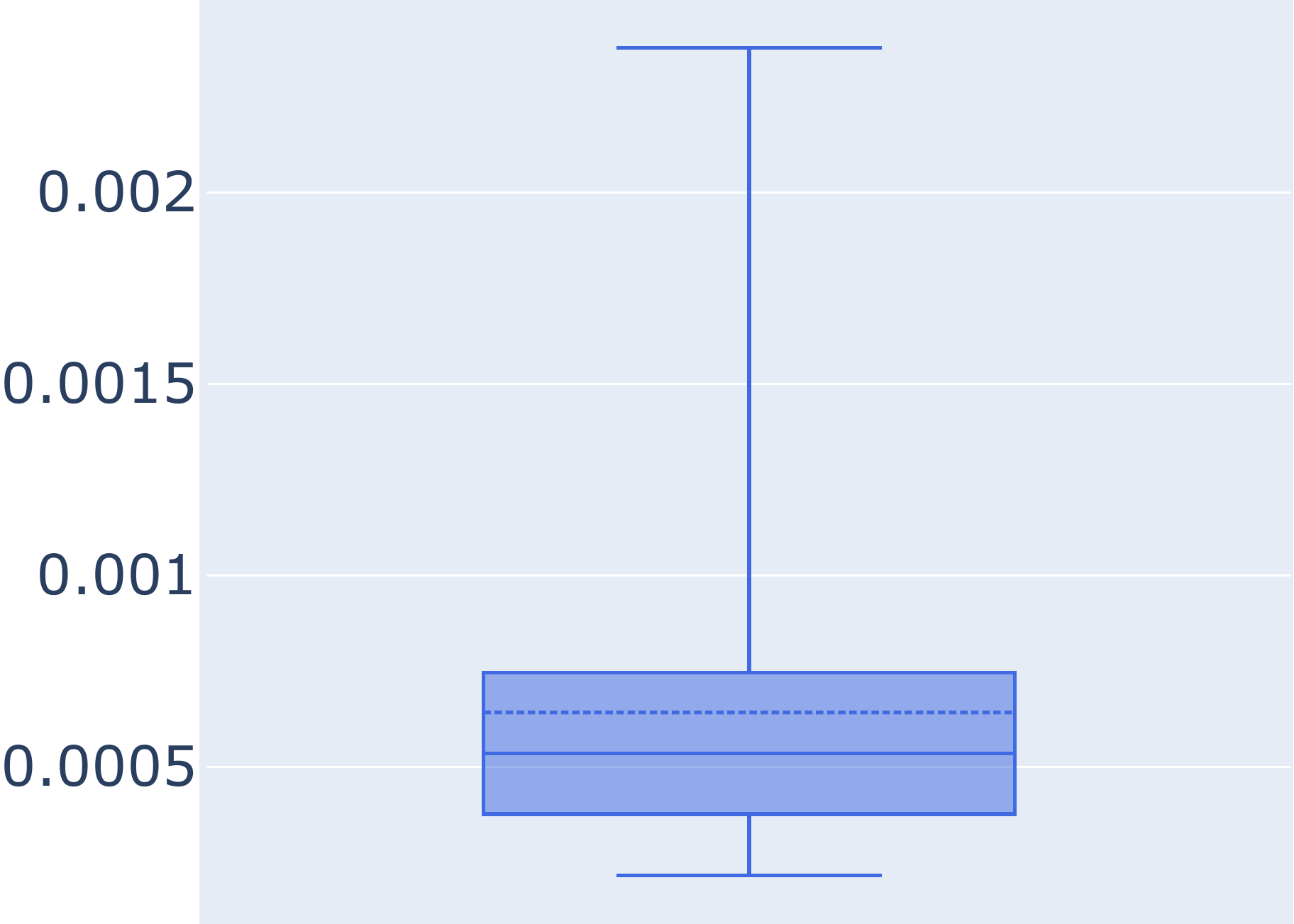}
        \caption{$\cyclepopping{}$ ($\eta=0.1$).} 
    \end{subfigure}
    \begin{subfigure}[b]{0.49\textwidth}
        \centering
        \includegraphics[scale=0.25]{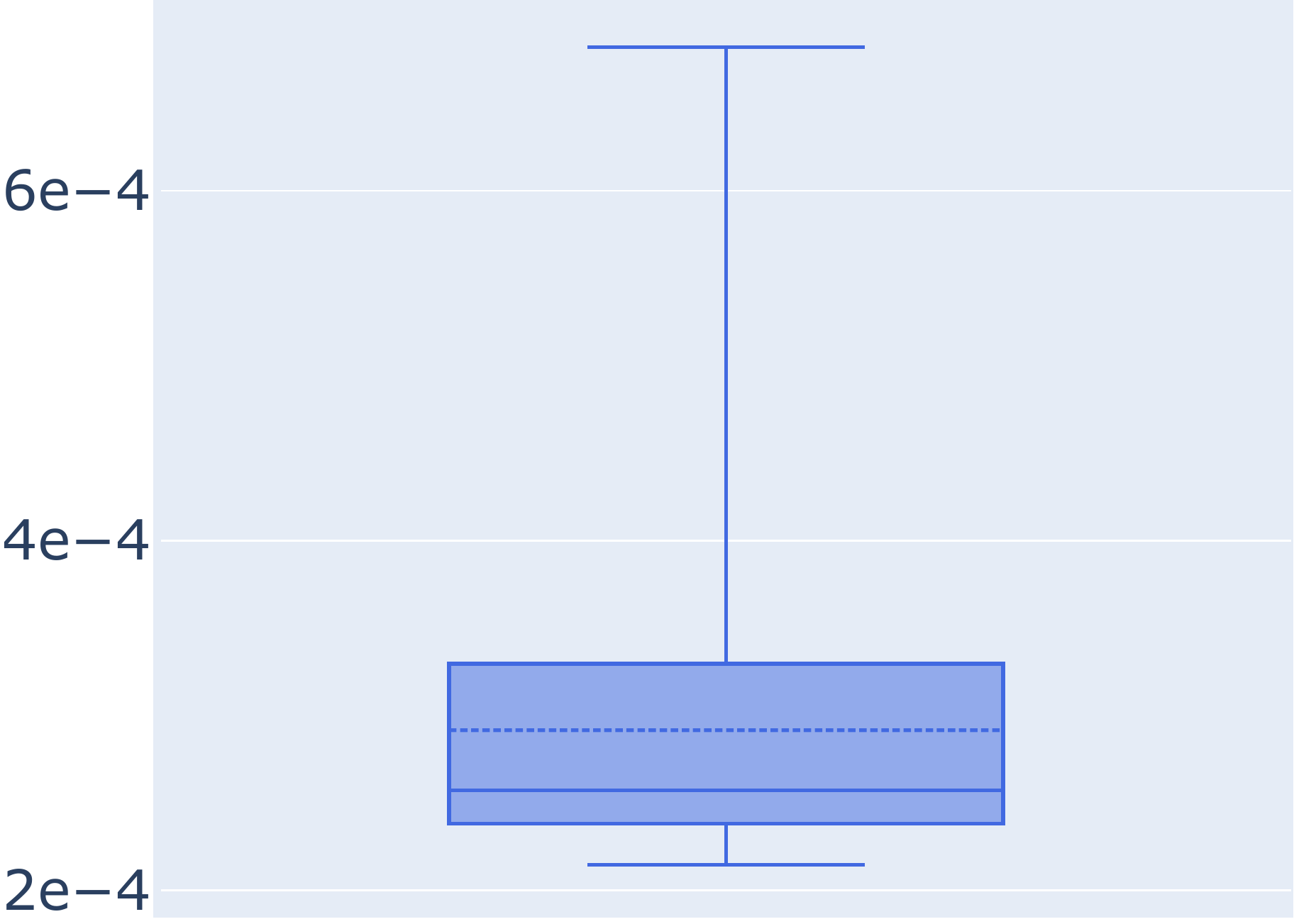}
        \caption{$\cyclepopping{}$ ($\eta=0.7$).}
    \end{subfigure}
    \caption{Timing (s) for sampling a CRSF with $q=0$ in a MUN$(n,p,\eta)$ with $n=100$, $p=0.6$  and where  $\eta=0.1$ (Left-hand side) and $\eta=0.7$ (Right-hand side).
    Boxplots are computed on $100$ independent samples with $\HKPV$ algorithm (top row) and $\cyclepopping{}$ (bottom row) with cycle weight capping. \label{fig:timings_HKPV_vs_CyclePopping}}    
\end{figure}

\begin{figure}
    \centering
        \includegraphics[scale=0.25]{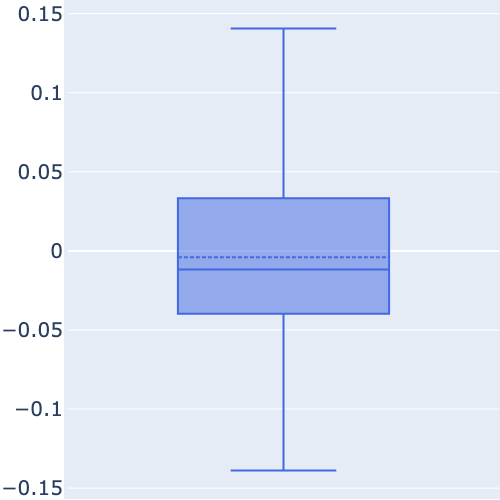}
        \caption{Relative differences between the exact leverage scores and their empirical estimates with $\HKPV$ sampling over $10^7$ Monte Carlo runs for a MUN$(n,p,\eta)$ graph with $n = 20$, $p = 0.3$, and $\eta = 0.9$.
        \label{fig:LS_est_HKPV}}    
\end{figure}
\FloatBarrier


\bibliography{References}

\end{document}